\documentclass[acmsmall, authorversion, nonacm, screen]{acmart}

\usepackage{xfrac}
\usepackage{flushend}
\usepackage{multicol}
\usepackage{enumitem}
\usepackage{tikz}
\usepackage{wrapfig}
\usepackage{worldflags}
\usepackage{multirow}
\usepackage{mathtools}
\AtBeginDocument{%
  \providecommand\BibTeX{{%
    \normalfont B\kern-0.5em{\scshape i\kern-0.25em b}\kern-0.8em\TeX}}}

\usepackage{algorithm}
\usepackage[noend]{algpseudocode}

\usepackage{appendix}
\usepackage{lipsum}
\usepackage{geometry}

\usepackage{amsmath}
\usepackage{mathtools}
\usepackage{amsthm}
\usepackage{xspace}
\usepackage{nicefrac}

\setcopyright{rightsretained}
\acmJournal{POMACS}
\acmYear{2025} \acmVolume{9} \acmNumber{1} \acmArticle{8} \acmMonth{3} \acmPrice{}\acmDOI{10.1145/3711701}

\usepackage[]{algpseudocode}

\newcommand{\sref}[2]{\hyperref[#2]{#1 \ref{#2}}}

\makeatletter
\newcommand{\algmargin}{\the\ALG@thistlm}

\renewcommand\paragraph{\@startsection{paragraph}{4}{\parindent}%
  {2pt}
  {-\parindent}
  {\ACM@NRadjust{\@parfont\@adddotafter}}}
\makeatother
\newlength{\whilewidth}
\algdef{SE}[parWHILE]{parWhile}{EndparWhile}[1]
{\parbox[t]{\dimexpr\linewidth-\algmargin}{%
		\hangindent\whilewidth\strut\algorithmicwhile\ #1\ \algorithmicdo\strut}}{\algorithmicend\ \algorithmicwhile}%
\algnewcommand{\parState}[1]{\State%
	\parbox[t]{\dimexpr\linewidth-\algmargin}{\strut #1\strut}}
\makeatother

\newtheorem{thm}{Theorem}[section]
\newtheorem{lem}[thm]{Lemma}
\newtheorem{pro}[thm]{Proposition}
\newtheorem{cor}[thm]{Corollary}

\newtheorem{dfn}[thm]{Definition}

\newtheorem{ass}[thm]{Assumption}

\usepackage{etoolbox}
\makeatletter
\patchcmd{\hyper@makecurrent}{%
    \ifx\Hy@param\Hy@chapterstring
        \let\Hy@param\Hy@chapapp
    \fi
}{%
    \iftoggle{inappendix}{%
        \@checkappendixparam{chapter}%
        \@checkappendixparam{section}%
        \@checkappendixparam{subsection}%
        \@checkappendixparam{subsubsection}%
        \@checkappendixparam{paragraph}%
        \@checkappendixparam{subparagraph}%
    }{}%
}{}{\errmessage{failed to patch}}

\newcommand*{\@checkappendixparam}[1]{%
    \def\@checkappendixparamtmp{#1}%
    \ifx\Hy@param\@checkappendixparamtmp
        \let\Hy@param\Hy@appendixstring
    \fi
}
\makeatletter

\newtoggle{inappendix}
\togglefalse{inappendix}

\apptocmd{\appendix}{\toggletrue{inappendix}}{}{\errmessage{failed to patch}}
\apptocmd{\subappendices}{\toggletrue{inappendix}}{}{\errmessage{failed to patch}}

\DeclareMathOperator*{\argmin}{arg\,min}

\newcommand{\calT}{\ensuremath{\mathcal{T}}}
\newcommand{\SOCO}{\ensuremath{\mathsf{SOCO}}\xspace}
\newcommand{\OCO}{\ensuremath{\mathsf{OCO}}\xspace}
\newcommand{\OWT}{\ensuremath{\mathsf{OWT}}\xspace}
\newcommand{\MTS}{\ensuremath{\mathsf{MTS}}\xspace}
\newcommand{\CFL}{\ensuremath{\mathsf{CFL}}\xspace}

\newcommand{\ALG}{\textsc{Alg}\xspace}

\newcommand{\ADV}{\textsc{Adv}\xspace}

\newcommand{\OPT}{\textsc{Opt}\xspace}
\newcommand{\CLIPacro}{\textsc{CLIP}\xspace}
\newcommand{\CLIP}{\textsc{SC}\xspace}
\newcommand{\CarbonClipper}{\textsc{ST-CLIP}\xspace}
\newcommand{\OFF}{\ensuremath{\mathtt{OFF}}\xspace}
\newcommand{\ON}{\ensuremath{\mathtt{ON}}\xspace}
\newcommand{\CASLB}{\ensuremath{\mathsf{SOAD}}\xspace}
\newcommand{\CSWM}{\ensuremath{\mathsf{SOAD}}\xspace}
\newcommand{\CSWMt}{\ensuremath{\mathsf{SOAD\text{-}T}}\xspace}
\newcommand{\SOAD}{\ensuremath{\mathsf{SOAD}}\xspace}
\newcommand{\SOADt}{\ensuremath{\mathsf{SOAD\text{-}T}}\xspace}
\newcommand{\PCM}{\textsc{PCM}\xspace}
\newcommand{\norm}[1]{\ensuremath{\Vert #1 \Vert_{\ell_1(\mathbf{w})}}}
\newcommand{\mbf}[1]{{\mathbf{#1}}}
\newcommand{\ex}{\mathbb{E}}
\newcommand{\diam}{\text{diam}}

\colorlet{blue}{black}
\colorlet{teal}{black}

\begin{document}

\title{Learning-Augmented Competitive Algorithms for Spatiotemporal Online Allocation with Deadline Constraints}

\author{Adam Lechowicz}
\affiliation{%
  \institution{University of Massachusetts Amherst}
  \country{USA}
}
\email{alechowicz@cs.umass.edu}

\author{Nicolas Christianson}
\affiliation{%
  \institution{California Institute of Technology}
  \country{USA}
}
\email{nchristianson@caltech.edu}

\author{Bo Sun}
\affiliation{%
  \institution{University of Waterloo}
  \country{Canada}
}
\email{bo.sun@uwaterloo.ca}

\author{Noman Bashir}
\affiliation{%
  \institution{Massachusetts Institute of Technology}
  \country{USA}
}
\email{nbashir@mit.edu}

\author{Mohammad Hajiesmaili}
\affiliation{%
  \institution{University of Massachusetts Amherst}
  \country{USA}
}
\email{hajiesmaili@cs.umass.edu}

\author{Adam Wierman}
\affiliation{%
  \institution{California Institute of Technology}
  \country{USA}
}
\email{adamw@caltech.edu}

\author{Prashant Shenoy}
\affiliation{%
  \institution{University of Massachusetts Amherst}
  \country{USA}
}
\email{shenoy@cs.umass.edu}

\renewcommand{\shortauthors}{Adam Lechowicz et al.}

\begin{abstract}

We introduce and study spatiotemporal online allocation with deadline constraints ($\mathsf{SOAD}$), a new online problem motivated by emerging challenges in sustainability and energy.  In $\mathsf{SOAD}$, an online player completes a workload by allocating and scheduling it on the points of a metric space $(X, d)$ while subject to a deadline $T$.  At each time step, a service cost function is revealed that represents the cost of servicing the workload at each point, and the player must irrevocably decide the current allocation of work to points.  Whenever the player moves this allocation, they incur a movement cost defined by the distance metric $d(\cdot, \ \cdot)$ that captures, e.g., an overhead cost.  $\mathsf{SOAD}$ formalizes the open problem of combining general metrics and deadline constraints in the online algorithms literature, unifying problems such as metrical task systems and online search.  We propose a competitive algorithm for $\mathsf{SOAD}$ along with a matching lower bound establishing its optimality.  Our main algorithm, \textsc{ST-CLIP}, is a learning-augmented algorithm that takes advantage of predictions (e.g., forecasts of relevant costs) and achieves an optimal consistency-robustness trade-off.  We evaluate our proposed algorithms in a simulated case study of carbon-aware spatiotemporal workload management, an application in sustainable computing that schedules a delay-tolerant batch compute job on a distributed network of data centers.  In these experiments, we show that \textsc{ST-CLIP} substantially improves on heuristic baseline methods.

\end{abstract}

\maketitle

\section{Introduction}
\label{sec:intro}

{\color{blue}
We introduce and study spatiotemporal online allocation with deadline constraints (\SOAD), an online optimization problem motivated by emerging challenges in sustainability.
In \SOAD, an online player aims to service a workload by allocating and scheduling it on one of $n$ points %
represented by a metric space $(X, d)$. They pay a service cost at a point if the workload is currently being serviced there, a spatial movement cost defined by the metric whenever they change the allocation between points, and a temporal switching cost when bringing the workload into or out of service at a single point.  
The workload arrives with a \textit{deadline constraint} $T$ that gives the player some \textit{slack}, i.e., the workload can be \textit{paused} for some time to avoid high cost periods without violating the constraint. %
}

\SOAD builds on a long history of related problems in online algorithms.
In particular, two lines of work share specific features in common with our setting.  One line of work focuses on metrical task systems (\MTS) and smoothed online convex optimization (\SOCO), where problems consider online optimization with movement costs over general metrics, but do not accommodate long-term constraints, such as deadlines~\cite{Borodin:92, Koutsoupias:09, chenSmoothedOnlineConvex2018a, bubeckChasingNestedConvex2019, Bubeck:21MTSTrees, Bansal:22, bubeckRandomizedServerConjecture2022}. A complementary line of work is that of online search problems with long-term constraints, such as one-way trading (\OWT) and online knapsack -- these problems enforce that a player's cumulative decisions satisfy a constraint over the entire input, but do not consider general metric (decision) spaces or movement costs~\cite{ElYaniv:01, Lorenz:08, Mohr:14, SunZeynali:20}.
\SOAD extends both \MTS/\SOCO-type problems and \OWT-type problems by simultaneously considering general metric movement/switching costs and deadline (i.e., long-term) constraints.

{\color{blue}
For many applications, the underlying problem to be solved often requires a model with \textit{both} smoothed optimization (i.e., movement costs) \textit{and} deadline constraints.  Furthermore, for an application such as carbon-aware workload management in data centers, where the spatial movement cost corresponds to, e.g., network delays (see \autoref{sec:examples}), it is necessary to consider a general metric space, since pairwise network latencies do not necessarily correspond to simple distances such as Euclidean (geographic) distances.
The question of whether it is possible to design competitive online algorithms in this combined setting has remained open for over a decade, with theoretical progress emerging only in the last few years in special cases such as the unidimensional setting, in $\ell_1$ vector spaces, or with different performance metrics such as regret~\cite{Jenatton:16, Badanidiyuru:18, Lechowicz:23, Lechowicz:24, Bostandoost:24, Lechowicz:24CFL}. 
This work seeks to close this gap by answering the following question:
}
\begin{center} 
{\color{blue}
\textit{Is it possible to design online algorithms for \SOAD 
that manage the challenges of \textbf{general metrics} and provide competitive guarantees
without violating the \textbf{deadline constraint}?
}
}
\end{center}

\noindent It is well known that problems related to \SOAD, such as \MTS and \OWT, are difficult in the sense that their competitive ratios scale in the size of the decision space or the ratio between maximal and minimal prices. However, these pessimistic lower bounds hold in the worst case, while in practice a decision-maker can often leverage data-driven machine learning approaches to obtain algorithms that perform better empirically. Recent work in the online algorithms literature has leveraged the paradigm of learning-augmented algorithms~\cite{Lykouris:18, Purohit:18} to design and analyze algorithms that can take advantage of patterns in the input via untrusted ``advice'' (e.g., predictions from a machine learning model) without losing adversarial competitive bounds.  Such learning-augmented algorithms have been designed for precursor problems to \SOAD, including \MTS and \OWT~\cite{Antoniadis:20MTS, Christianson:22, Christianson:23MTS, SunLee:21}. 
{\color{blue}
In the \SOAD setting, supported by the availability of practical predictions for our motivating applications and a lack of learning-augmented algorithmic strategies that accommodate both general metrics and deadline constraints, we additionally consider the question:
}
\begin{center} 
{\color{blue}
\textit{ Can we design algorithms for \SOAD that integrate \textbf{untrusted advice} (such as machine-learned predictions) to further improve performance without losing worst-case guarantees?  }
}
\end{center}

\subsection{Related work}

Our results address a long-standing open problem of combining online optimization with general switching costs ($\mathsf{MTS/SOCO}$) and deadline constraints.  Although general \MTS is a famously well-studied problem in online algorithms~\cite{Borodin:92, Blum:97, Bartal:97, Bansal:10, Coester:19, Bansal:22}, it has not been studied under the general form of long-term intertemporal constraints that we consider in the \SOAD formulation.  Amongst the two lines of related work, prior work has focused on either designing \MTS-style algorithms using techniques such as mirror descent~\cite{Bubeck:21MTSTrees}, primal-dual optimization~\cite{Bansal:10}, and work functions~\cite{Abernethy:10}, or \OWT-style algorithms using techniques such as threshold-based algorithm design~\cite{Lorenz:08, Zhou:08}, pseudo-reward maximization~\cite{SunZeynali:20, Yang2021Competitive}, and protection-level policies~\cite{ElYaniv:01}.  As these distinct techniques have been tailored to their respective problem settings, there has been almost no cross-pollination between \MTS-type and \OWT-type problems until this work.
The combination of smoothed optimization and long-term constraints has drawn recent attention in the paradigm of regret analysis in problems such as bandits with knapsacks and \OCO with long-term constraints~\cite{Jenatton:16, Badanidiyuru:18}.
However, despite established connections between \MTS and online learning~\cite{Blum:97}, the problem of optimal \textit{competitive} algorithm design in general settings of this form has yet to be explored.

A select few works~\cite{Lechowicz:23, Lechowicz:24, Bostandoost:24, Lechowicz:24CFL} consider both \textit{switching costs} and deadline constraints in a competitive regime, although they are restricted to special cases such as unidimensional decisions or $\ell_1$ metrics.  
{\color{blue}
Due to these assumptions, their results and algorithms fail to capture the general problem that we consider. As just one example, all of these works assume that the switching cost is only temporal, in the sense that the online player pays the same cost whether they are switching into or out of a state that makes progress towards the deadline constraint.  
This assumption is overly restrictive because it cannot accommodate switching cost situations that may arise in motivating applications, e.g., the case where the player chooses to move between points of the metric while simultaneously switching into an \ON state.
}

Our work also contributes to the field of learning-augmented algorithms, 
designed to bridge the performance of untrusted advice and worst-case competitive guarantees \cite{Lykouris:18, Purohit:18}.
{\color{blue}
Learning-augmented design has been studied in many online problems including ski rental \cite{Wei:20}, bipartite matching \cite{jinOnlineBipartiteMatching2022}, and several related problems including \MTS/\SOCO and \OWT~\cite{Antoniadis:20MTS, Christianson:22, SunLee:21, Lee:24, Lechowicz:24, Lechowicz:24CFL}.
}
For \MTS/\SOCO, a dominant algorithmic paradigm is to adaptively combine the actions of a robust decision-maker and those of e.g., a machine-learning model~\cite{Antoniadis:20MTS, Christianson:22}; optimal trade-offs between robustness and consistency have also been shown in the case of general metrics~\cite{Christianson:23MTS}.  For \OWT and $k$-search, several works have given threshold-based algorithms incorporating predictions that are likewise shown to achieve an optimal robustness-consistency trade-off~\cite{SunLee:21, Lee:24}. 
The advice models for these two tracks of literature are quite different and lead to substantially different algorithms -- namely, online search problems typically assume that the algorithm receives a prediction of, e.g., the \textit{best price}, while \MTS/\SOCO typically consider \textit{black-box advice} predicting the optimal decision at each time step.  Amongst the limited prior literature that considers learning-augmentation in problems with switching costs and deadline constraints~\cite{Bostandoost:24, Lechowicz:24, Lechowicz:24CFL}, both advice models have been considered, underscoring the challenge of %
the \SOAD setting, where the optimal choice of advice model (and corresponding design techniques) is not obvious \emph{a priori}.

\subsection{Contributions}

{\color{blue}
Our main technical contributions make progress on a longstanding problem in online optimization that models emerging practical problems in areas such as sustainability.
Our algorithms and lower bounds for \SOAD  are the first results to consider competitive analysis for deadline-constrained problems on general metrics.  
We obtain positive results for both of the questions posed above under assumptions informed by practice.  In particular, we provide the first competitive algorithm, \PCM (\textbf{P}seudo-\textbf{C}ost \textbf{M}inimization, see \sref{Algorithm}{alg:pcm}), %
}
for this type of problem in \autoref{sec:pcm}, and show that it achieves the best possible competitive ratio up to $\log$ factors that result from the generality of the metric.  Surprisingly, the competitive upper bound we prove for \SOAD (\autoref{thm:etaCompPCM}) compares favorably against known strong lower bounds for precursor problems such as \MTS and \OWT; which suggests an insight that additional structure imposed by constraints can actually facilitate competitive decision making, despite the added complexity of the general metric and deadline constraint.  
To achieve this result, we develop theoretical tools from both lines of related work that help us tackle challenging components of \SOAD.  
For instance, we leverage randomized metric embeddings and optimal transport to endow a general metric $(X, d)$ with a structure that facilitates analysis.  
From the online search literature, we leverage techniques that balance the trade-off between cost and constraint satisfaction, specifically generalizing these ideas to operate in the more complex metric setting necessitated by \SOAD.

{\color{blue}
In \autoref{sec:clip}, we introduce our main learning-augmented algorithm, \CarbonClipper, which integrates black-box decision advice based on, e.g., machine-learned predictions to significantly improve performance without losing worst-case competitive bounds.
}
In our approach, we first prove an impossibility result on the robustness-consistency trade-off for any algorithm.  Using an adaptive optimization-based framework first proposed by~\cite{Lechowicz:24CFL}, we design an algorithm that combines the theoretical tools underpinning \PCM in concert with a constraint hedging against worst-case service costs and movement costs that threaten the desired consistency bound.  This ensures \CarbonClipper attains the optimal trade-off (up to $\log$ factors) in general metric spaces.

{\color{blue}
In \autoref{sec:timevarying}, motivated by real-world applications where the movement cost between points of the metric may not be constant,
we present a generalization of \SOAD where distances are allowed to be \textit{time-varying} and show that our algorithms extend to this case.
Finally, in \autoref{sec:exp}, we evaluate our algorithms in a case study of carbon-aware spatiotemporal workload management (see \autoref{sec:examples}) on a simulated global network of data centers.  We show that \CarbonClipper is able to leverage \textit{imperfect} advice and significantly improve on heuristic baselines for the problem.
}

\section{Problem Formulation, Motivating Applications, Challenges, and Preliminaries} 
\label{sec:problem}
{\color{blue}
In this section, we introduce the spatiotemporal online allocation with deadline constraints (\SOAD) problem and provide motivating applications as examples.  We also discuss some intrinsic challenges in \SOAD that prevent the direct application of existing techniques, and introduce relevant preliminaries from related work.
}

\subsection{Spatiotemporal online allocation with deadline constraints (\SOAD)}\label{sec:prob}

\noindent\textbf{Problem statement.}
{\color{blue}
Consider a decision-maker that manages $n$ points %
defined on a metric space $(X,d)$, where $X$ denotes the set of points and $d(u, \ v)$ denotes the distance %
between any two points $u,v \in X$.
In a time-slotted system, the player aims to complete a unit-size workload before a deadline $T$ while minimizing the total service cost by allocating the workload across points and time. 
}

\textit{Allocation definition. \ }
{\color{blue}
The decision-maker specifies a spatial allocation to one of the points $u \in X$.  At the chosen point, they also make a temporal allocation that \textit{fractionally} divides the allocation between two states, $\ON^{(u)}$ and $\OFF^{(u)}$, where the allocation to $\ON^{(u)}$ represents the amount of resources actively servicing the workload.
Let $\mbf{x}_t := \{ {x}_t^{\ON^{(u)}}, x_t^{\OFF^{(u)}} \}_{u\in X}$ denote the allocation decision at time $t$ across all points and states, %
where ${x}_t^{\ON^{(u)}}$ and $x_t^{\OFF^{(u)}}$ denote the allocation to the \ON and \OFF states at point $u$, respectively.
The feasible set for this vector allocation is given by $\mathcal{X} :=\{ \mbf{x}\subseteq [0,1]^{2n}: {x}^{\ON^{(u)}} + x^{\OFF^{(u)}} \in \{0,1\}, \forall u\in X, \Vert \mbf{x}\Vert_1 = 1\}$. %
}

\textit{Deadline constraint. \ } Let $c(\mbf{x}_t): \mathcal{X} \rightarrow [0,1]$ denote a constraint function that is known to the decision-maker and models the fraction of the workload completed by an allocation $\mbf{x}_t$. 
{\color{blue}
Specifically, we let $c(\mbf{x}_t) = \sum_{u \in X} c^{(u)} \cdot {x}_t^{\ON^{(u)}}$, where ${c}^{(u)}$ is a positive \textit{throughput constant} that encodes how much of the workload is completed 
during one time slot with a full allocation to the state $\ON^{(u)}$.
}
Across the entire time horizon, the decision-maker is subject to a \textit{deadline constraint} stipulating that the cumulative allocations must satisfy $\sum^T_{t=1} c(\mbf{x}_t) \ge 1$. This encodes the requirement that sufficient allocations must be assigned to the $\ON$ states to finish a unit-size workload by the deadline $T$. 
{\color{blue}
The unit size is w.l.o.g. by scaling $c(\cdot)$ appropriately -- e.g., if a workload doubles in size, $c(\cdot)$ is scaled by a factor of $\nicefrac{1}{2}$ to reflect this.
}

\textit{Service and switching costs. \ } %
{\color{blue}
At each time $t$, the cost of allocation $\mbf{x}_t$ consists of a \textit{service cost} ${f_t(\mbf{x}_t) = \sum_{u \in X} f_t^{(u)} \cdot {x}_t^{\ON^{(u)}}}$ for allocations to \ON states, where $f_t^{(u)}$ represents the service cost of point $u$ at time $t$; and a \textit{switching cost} $g(\mbf{x}_t, \mbf{x}_{t-1})$ that includes a spatial movement cost of moving the allocation between points and a temporal switching cost incurred between \ON and \OFF states within one point.
Specifically, whenever the decision-maker changes the allocation across points, they pay a movement cost $d(u_{t-1}, u_{t})$, where $u_{t'} = \{u\in X: {x}^{\ON^{(u)}}_{t'} + x^{\OFF^{(u)}}_{t'} = 1\}$ is the location of the allocation at time $t'$ (in \autoref{sec:timevarying}, we give a generalization where $d(\cdot, \ \cdot)$ also varies with time).
}
{\color{blue}
Within each point, the decision-maker pays a switching cost $\Vert \mbf{x}_t - \mbf{x}_{t-1}\Vert_{\ell_1(\boldsymbol{\beta})} = \sum_{u\in X} \beta^{(u)} |{x}^{\ON^{(u)}}_{t} - {x}^{\ON^{(u)}}_{t-1}|$, where $\beta^{(u)}$ is the switching overhead factor when changing the \ON \ / \OFF allocation %
at point $u$. The overall switching cost is $g(\mbf{x}_t, \mbf{x}_{t-1}) \coloneqq d(u_{t-1}, u_{t}) + \Vert \mbf{x}_t - \mbf{x}_{t-1}\Vert_{\ell_1(\boldsymbol{\beta})}$, and $g$ is known in advance.
}
The decision-maker starts (at $t=0$) with a full allocation at some $\OFF$ state, and they \textit{must} end (at $t=T$+1) with their full allocation at any $\OFF$ state.

\textit{Spatiotemporal allocation with deadline constraints. \ } %
The objective of the player is to minimize the total cost while satisfying the workload's deadline constraint.
Let $\mathcal{I}:=\{f_t(\cdot)\}_{t\in[T]}$ %
denote an input sequence of \SOAD. For a given $\mathcal{I}$, the offline version of the problem can be formulated as:
\begingroup
\allowdisplaybreaks
\begin{align} 
[\SOAD]\quad\min_{\{\mbf{x}_t\}_{ t \in [T] }} &\underbrace{ \ \sum\nolimits_{t=1}^T f_t( \mbf{x}_t ) }_{\text{Service cost}} \ + \underbrace{ \sum\nolimits_{t=1}^{T+1} g \left( \mbf{x}_t, \mbf{x}_{t-1} \right) }_{\text{Switching cost (e.g., overhead)}} 
\text{s.t.}\quad \underbrace{\sum\nolimits_{t=1}^T c(\mbf{x}_t) \geq 1,}_{\text{Deadline constraint}} \label{align:obj} \quad \mbf{x}_t \in \mathcal{X}.
\end{align}
\endgroup
In the offline setting, we note that the above formulation is convex, implying that it can be solved efficiently using, e.g., iterative methods.
However, our aim is to design an online algorithm that chooses an allocation $\mbf{x}_t$ for each time $t$ without knowing future costs $\{f_{t'}(\cdot)\}_{t'> t}$. 

{\color{blue}
\subsection{Motivating applications} \label{sec:examples}

In this section, we give examples of applications that motivate the \SOAD problem. 
We are particularly motivated by an emerging class of \textit{carbon-aware sustainability problems} that have attracted significant attention in recent years -- see the first example.  
\SOAD also generalizes canonical online search problems such as one-way trading~\cite{ElYaniv:01}, making it broadly applicable across domains as we discuss in detail below.  We focus on key components of each setting without exhaustively discussing idiosyncrasies, although we mention some extensions of \SOAD in each setting.  We defer a few more problem examples to \autoref{apx:examples}.

\noindent\textbf{Carbon-aware workload management in data centers. \ }
Consider a delay-tolerant compute job scheduled on a distributed network of data centers with the goal of minimizing the total carbon (CO$_2$) emissions of the job.  Each job arrives with a deadline $T$ that represents its required completion time, typically in minutes or hours.  Service costs $f_t^{(u)}$ represent the carbon emissions of executing a workload at full speed in data center $u$ at time $t$.  
The metric space $(X,d)$ and the spatial movement cost $d(u_{t}, u_{t-1})$ capture the \textit{carbon emissions overhead} of geographically migrating a compute workload between data centers.
The temporal switching cost $\Vert \mbf{x}_t - \mbf{x}_{t-1}\Vert_{\ell_1(\boldsymbol{\beta})}$ captures the carbon emissions overhead %
due to reallocation of resources (e.g., scaling up/down) within a single data center~\cite{Hanafy:2023:War}.  Finally, the constraint function $c( \mbf{x}_t )$ encodes what fraction of the job is completed by a given scheduling decision $\mbf{x}_t$.
The topic of shifting compute in time and space to decrease its carbon footprint has seen significant attention in recent years~\citep{Chien:21, bashir2021enabling, Wiesner:21, radovanovic2022carbon, acun2022holistic, Hanafy:23, Sukprasert:23, Chien:23}, particularly for compute needs with long time scales and flexible deadlines (e.g., ML training), which realize the most benefits from temporal shifting.
These works build on a long line of work advancing sustainable data centers more broadly (e.g., in terms of energy efficiency), some of which leverage techniques from online optimization~\cite{liu2011greening, Liu:11, lin2012dynamic, liu2012renewable, Lin:12, Liu:13, Gan:14, Wierman:14, Shuja:16, Nadjaran:17, Dou:17, zhang2017optimal, gupta2019combining, Zheng:20}.
We comment that \SOAD is the first online formulation that can model the necessary combined dimensions of spatial and temporal switching costs with deadlines.  
However, we also note that some aspects of this problem may not yet be fully captured by \SOAD -- for instance, it might be necessary to consider \textit{multiple concurrent} batch workloads rather than a single one, resource contention, data center capacity constraints, or processing delays caused by migration that are not fully captured by the current formulation.  In this sense, \SOAD serves as a building block that could accommodate extensions to consider these aspects of the problem -- we consider one such extension in \autoref{sec:timevarying}.

\noindent\textbf{Supply chain procurement. \ } Consider a firm that must source a certain amount of a good before a deadline $T$, where the good is stored in several regional warehouses~\cite{Ma:19}.
Service costs $f_t^{(u)}$ are proportional to the per unit cost of purchasing and transporting goods from warehouse $u$ during time slot $t$.
The metric space $(X,d)$ and the spatial movement cost $d(u_{t}, u_{t-1})$ capture the overhead of switching warehouses, including, e.g., personnel costs to travel and inspect goods. 
The temporal switching cost $\Vert \mbf{x}_t - \mbf{x}_{t-1}\Vert_{\ell_1(\boldsymbol{\beta})}$ captures the overhead of stopping or restarting the purchasing and transport of goods from a single warehouse.  Finally, the constraint function $c( \mbf{x}_t )$ dictates how many goods can be shipped during time $t$ according to purchasing decision $\mbf{x}_t$.
We note that in practice, the firm may need to purchase from multiple warehouses concurrently -- they are restricted to purchase from only one in the strict \SOAD formulation given above, but this can be relaxed without affecting the algorithms or results that we present in the rest of the paper.

\noindent\textbf{Mobile battery storage. \ }  Consider a mobile battery storage unit (e.g., a battery trailer~\cite{Moxion}) that must service several discharge locations by, e.g., the end of the day (deadline $T$), with the goal of choosing when and where to discharge based on the value that storage can provide in that time and place.  Service costs $f_t^{(u)}$ can represent the value of discharging at location $u$ during time slot $t$ (lower is better).  The metric space $(X,d)$ and the spatial movement cost $d(u_{t}, u_{t-1})$ capture the \textit{overhead} (e.g., lost time or fuel cost) due to moving between locations, and the temporal switching cost $\Vert \mbf{x}_t - \mbf{x}_{t-1}\Vert_{\ell_1(\boldsymbol{\beta})}$ captures the overhead of connecting or disconnecting from a discharge point at a single location, including e.g., cell degradation due to cycling~\cite{Zhang:06}. %
The constraint function $c( \mbf{x}_t )$ captures how much energy has been discharged during time $t$ according to decision $\mbf{x}_t$.  The problem of maximizing the utility that mobile battery storage provides may be useful in e.g., emergency relief situations where the main power grid has gone down.
A light extension of \SOAD may capture the case where the travel time from point to point significantly affects the feasible discharge time at the destination; such an extension would factor any lost time into the constraint function $c( \mbf{x}_t )$.

}

{\color{blue}
\subsection{Background \& assumptions}\label{sec:bkgd}

In this section, we provide background on the competitive analysis used throughout the paper and formalize our assumptions on the costs in \SOAD, motivated by the structure of applications. 

\noindent\textbf{Competitive analysis. \ }
We evaluate the performance of an online algorithm for this problem via the \textit{competitive ratio}~\cite{Manasse:88, Borodin:92}:
let $\OPT(\mathcal{I})$ denote the cost of an optimal offline solution for instance $\mathcal{I}$, and let $\ALG(\mathcal{I})$ denote the cost incurred by running an online algorithm $\ALG$ over the same instance. Then the competitive ratio of $\ALG$ is defined as $
\textnormal{CR}(\ALG) \coloneqq \sup_{\mathcal{I} \in \Omega} \nicefrac{ \ALG(\mathcal{I}) }{ \OPT(\mathcal{I}) } \eqqcolon \eta
$, where $\Omega$ is the set of all feasible inputs for the problem, and \ALG is said to be $\eta$-\textit{competitive}.
Note that $\textnormal{CR}(\ALG)$ is always at least $1$, and a \textit{smaller} competitive ratio implies that the online algorithm is guaranteed to be \textit{closer} to the offline optimal solution.  
If $\ALG$ is randomized, we replace the cost $\ALG(\mathcal{I})$ with the expected cost (over the randomness of the algorithm).

}
\begin{ass}
\label{ass:hitting-cost}
{\color{blue}
    Each service cost function $f_t(\cdot)$ satisfies bounds, i.e., ${{f}_t^{(u)} \in [{c}^{(u)} L, {c}^{(u)} U]}$ %
    for all $\ON^{(u)} : u \in X$ and for all $t \in [T]$, where $L$ and $U$ are known, positive constants. 
}
\end{ass}\vspace{-0.5em}
{\color{blue}
\noindent This assumption encodes the physical idea that there exist upper and lower bounds on the service cost faced by the player.  $L$ and $U$ are normalized by the throughput coefficient $c^{(a)}$ so that they can be independent of the amount of the deadline constraint satisfied by servicing the workload at a specific point $a \in X$.
}
\begin{ass}
\label{ass:switching-cost}
The temporal switching cost factor is bounded by $\beta^{(a)} \le \tau {c}^{(a)}$ for all $a \in X$. The normalized spatial distance between any two \ON states is upper bounded by $D$, i.e., $D =  \sup_{u, v \in X : u \not = v} \frac{d( u, \ v )}{\min \{ c^{(u)}, c^{(v)} \} }$. {\color{blue} Further, we assume that $D + 2\tau \leq U-L$.  }
\end{ass}\vspace{-0.5em}
{\color{blue}
In \SOAD, $\tau$ represents the worst-case overhead of stopping, starting, or changing the rate of service %
at a single point, while $D$ represents the worst-case overhead incurred by moving the allocation between the two most distant points.  In, e.g., the applications mentioned above, we typically expect $\tau$ to be much smaller than $D$.
Note that there are two \ON states with a normalized distance greater than $U-L$, one of these states should be pruned from the metric, because moving the allocation between them would negate any benefit to the service cost. 
}
Specifically, recall that $L$ and $U$ give bounds on the total service cost of the workload, and consider an example of two points $u, v$ that are normalized distance $D' > U-L$ apart, with ${c}^{(u)} = {c}^{(v)} = 1$, and $\tau = 0$.
Let the starting point $u$ and other point $v$ have service costs that are the worst-possible and best-possible (i.e., $U$ and $L$), respectively. Observe that if the player opts to move the allocation to $v$, they incur a movement cost of $D'$ and a total objective of $L + D'$.  In contrast, if they had stayed at point $u$, their total objective would be $U$, which is $< L + D'$ by assumption.

\subsection{Connections to existing models and challenges}\label{sec:connections}

As discussed in the related work, %
\SOAD exhibits similarities to two long-standing tracks of literature in online algorithms; however, \SOAD is distinct from and cannot be solved by existing models.

The first of these is the work on the classic \textit{metrical task systems} (\MTS) problem introduced by \citet{Borodin:92} and related forms, including smoothed online convex optimization (\SOCO)~\cite{chenSmoothedOnlineConvex2018a}. In these works, an online player makes decisions with the objective of minimizing the sum of the service cost and switching cost. However, standard algorithms for $\mathsf{MTS/SOCO}$ are not designed to handle the type of long-term constraints (e.g., such as deadlines) that \SOAD considers. 
{\color{blue}
Moreover, standard \MTS and \SOCO algorithms are designed to address either movement cost over points (i.e., $d(a_{t-1},a_t)$) or temporal switching cost (i.e., $\Vert \mbf{x}_t - \mbf{x}_{t-1}\Vert_{\ell_1(\boldsymbol{\beta})}$). \SOAD requires a spatiotemporal allocation that considers both types of switching costs simultaneously.  
}

On the other hand, the \textit{one-way trading} (\OWT) problem introduced by \citet{ElYaniv:01} and related online knapsack problems~\cite{Zhou:08, SunZeynali:20} consider online optimization with long-term constraints. To address these constraints, canonical algorithms use techniques such as threat-based or threshold-based designs to "hedge" between the extremes of quickly fulfilling the constraint and waiting for better opportunities that may not materialize. However, these works do not consider switching costs and rarely address multidimensional decision spaces.

The design of algorithms via competitive analysis for \MTS/\SOCO with long-term constraints has long been an open problem, and has seen only limited progress in a few recent works. 
{\color{blue}
These works have primarily leveraged techniques from online search, generalizing unidimensional problems such as $k$-search~\cite{Lechowicz:23} and \OWT~\cite{Lechowicz:24, Bostandoost:24} to additionally consider a \textit{temporal switching cost}. 
}
A recent study considered a generalization to the multidimensional case, introducing the problem of convex function chasing with a long-term constraint (\CFL)~\cite{Lechowicz:24CFL}.  The authors of this work propose a competitive algorithm for \CFL, although their results depend on a very specific metric structure ($\ell_1$ vector spaces or weighted star metrics), which cannot be used to model the general spatial and temporal switching costs in \SOAD.  
{\color{blue}
Furthermore, even in the multidimensional case, these existing works that consider switching costs and long-term constraints assume a \textit{single source} of switching costs, i.e., that the cost to switch into a state making progress towards the long-term constraint is the same as the cost to switch out of that state.  
This type of structure and analysis fails in the \SOAD setting due to the generality of the metric (i.e., moving to a new point while ``switching \ON'' to complete some of the workload and ``switching \OFF'' within the new point have different costs). 
}

{\color{blue}
Worst-case competitive guarantees provide robustness against non-stationarities in the underlying environment, which may be desired for applications such as carbon-awareness due to the demonstrated non-stationarity associated with such signals~\cite{Sukprasert:23}.
However, algorithms that are purely optimized for worst-case guarantees are often overly pessimistic. %
To address this, we study learning-augmented design in the \SOAD setting, which brings additional challenges.
}
For instance, existing learning-augmented results for \MTS/\SOCO and \OWT each leverage distinct algorithm design strategies based on different advice models that separately address features of their problem setting (i.e., switching costs, deadline constraints).  These prior results naturally prompt questions about how to incorporate ML advice in a performant way that can simultaneously handle the generality of the switching costs in \SOAD while ensuring that the deadline constraint is satisfied.

\subsection{Preliminaries of technical foundations}\label{sec:prelims}

In this section, we introduce and discuss techniques from different areas of the online algorithms literature that we use in subsequent sections to address the \SOAD challenges discussed above.

\smallskip

\noindent\textbf{Unifying arbitrary metrics. \ } 
The generality of the metric space in \SOAD is a key challenge that precludes the application of algorithm design techniques from prior work requiring specific metric structures.
{\color{blue}
In classic \MTS, the online player also makes decisions in an arbitrary metric space $(X, d)$, which poses similar challenges for algorithm design.
A key result used to address this is that of \citet{Fakcharoenphol:07}, who show that for any $n$-point metric space $(X,d)$,  there exists a probabilistic embedding into a \textit{hierarchically separated tree} (HST) $\calT = (V,E)$ with at most $O( \log n )$ distortion, i.e., $\mathbb{E}_{\calT}\left[ d^{(\calT)}(u,v) \right] \leq O(\log n ) d(u,v)$ for any $u, v \in X$. %
For \MTS, this result implies that any $\eta$-competitive algorithm for \MTS \textit{on trees} is immediately $O(\log n) \eta$-competitive in expectation for \MTS on general $n$-point metrics, exactly by leveraging this embedding.
}

To solve \MTS using such a tree, \citet{Bubeck:21MTSTrees} consider a randomized algorithm on the leaves of $\calT$ denoted by $\mathcal{L}$, where the nodes of $\mathcal{L}$ correspond to points in $X$. 
{\color{blue}
This randomized metric space is given by $(\Delta_{\mathcal{L}}, \ \mathbb{W}^1)$, where $\Delta_{\mathcal{L}}$ is the probability simplex over the leaves of $\calT$, and $\mathbb{W}^1$ denotes the \textit{Wasserstein-1 distance}.  For two probability distributions $\mbf{p}, \mbf{p}' \in \Delta_{\mathcal{L}}$, the Wasserstein-1 distance is defined as $\mathbb{W}^1 (\mbf{p}, \mbf{p}') \coloneqq \min_{\pi_{x, x'} \in \Pi(\mbf{p}, \mbf{p}')} \mathbb{E} \left[ d( x , x' ) \right]$, where $\Pi(\mbf{p}, \mbf{p}')$ is the set of transport distributions over $\mathcal{L}^2$ with marginals $\mbf{p}$ and $\mbf{p'}$.  
A randomized algorithm that produces marginal distributions $\mbf{p} \in \Delta_{\mathcal{L}}$ then \textit{couples} consecutive decisions according to the optimal transport plan $\pi_{x,x'}$ defined by Wasserstein-1. 
\citet{Bubeck:21MTSTrees} further show that $(\Delta_{\mathcal{L}}, \ \mathbb{W}^1)$ is bijectively isometric to %
a convex set $K$ with a weighted $\ell_1$ norm $\norm{\cdot}$ based on edge weights in the tree.
}

\textit{Metric tree embedding for \SOAD. \ } 
{\color{blue}
Prior approaches by \citet{Fakcharoenphol:07} and \citet{Bubeck:21MTSTrees} are able to manage the spatial movement cost from moving allocation between points.
} %
To further accommodate the temporal switching cost between $\ON$ and $\OFF$ states within a single point, we develop a probabilistic tree embedding $\calT = (V,E)$ and the corresponding vector space $(K, \norm{\cdot})$ in the following \sref{Definition}{dfn:tree} and \sref{Definition}{dfn:vs}, respectively.

\begin{dfn}[Probabilistic tree embedding $\calT = (V,E)$ for \SOAD] \label{dfn:tree}
{\color{blue}
Let $(X,d)$ denote the underlying metric space over $n$ points, and let $\calT'$ denote an HST constructed on the points of $X$ according to the method by~\citet{Fakcharoenphol:07}. 
}
Label the leaves of $\calT'$ according to the $n$ $\ON$ states, one for each point.  
{\color{blue}
Then the final tree $\calT$ is constructed by adding $n$ edges and $n$ nodes to just the leaves of $\calT'$ -- each new node represents the corresponding $\OFF$ state at that point, and the new edge is weighted according to the temporal switching cost at that point (i.e., $\beta^{(u)}$).  
The resultant ``state set'' $\mathcal{S}$ includes both the leaves of $\calT$ %
(\OFF states) and their immediate predecessors (\ON states). 
}

\end{dfn}
\begin{figure*}[h]
    \minipage{\textwidth}
    \centering\vspace{-1em}
    \includegraphics[width=0.8\linewidth]{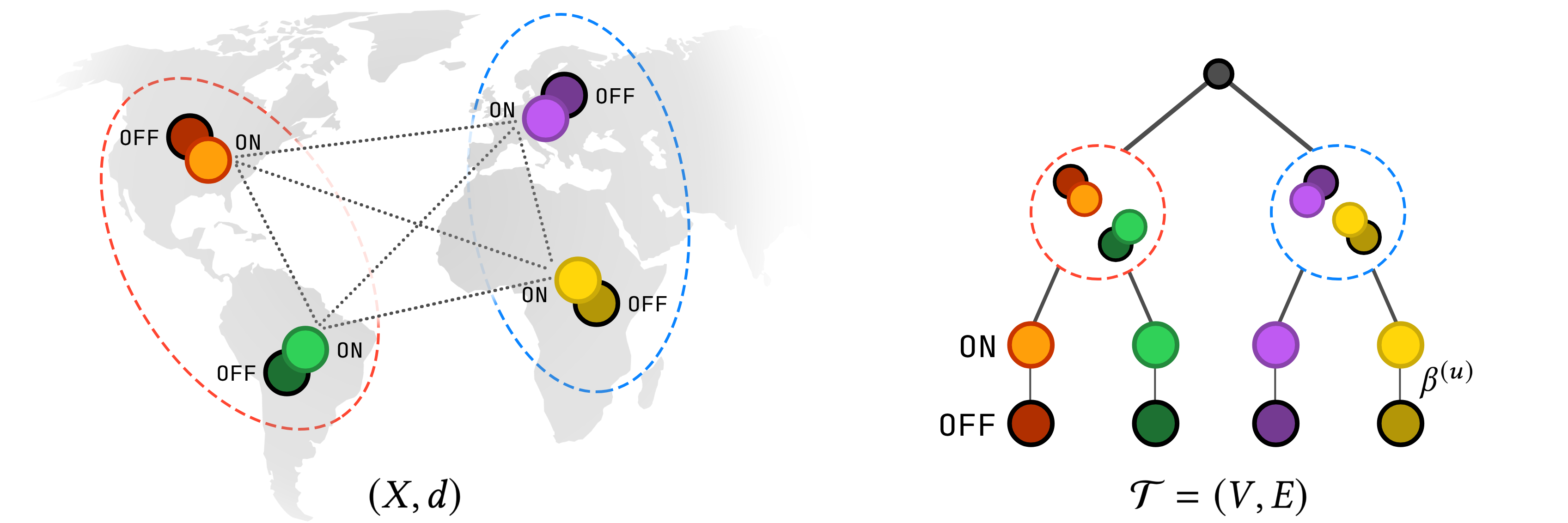}\vspace{-1em}
    \caption{An illustration of the probabilistic tree embedding (\sref{Def.}{dfn:tree}) for the motivating application.  {\color{blue} Points in the metric are represented as pairs of circles on the left.}  On the right, the first three levels of the tree approximate the metric space $(X,d)$~\cite{Fakcharoenphol:07}, and the last level captures the \ON \ / \OFF structure of \SOAD. }\vspace{-1em}
    \endminipage\hfill
\end{figure*}
\noindent {\color{blue}
Note that $\calT$ preserves distances between the points of $X$ with expected $O(\log n)$ distortion, while the switching cost between \ON and \OFF states at a single point is preserved exactly.
}
Our competitive algorithm (see \autoref{sec:pcm}) operates in a vector space constructed according to this HST embedding.

\begin{dfn}[Vector space $(K, \norm{\cdot})$]
\label{dfn:vs}
{\color{blue}
Given a hierarchically separated tree $\calT = (V,E)$ constructed according to \sref{Definition}{dfn:tree}, with root $r \in V$, state set $\mathcal{S} \subseteq V$, and leaf set $\mathcal{L} \subset \mathcal{S}$, let $P^{(u)}$ denote the parent of any node $u \in V \setminus r$.
}
Construct the following set:
{\small
\setlength{\abovedisplayskip}{2pt}
\setlength{\belowdisplayskip}{2pt}
{\color{blue}
\begin{equation*}
K \coloneqq \left\{ \mbf{k} \in \mathbb{R}^{\vert V \vert} : \mbf{k}^{(r)} = 1, \text{ and } \ \forall u \in V \setminus \mathcal{L}, \mbf{k}^{(u)} = \sum_{v : P^{(v)} = u} \mbf{k}^{(v)}, \text{ and } \ \forall u \in \mathcal{S}, \mbf{k}^{(u)} \in [0,1] \right\}.
\end{equation*}}}
\noindent Let $\mathbf{w}$ be a non-negative weight vector on vertices of $\calT$, where $\mbf{w}^{(r)} = 0$ and $\mbf{w}^{(u)} > 0$ for all $u \in V \setminus r$. 
{\color{blue}
Recall that the edges of $\mathcal{T}$ are weighted -- we let $\mbf{w}^{(u)}$ denote the weight of edge $\{ P^{(u)}, u \}$, and define a weighted $\ell_1$ norm as $\norm{\mbf{k}} \coloneqq \sum_{u \in V} \mbf{w}^{(u)} \vert \mbf{k}^{(u)} \vert$, for any $\mbf{k} \in K$.
}
{\color{blue}
Finally, we define a linear map from $\Delta_{\mathcal{S}}$ to $K$ (and its corresponding inverse), given by a matrix map $\Phi : \mathbb{R}^{2n} \to \mathbb{R}^{\vert V \vert}$.  
}
We let $A^{(u)}$ denote the set of node $u$'s ancestors in $\calT$, and (with a slight abuse of notation) let $\mathcal{S}(i) : \{ 0, \dots, 2n \} \to \mathcal{S}$ and $K(i) : \{ 0, \dots, \vert V \vert \} \to V$ denote indexing maps that recover the object in $\mathcal{S}$ or $V$, respectively.
Then $\Phi \in \mathbb{R}^{\vert V \vert \times 2n}$ and $\Phi^{-1} \in \mathbb{R}^{2n \times \vert V \vert}$ are defined as follows:
{\small\begin{equation*}
\begin{aligned}[c]
\Phi_{i, j} \coloneqq \begin{cases}
    1 & \text{if } \ K(i) = r\\
    1 & \text{if } \ K(i) = \mathcal{S}(j)\\ %
    1 & \text{if } \ K(i) \in \mathcal{S} \text{ and } P^{(\mathcal{S}(j))} = K(i) \\
    1 & \text{if } \ K(i) \in V \setminus \mathcal{S} \text{ and } K(i) \in A^{(\mathcal{S}(j))}\\
    0 & \text{otherwise}
\end{cases}
\end{aligned}
\hspace{1em}
\begin{aligned}[c]
\Phi^{-1}_{i, j} \coloneqq \begin{cases}
    1 & \text{if } \ \mathcal{S}(i) = K(j)\\  %
    -1 & \text{if } \ K(j) \in \mathcal{L} \text{ and } P^{(K(j))} = \mathcal{S}(i) \\
    0 & \text{otherwise}
\end{cases}
\end{aligned}
\end{equation*}}
\end{dfn}
{\color{blue}
In words, $\Phi$ maps a distribution over $\Delta_{\mathcal{S}}$ to the corresponding vector in $K$ by accumulating probability mass upwards from the leaves of $\calT$ towards the root.
$\Phi^{-1}$ reverses this by selecting the appropriate indices for $u \in \mathcal{S}$ from $K$, and recovers probabilities by subtracting the mass at the $\OFF$ state from the $\ON$ state (since the \OFF state is a leaf, the \ON state accumulates its probability in $K$).
}

\textit{Randomized algorithm for \SOAD. \ } 
{\color{blue}
We define some shorthand notation. For a decision $\mbf{k} \in K$, $\mbf{p} = \Phi^{-1} \mbf{k}$ gives a corresponding probability distribution on $\Delta_{\mathcal{S}}$.  Note that $\mathcal{X} \subset \Delta_{\mathcal{S}} \subset \mathbb{R}^{2n}$, and by linearity of expectation, the service and constraint functions $f_t( \cdot ), c( \cdot ) : \mathcal{X} \rightarrow \mathbb{R}$ remain well-defined (in expectation) on $\Delta_{\mathcal{S}}$.
Within $K$, we let $\overline{f}_t(\mbf{k}) = f_t(\Phi^{-1} \mbf{k})$ and $\overline{c}(\mbf{k}) = c(\Phi^{-1} \mbf{k})$.   For a given \textit{starting point} $s \in X$, we slightly abuse notation and let $\delta_s \in \Delta_{\mathcal{S}}$ denote the Dirac measure supported at $\OFF^{(s)}$.
}
{\color{blue} 
Recall that the \SOAD formulation specifies an allocation that is \textit{discrete} in terms of choosing a point in the metric, and \textit{fractional} in terms of the resource allocation at a given point. To capture this structure while using the embedding results discussed above, we consider a \textit{mixed setting} that is probabilistic in spatial assignment but deterministic in the \ON \ / \OFF allocation.
We state the equivalence of this setting and the fully probabilistic one below, deferring the proof to \autoref{apx:equivRand}. 
}

\begin{thm} \label{thm:equivRand}
{\color{blue}
    For a randomized \SOAD decision $\mbf{p}_t \in \Delta_{\mathcal{S}}$, the expected cost is equivalent if a point in $X$ is first chosen probabilistically and the \ON \ / \OFF probabilities at that point are interpreted as (deterministic) fractional allocations.
}
\end{thm}

\smallskip

\noindent\textbf{Enforcing a deadline constraint using pseudo-cost. \ } 
Existing algorithms for \MTS-type problems are not designed to handle a deadline constraint while remaining competitive. For \SOAD, we draw from the \textit{pseudo-cost minimization}~\cite{SunZeynali:20} approach for online search problems, where the player is subject to a long-term \textit{buying/selling} constraint that poses similar algorithmic challenges. %

Under the pseudo-cost framework, we start by assuming that a \textit{mandatory allocation} condition exists to strictly enforce the deadline constraint.  Let $z^{(t)}$ denote the fraction of the deadline constraint satisfied (in expectation) up to time $t$ (we henceforth call this the \textbf{\textit{utilization}}). 
{\color{blue}
To avoid violating the constraint, a mandatory allocation begins at time $j$, as soon as $( T - (j+1) ) {c}^{(u)} < ( 1 - z^{(j)} ) \ \forall u \in X$, i.e., when the remaining time after the current slot would be insufficient to satisfy the constraint.
Note that in practice, $z^{(j)}$ would be replaced by the actual constraint satisfaction so far.
}
During the mandatory allocation, a cost-agnostic player takes control and makes maximal allocation decisions %
to ensure the workload is finished before the deadline.

{\color{blue}
Intuitively, the mandatory allocation complicates competitive analysis -- in the worst-case, an adversary can present the worst service cost ($U$) during the final steps.
The key idea behind pseudo-cost minimization is to rigorously characterize a \textit{trade-off} between completing the constraint ``too early'' and waiting too long (i.e., risking a mandatory allocation) using a \textit{pseudo-cost function}.
}
Such a function takes the lower and upper bounds on service cost (i.e., $L$ and $U$) as parameters, and assigns a \textit{pseudo-cost} to each increment of progress towards the constraint.  
{\color{blue}
In an algorithm, this function is used by solving a small minimization problem at each step, whose objective considers the true cost of a potential decision and an integral over the pseudo-cost function -- generating decisions using this technique creates a connection between the utilization and the best service cost encountered throughout the sequence, ensuring that the algorithm completes ``exactly enough'' of the constraint before mandatory allocation in order to achieve a certain competitive ratio against the best service cost, which is a lower bound on \OPT.}
We also note that in the learning-augmented setting, the pseudo-cost minimization problem can be combined with a \textit{consistency constraint}, as shown by~\cite{Lechowicz:24CFL}, to integrate certain forms of advice without losing the robust (i.e., competitive) qualities of the pseudo-cost.

\section{\PCM: A Competitive Online Algorithm}
\label{sec:pcm}

This section presents a randomized competitive algorithm for \SOAD that leverages the metric tree embeddings and pseudo-cost minimization design discussed above in \autoref{sec:prelims}.
{\color{blue}
We further show that our algorithm achieves a competitive ratio that is optimal for $\SOAD$ up to $\log$ factors.
}

\smallskip

\subsection{Algorithm description}

We present a randomized pseudo-cost minimization algorithm (\PCM) in \autoref{alg:pcm}. \PCM operates on the metric space $(K, \norm{\cdot})$ defined in \sref{Definition}{dfn:vs} and extends the original pseudo-cost minimization framework~\cite{SunZeynali:20} to address the setting where the decision space is given by an arbitrary convex set $K$ with distances given by $\norm{\cdot}$.

We define a pseudo-cost function $\psi(z):[0,1]\to [L, U]$, where $z$ is the utilization (i.e., the completed fraction of the deadline constraint in expectation).
Our construction of $\psi$ takes advantage of additional structure in the \SOAD setting -- 
{\color{blue}
this function depends on the parameters of the \SOAD problem, including $U, L, D$, and $\tau$ specified in \sref{Assumptions}{ass:hitting-cost} and~\ref{ass:switching-cost}. 
}
\vspace{-0.4em}
\begin{dfn}[Pseudo-cost function $\psi$ for \SOAD] \label{dfn:psi-pcm}
For a given parameter $\eta > 1$, the pseudo-cost function is defined as $\psi(z) = U - \tau + (\nicefrac{U}{\eta} - U + D + \tau) \exp(\nicefrac{z}{\eta}), z\in[0,1]$.
\end{dfn}
\vspace{-0.4em}
{\color{blue}
\noindent Given the pseudo-cost function from \sref{Def.}{dfn:psi-pcm}, \PCM solves a minimization problem \eqref{eq:pseudocostPCM} at each step $t$ to generate a decision $\mbf{k}_t \in K$; 
}
the objective of this problem is to minimize a combination of the per-step cost plus a pseudo-cost term that encourages (deadline) constraint satisfaction.
At a high level, the $\psi$ term enforces that $\mbf{k}_t$ satisfies ``exactly enough'' of the deadline constraint (in expectation) to make adequate progress and maintain an expected competitive ratio of $\eta$ against the current estimate of $\OPT$, without ``overbuying'' and preventing better costs from being considered in the future.
{\color{blue}
At a glance, it is not obvious that the pseudo-cost minimization problem is straightforward to actually solve in practice.  In the following, we show that \eqref{eq:pseudocostPCM} is a convex minimization problem.
}
\vspace{-0.4em}
{\color{blue}
\begin{theorem}\label{thm:PCMconvex}
    Under the assumptions of \SOAD, the pseudo-cost minimization \eqref{eq:pseudocostPCM} is a convex minimization problem.
\end{theorem}
}
\vspace{-0.4em}

\begin{algorithm}[t]
   \caption{Pseudo-cost minimization algorithm for \SOAD (\PCM)}
   \label{alg:pcm}
{\small
\begin{algorithmic}
   \State {\bfseries input:} constraint function $c(\cdot)$, convex set $K$ with distance metric $\lVert \cdot \rVert_{\ell_1 (\mbf{w})}$, pseudo-cost function $\psi(z)$, starting \OFF state $s \in \mathcal{S}$.
   \State {\bfseries initialize:} $z^{(0)} = 0; \ \mbf{k}_0 = \Phi \delta_s; \ \mbf{p}_0 = \delta_s$.
   \While {cost function $f_t (\cdot)$ is revealed and $z^{(t-1)} < 1$}
   \State Solve \textbf{pseudo-cost minimization} problem: {\small \begin{align}
       \mbf{k}_t &= \argmin_{\mbf{k} \in K : \overline{c}(\mbf{k}) \leq 1-z^{(t-1)}} \overline{f}_t(\mbf{k}) + \lVert \mbf{k} - \mbf{k}_{t-1} \rVert_{\ell_1 (\mbf{w})} - \int_{z^{(t-1)}}^{z^{(t-1)} + \overline{c}(\mbf{k})} \psi (u) du \label{eq:pseudocostPCM},\\
       \mbf{p}_t &= \Phi^{-1} \mbf{k}_t.
    \end{align}}
   \State Update utilization $z^{(t)} = z^{(t-1)} + c(\mbf{p}_t)$.
   \EndWhile
\end{algorithmic}}
\end{algorithm}
\setlength{\textfloatsep}{6pt}%
\setlength{\intextsep}{0pt}

{\color{blue} We defer the proof of \autoref{thm:PCMconvex} to \autoref{apx:pseudo-convex}.  At a high-level, the result implies that the solution to \eqref{eq:pseudocostPCM} can be found efficiently using convex programming techniques~\cite{BV:04, CVXPY}. }
Compared to prior works~\cite{Lechowicz:24CFL, SunZeynali:20}, our design of $\psi$ differentiates between the spatial movement cost and temporal switching cost (in particular, $D$ only appears within the exponential, while $\tau$ appears inside and outside of the exponential term).
This removes a source of pessimism %
-- when \PCM makes a decision to move to a distant point in the metric (i.e., paying a worst-case factor of $D$), it can safely assume that it will only have to pay a factor of $\tau$ to switch \OFF if the next cost function is bad.  
{\color{blue}
This allows \PCM to take advantage of spatially distributed \OFF states, where all of the existing works that use a pseudo-cost paradigm for temporal load-shifting cannot. 
}

{\color{blue}
We note that \PCM's decisions in $K$ are marginal probability distributions over $\Delta_{\mathcal{S}}$ -- %
we briefly detail how feasible deterministic decisions in $\mathcal{X}$ are extracted from these outputs.
We assume the player interprets distributions according to a mixed random / fractional setting (see \autoref{thm:equivRand}), allowing them to make fractional resource allocation decisions within a single point while the allocation is probabilistically assigned to a single point.
}
{\color{blue}
At each time step, \PCM generates $\mbf{p}_t \in \Delta_{\mathcal{S}}$.  We let $\mbf{r}_t \coloneqq \{ r_t^{(u)} \leftarrow p_t^{\ON^{(u)}} \kern-0.5em + p_t^{\OFF^{(u)}} \kern-0.5em : u \in X \}$ aggregate the $\texttt{ON / OFF}$ probabilities at each location of $X$.  Given this \textit{spatial} distribution over {$X$}, consecutive decisions should be jointly distributed according to the optimal {transport} plan between $\mbf{r}_{t-1}$ and $\mbf{r}_t$, given by $(\mbf{r}_{t}, \mbf{r}_{t-1}) \thicksim \pi_t \coloneqq \argmin_{\pi \in \Pi(\mbf{r}_t, \mbf{r}_{t-1})} \mathbb{E} \left[ d( u_t, u_{t-1} ) \right]$, 
where $( u_t, u_{t-1} ) \thicksim \pi_t$ and $\Pi(\mbf{r}_t, \mbf{r}_{t-1})$ is the set of distributions over $X^2$ with marginals $\mbf{r}_t$ and $\mbf{r}_{t-1}$.  If the decision-maker couples decisions according to $\pi_t$, then the \textit{expected} spatial movement cost of the deterministic decisions is equivalent to $\mathbb{W}^1 ( \mbf{r}_t, \mbf{r}_{t-1} ) $, the spatial Wasserstein-1 distance between $\mbf{r}_t$ and $\mbf{r}_{t-1}$.
Given a previous deterministic point assignment $u_{t-1}$,  the player can obtain the point assignment $u_t$ by sampling through the conditional distribution $\pi_t(u_t|u_{t-1})$.
The fractional \texttt{ON / OFF} allocation in $\mbf{x}_t$ %
at the chosen location $u_t \in X$ is then given by $\nicefrac{ p_t^{\ON^{(u)}}}{ r_t^{(u)} }$ and $\nicefrac{ p_t^{\OFF^{(u)}}}{r_t^{(u)}}$, respectively; by \autoref{thm:equivRand}, this gives that $\mathbb{E} \left[ g( \mbf{x}_t, \mbf{x}_{t-1} ) \right] = \mathbb{W}^1 (\mbf{p}_t , \mbf{p}_{t-1})$. 
}

\subsection{Main results}

In \autoref{thm:etaCompPCM}, we state a bound on the competitive ratio of \PCM.  
\begin{thm} \label{thm:etaCompPCM}
Under Assumptions~\ref{ass:hitting-cost} and~\ref{ass:switching-cost}, \emph{\PCM} is $O(\log n) \eta$-competitive for \SOAD, where $\eta$ is the solution to $\ln \left( \frac{U - L - D - 2\tau}{U - \nicefrac{U}{\eta} - D} \right) = \frac{1}{\eta}$ and given by:
{\setlength{\abovedisplayskip}{0pt}
\setlength{\belowdisplayskip}{3pt}
\begin{equation}
    \eta \coloneqq \left[ W \left( \frac{(D+L-U+2\tau) \exp \left(\frac{D-U}{U} \right) }{ U } \right) + \frac{U - D}{U} \right]^{-1}, \label{eq:eta}
\end{equation}}
where $W$ is the Lambert $W$ function~\cite{Corless:96LambertW}, and the $O(\log n)$ factor is due to the tree embedding~\cite{Fakcharoenphol:07}.
\end{thm}

Compared to previous works such as $\MTS$ and $\OWT$, the competitive bound in \autoref{thm:etaCompPCM} compares favorably. 
{\color{blue}
In particular, the upper bound is \textit{better} than one might expect from e.g., combining the bounds of $\MTS$ and $\OWT$.
}
For the minimization variant of \OWT, the optimal competitive ratio due to \citet{Lorenz:08} is $\left[ W \left( \left( \nicefrac{1}{\theta} - 1 \right) e^{-1} \right) + 1 \right]^{-1}$, where $\theta$ is defined as $\nicefrac{M}{m}$, and $M \geq m$ are bounds on the prices (i.e., $(m, M) \approx (L, U)$).  For \MTS, the randomized state-of-the-art due to \citet{Bubeck:21MTSTrees} is $O((\log n)^2)$.  Asymptotically, compared to both of these bounds, $\eta$ ``loses'' a $\log$ factor depending on the number of points in the metric, and it is known that $W(x) \thicksim \ln (x)$ as $x \to \infty$ ~\cite{HoorfarHassani:08, Stewart:09}.  
{\color{blue}
Compared to \OWT, $\eta$ adds a dependency on $D$ and $\tau$, parameters describing the cost due to the metric and switching, but we note that \sref{Assumption}{ass:switching-cost} (i.e., bounds on $D$ and $\tau$ in terms of $U$ and $L$) prevents the competitive ratio from significantly increasing. 
}

{\color{blue}
Given the result in \autoref{thm:etaCompPCM}, a natural question is whether any online algorithm for \SOAD can achieve a better competitive bound.  
We answer this in the negative, showing that \PCM's competitive ratio is the best achievable up to $\log$ factors that are due to the metric embedding.  In particular, we show a class of difficult instances on which no algorithm can achieve a competitive ratio better than $\eta$ -- since the definition of the competitive ratio covers all valid inputs, this gives a corresponding lower bound on the competitive ratio of any algorithm for \SOAD. %
}

\begin{thm} \label{thm:lowerboundCSWM}
{\color{blue}
For any $U, L, \tau$, and $D \in [0, (U-L))$, there exists a set of \SOAD instances on a weighted star on which no algorithm \emph{$\ALG$} can achieve \emph{$\nicefrac{\ALG}{\OPT}$} better than $\eta$ (for $\eta$ defined in \eqref{eq:eta}).}
\end{thm}

\subsection{Proof overviews}

{\color{blue}
We give proof sketches of \sref{Theorems}{thm:etaCompPCM} and \ref{thm:lowerboundCSWM}, relegating the full proofs of both to \autoref{apx:pcm}.
}

\smallskip
\noindent \textbf{Proof Sketch of \autoref{thm:etaCompPCM}.}
To show this result, we give two lemmas to characterize the cost of $\OPT$ and the expected cost of $\PCM$, respectively.  First, note that the solution given by \PCM is feasible, by definition of the mandatory allocation (i.e., $\sum_{t=1}^T c(\mbf{p}_t) = 1$).
On an arbitrary \SOAD instance $\mathcal{I} \in \Omega$, we denote the final utilization (before the mandatory allocation) by $z^{(j)}$.  
\vspace{-0.4em}
\begin{lem}
\label{lem:pcm-opt-lb}
The offline optimum is lower bounded by $\emph{\OPT}(\mathcal{I}) \ge \frac{\max \left\{ \psi(z^{(j)}) - D, L \right\}}{O(\log n)}$.
\end{lem}
\vspace{-0.4em}
We can show by contradiction that for any instance, the definition of the pseudo-cost minimization enforces that $\psi(z^{(j)}) - D$ is a lower bound on the \textit{best service cost} seen in the sequence (see \eqref{align:mingrad}). 
Note that the best choice for $\OPT$ is to service the entire workload at the minimum cost (if it is feasible).
{\color{blue}
This yields a corresponding lower bound on $\OPT$ -- formally, $\OPT(\mathcal{I}) \geq \nicefrac{\max \{ \psi(z^{(j)}) - D, \  L \}}{O(\log n)}$, where the $\log$ factor appears due to the distortion in the metric tree embedding.
}
\vspace{-0.4em}
\begin{lem}
\label{lem:pcm-alg-ub}
$\emph{\PCM}$'s expected cost is bounded by $\ex[ \emph{\PCM}(\mathcal{I}) ] \le \int_{0}^{z^{(j)}} \kern-1em \psi(u) du + (1- z^{(j)}) U + \tau z^{(j)}$.
\end{lem}
\vspace{-0.4em}
The definition of the pseudo-cost provides an automatic bound on the expected cost incurred during any time step where progress is made towards the deadline constraint (i.e., whenever the service cost is non-zero).  
{\color{blue}
We show that $\tau z^{(j)}$ is an upper bound on the \textit{excess cost} that can be incurred by \PCM in the other time steps (i.e., due to temporal switching costs, see \eqref{align:excesscost}).}
Summing over all time steps, this gives that the expected cost of \PCM is upper bounded by $\int_{0}^{z^{(j)}} \kern-1em \psi(u) du + (1- z^{(j)}) U + \tau z^{(j)}$, where $(1- z^{(j)})U$ is due to the mandatory allocation.

Combining the two lemmas and using the definition of the pseudo-cost function to observe that $\int_{0}^{z^{(j)}} \kern-1em \psi(u) du + (1- z^{(j)}) U + \tau z^{(j)} \le \eta \left[ \psi(z^{(j)}) - D \right]$ (see \eqref{align:bestservicecost}) completes the proof. \hfill \qed

\smallskip
\noindent \textbf{Proof Sketch of \autoref{thm:lowerboundCSWM}.}
In \sref{Definition}{dfn:yadversary-caslb}, we define a class of \textit{$y$-adversaries} denoted by $\mathcal{G}_y$ and $\mathcal{A}_y$ for $y \in [L,U]$, along with a corresponding weighted star metric $X$ that contains $n$ points, each with $2$ states (\ON and \OFF), where the distance between any two points in the metric is exactly $D$.  These adversaries present cost functions at the $\ON$ states of $\mathcal{X}$ in an adversarial order that forces an online algorithm to incur a large switching cost.  
{\color{blue}
The $\mathcal{G}_y$ adversary presents a cost function at each step that is ``bad'' (i.e., $U$) in all $\ON$ states except for one which is \textit{not} at the starting point or the current state of online algorithm \ALG.
}
The $\mathcal{A}_y$ adversary starts by exactly mimicking $\mathcal{G}_y$ and presenting ``good'' cost functions at distant points, before eventually presenting ``good'' cost functions at the starting point.  Both adversaries present ``good'' cost functions in an adversarial non-increasing order, such that the optimal solutions approach $y$ -- formally, $\OPT(\mathcal{G}_y) \to \min\{ y + D + \tau, U \}$, and $\OPT(\mathcal{A}_y) \to y$.  
{\color{blue}
By competing against both adversaries simultaneously, this construction captures a trade-off between being too eager/reluctant to move away from the starting point.
}

{\color{blue}
Under this special metric and class of adversaries, the cost of any (potentially randomized) online algorithm \ALG can be fully described by two arbitrary \textit{constraint satisfaction functions} $s(y), t(y) : [L,U] \to [0,1]$ (see \eqref{eq:lb-special-cost}), where each function corresponds one of two \textit{stages} of the adversary (i.e., ``good'' cost functions arriving at spatially distant points, or at the starting point).
}
For \ALG to be $\eta^\star$-competitive (where $\eta^\star$ is unknown), we give corresponding conditions on $s(y)$ and $t(y)$ expressed as differential inequalities (see \eqref{align:lb-conditionineq}).  
{\color{blue}
By applying 
Gr\"{o}nwall's Inequality~\cite[Theorem 1, p. 356]{Mitrinovic:91}, this gives a \textit{necessary condition} such that $\eta^\star$ must satisfy: $\eta^\star \ln \left( \frac{U - L - D - 2\tau}{ U - \nicefrac{U}{\eta^\star} - D } \right) - \frac{\eta^\star D + \eta^\star 2 \tau}{\nicefrac{U}{\eta^\star} - U + D} \leq s(L) \leq 1 - t(L) \leq 1 - \frac{\eta^\star D + \eta^\star 2 \tau}{\nicefrac{U}{\eta^\star} - U + D}$.
}
The optimal $\eta^\star$ is obtained by solving for the transcendental equation that arises when the inequalities are binding, yielding the result. \hfill \qed

\section{\CarbonClipper: A Learning-Augmented Algorithm}
\label{sec:clip}
{\color{blue}
In this section, we consider how a learning-augmented algorithm for \SOAD can leverage \textit{untrusted advice} to improve on the average-case performance of \PCM while retaining worst-case guarantees. }
For learning-augmented algorithms, competitive ratio is interpreted via the notions of  \textit{consistency} and \textit{robustness}~\cite{Lykouris:18, Purohit:18}.  Letting $\ALG$ denote a learning-augmented online algorithm provided with advice denoted by $\ADV$, $\ALG$ is said to be $\alpha$-\textbf{consistent} if it is $\alpha$-competitive with respect to $\ADV$, %
and $\gamma$-\textbf{robust} if it is $\gamma$-competitive with respect to $\OPT$ when given any advice (i.e., regardless of $\ADV$'s performance).
We present \CarbonClipper (see \autoref{alg:carbonclipper}), which uses an adaptive optimization-based approach combined with the robust design of \PCM to achieve an optimal consistency-robustness trade-off.  We start by formally defining the advice model we use below.
\begin{dfn}[Black-box advice model for \SOAD] \label{dfn:adv}
{\color{blue}
For a given \SOAD instance $\mathcal{I} \in \Omega$, we let $\emph{\ADV}(\mathcal{I})$ denote \textbf{untrusted black-box decision advice}, i.e., %
$\emph{\ADV}(\mathcal{I}) \coloneqq  \{ \mbf{a}_t \in \Delta_{\mathcal{S}} : t \in [T] \}.$
}
If \emph{\ADV} is correct, a player that plays $\mbf{a}_t$ at each step attains the optimal solution ($\emph{\ADV}(\mathcal{I}) = \emph{\OPT}(\mathcal{I})$).
\end{dfn}
{\color{blue}
Although $\mbf{a}_t$ is defined on the probability simplex $\Delta_{\mathcal{S}}$, a deterministic \ADV at time $t$ is given by combining the Dirac measure supported at a point and a specific \ON \ / \OFF allocation.
We henceforth assume that $\ADV$ is \textit{feasible}, satisfying the %
constraint ($\sum\nolimits_{t=1}^T c(\mbf{a}_t) \geq 1$).  }
{\color{blue}
While it is not obvious a machine learning model could directly provide such feasible predictions, in practice, we leverage the black-box nature of the definition to combine e.g., machine-learned predictions of relevant costs with a post-processing pipeline that solves for a \textit{predicted optimal solution} (see \autoref{sec:expsetup}).
}

\subsection{\CarbonClipper: an optimal learning-augmented algorithm}
{\color{blue}
We present \CarbonClipper (\textbf{s}patio\textbf{t}emporal \textbf{c}onsistency-\textbf{li}mited \textbf{p}seudo-cost minimization,  \autoref{alg:carbonclipper}), which %
exactly matches a lower bound on the optimal robustness-consistency trade-off (\autoref{thm:optimalconstrobCSWM}) for \SOAD. 
}
\CarbonClipper takes a hyperparameter $\varepsilon \in (0, \eta - 1]$, which parameterizes a trade-off between following the untrusted advice ($\varepsilon \to 0$) and prioritizing robustness ($\varepsilon \to \eta - 1$).
We start by defining a \textit{target robustness factor} $\gamma^{(\varepsilon)}$, which is the unique solution to the following equation:
{\setlength{\abovedisplayskip}{3pt}
\setlength{\belowdisplayskip}{3pt}
\begin{equation}
    \gamma^{(\varepsilon)} = \varepsilon + \frac{U}{L} - \frac{\gamma^{(\varepsilon)} (U-L+D)}{L} \ln \left( \frac{U - L - D - 2\tau}{U - \nicefrac{U}{\gamma^{(\varepsilon)}} - D - 2\tau} \right). \label{eq:gamma}
\end{equation}}
We note that $\gamma^{(\varepsilon \to 0)} \to \nicefrac{U}{L}$, which is a trivial competitive ratio for any mandatory allocation scheme (i.e., if the entire constraint is satisfied at the deadline for the worst price $U$). 
{\color{teal}
The precise value of $\gamma^{(\varepsilon)}$ originates from a robustness-consistency lower bound (\autoref{thm:optimalconstrobCSWM}), and \CarbonClipper uses it to define a pseudo-cost function $\psi^{(\varepsilon)}$ that enforces $\gamma^{(\varepsilon)}$-robustness in its decisions.
}
\begin{dfn}[Pseudo-cost function $\psi^{(\varepsilon)}$ for \SOAD]
For $\rho \in [0,1]$ and $\gamma^{(\varepsilon)}$ given by \eqref{eq:gamma}, let $\psi^{(\varepsilon)}(\rho)$ be defined as: $\psi^{(\varepsilon)}(\rho) = U +D - \tau + (\nicefrac{U+D}{\gamma^{(\varepsilon)}} - U + D + \tau) \exp(\nicefrac{\rho}{\gamma^{(\varepsilon)}})$.
\end{dfn}
Similarly to \PCM (see \autoref{sec:pcm}), $\psi^{(\varepsilon)}$ is used in a minimization problem solved at each time step to obtain a decision.
However, since \CarbonClipper %
must also consider the actions of \ADV, it follows the consistency-limited pseudo-cost minimization paradigm, which places a \textit{consistency constraint} on the aforementioned minimization.  
{\color{blue}
This constraint enforces that \CarbonClipper always satisfies $(1+\varepsilon)$-consistency, which is salient when $\ADV$ is close to optimal.  Within this feasible set, the pseudo-cost minimization drives \CarbonClipper towards decisions that are ``as robust as possible.''
}

\noindent\textbf{Additional challenges in algorithm design. \ }
{\color{teal}
In contrast to prior applications of the \CLIPacro technique~\cite{Lechowicz:24CFL}, the \SOAD setting introduces a disconnect between the advice and the robust algorithm (e.g., \PCM); specifically, \ADV furnishes decisions that are supported on the (randomized) metric $\Delta_{\mathcal{S}}$, while \PCM makes decisions on the tree metric given by $(K, \norm{\cdot})$.
Since the \CLIPacro technique effectively ``combines'' \ADV with a robust algorithm, this poses a challenge in the \SOAD setting, introducing a $O(\log n)$ dependency in the consistency bound.\footnote{ Directly applying the \CLIPacro technique to the $(K, \norm{\cdot})$ decision space considered in \PCM yields an unremarkable consistency upper bound of $O(\log n) (1+\varepsilon)$, due to the distortion in the tree metric.}
}
With \CarbonClipper (see \autoref{alg:carbonclipper}), we carefully decouple the ``advice side'' and the ``robust side'' of the \CLIPacro technique to achieve a $(1 + \varepsilon)$-consistency bound.
{\color{teal}
While a $O(\log n)$ factor is likely unavoidable on arbitrary metrics in the adversarial setting of robustness (e.g., as is the case for metrical task systems~\cite{Bubeck:21MTSTrees, Christianson:23MTS}), the non-adversarial setting of consistency (i.e., when advice is correct) implies that such a factor should be avoidable.  Furthermore, removing a factor of $O(\log n)$ allows \CarbonClipper to achieve consistency arbitrarily close to $1$, which is often desirable in practice when the advice is often of high quality.
}

{\color{teal}
To accomplish this decoupling, \CarbonClipper uses the pseudo-cost minimization defined in \eqref{eq:pseudocostCLIP} to generate intermediate ``robust decisions'' ($\mbf{k}_t \in K$) on the tree embedding (see \sref{Def.}{dfn:tree}).
}
These decisions are converted into marginal probability 
distributions on the underlying simplex (i.e., $\mbf{p}_t \in \Delta_{\mathcal{S}}$) \textit{before} evaluating the consistency constraint.
Since \ADV also specifies decisions on $\Delta_{\mathcal{S}}$, this decoupling allows the constraint to directly compare the running cost of \CarbonClipper and \ADV, without losing a $\log(n)$ factor due to the tree embedding.
To hedge against worst-case scenarios that might cause \CarbonClipper to violate the desired $(1 + \varepsilon)$-consistency, the consistency constraint in \eqref{eq:const-constraint} extrapolates the cost of such scenarios on the randomized decision space $\Delta_{\mathcal{S}}$.

\noindent \textbf{Notation. \ } We introduce some shorthand notation to simplify the algorithm's pseudocode as follows: {\color{blue} we let $\CLIP_t$ denote the expected cost of \CarbonClipper's decisions up to time $t$, i.e., $\CLIP_t \coloneqq \sum_{j=1}^{t} f_{j}(\mbf{p}_j) + \mathbb{W}^1( \mbf{p}_{j}, \mbf{p}_{j-1} )$, and similarly let $\ADV_t$ denote the (expected) cost of the advice up to time $t$: $\ADV_t \coloneqq \sum_{j=1}^{t} f_{j}(\mbf{a}_j) + \mathbb{W}^1( \mbf{a}_{j}, \mbf{a}_{j-1} )$.}  As $z^{(t)}$ denotes the utilization of \CarbonClipper, we let $A^{(t)}$ denote the utilization of $\ADV$ at time $t$ (i.e., the expected fraction of the deadline constraint satisfied by \ADV so far).
{\color{blue}
In addition, \CarbonClipper also keeps track of a \textit{robust pseudo-utilization} $\rho^{(t)} \in [0,1]$; this term describes the portion of its decisions thus far that are attributable to the robust pseudo-cost minimization, and we have $\rho^{(t)} \leq z^{(t)}$ for all $t \in [T]$.   This quantity is updated according to the $\overline{\mbf{k}}_t$ that solves \textit{an unconstrained minimization} in \eqref{eq:robustUnconstrained}, ensuring that when \ADV has incurred a ``bad'' service cost that would otherwise not be considered by the robust algorithm, the pseudo-cost $\psi^{(\varepsilon)}$ maintains some headroom to accept better service costs that might arrive in the future.}

\noindent \textbf{Consistency constraint intuition. \ }  Within the constraint \eqref{eq:const-constraint}, \CarbonClipper encodes several ``worst-case'' scenarios that threaten the desired consistency bound.  The first three terms on the left-hand side and the $\ADV_t$ term on the right-hand side consider the actual cost of \CarbonClipper and \ADV so far, along with the current decision under consideration, %
where the expected switching cost is captured by the optimal transport plan with respect to the previous decision.  

The $\mathbb{W}^1(\mbf{p}, \mbf{a}_t)$ term on the left-hand side charges \CarbonClipper in advance for the expected movement cost between it and the advice -- the reasoning for this term is to hedge against the case where the constraint becomes binding in future steps, thus requiring \CarbonClipper to \textit{move} and follow \ADV.  If the constraint did not charge for this potential movement cost in advance, a binding constraint in future time steps might result in either an infeasible problem or a violation of %
$(1 + \varepsilon)$-consistency.
The $\tau c(\mbf{a}_t)$ term on both sides charges both \ADV and \CarbonClipper in advance for the temporal switching cost they must incur before the deadline -- \CarbonClipper is charged according to $\mbf{a}_t$ (as opposed to $\mbf{p}_t$) to continue hedging against the case where it must move to follow the advice in future time steps, finally paying $\tau c(\mbf{a}_t)$ to switch \OFF at the deadline.

On the right-hand side, the  $( 1 - A^{(t)} ) L $ term assumes that $\ADV$ can satisfy the remaining deadline constraint at the best marginal service cost $L$.  
{\color{blue}
In contrast, the final terms on the left-hand side $(1 - z^{(t-1)} - c(\mbf{p}))L + \max( (A^{(t)} - z^{(t-1)} - c(\mbf{p})), \ 0)(U-L)$ balance between two scenarios -- namely, they assume that \CarbonClipper can satisfy a fraction of the remaining constraint (up to $( 1 - A^{(t)} )$) at the best cost by following \ADV, but any excess beyond this (given by $(A^{(t)} - z^{(t)})$), must be fulfilled at the worst service cost $U$, possibly during a mandatory allocation.  

At a high level, \CarbonClipper's constraint on %
$\Delta_{\mathcal{S}}$ combined with the pseudo-cost minimization on $(K, \norm{\cdot})$ generates decisions that are \textit{maximally robust} while preserving consistency.}

\begin{algorithm*}[t]
   \caption{\CarbonClipper (spatiotemporal consistency-limited pseudo-cost minimization) for \SOAD}
   \label{alg:carbonclipper}
{\small
\begin{algorithmic}
   \State {\bfseries input:} Consistency parameter $\varepsilon$, constraint function $c(\cdot)$, pseudo-cost $\psi^{(\varepsilon)}(\cdot)$, starting \OFF state $s \in \mathcal{S}$.
   \State {\bfseries initialize:} $z^{(0)} = 0; \ \rho^{(0)} = 0; \ A^{(0)} = 0; \ \CLIP_0 = 0; \ \ADV_0 = 0; \ \mbf{k}_0 = \Phi \delta_s; \ \mbf{p}_0 = \mbf{a}_0  = \delta_s$.
   \While {cost function $f_t (\cdot)$ is revealed, untrusted advice $\mbf{a}_t$ is revealed, and $z^{(t-1)} < 1$}
   \State Update advice cost $\ADV_t \gets \ADV_{t-1} + f_t(\mbf{a}_t) + \mathbb{W}^1 (\mbf{a}_t, \mbf{a}_{t-1})$ and advice utilization $A^{(t)} \gets A^{(t-1)} + c(\mbf{a}_t)$.
   \State Solve \textbf{constrained} pseudo-cost minimization problem: {\scriptsize \begin{align}
        & \mbf{k}_t = \argmin_{\mbf{k} \in K : \overline{c}(\mbf{k}) \leq 1-z^{(t-1)}} \overline{f}_t(\mbf{k}) + \lVert \mbf{k} - \mbf{k}_{t-1} \rVert_{\ell_1 (\mbf{w})} - \int_{\rho^{(t-1)}}^{\rho^{(t-1)} + \overline{c}(\mbf{k})} \psi^{(\varepsilon)} (u) \ du, \label{eq:pseudocostCLIP} \\
       \begin{split}
       &\text{\small such that} \hspace{0.5em} \mathbf{p} \gets \Phi^{-1} \mbf{k} \hspace{0.5em} \text{\small  and},\\
        &\CLIP_{t-1} + f_t(\mbf{p}) + \mathbb{W}^1 ( \mbf{p}, \mbf{p}_{t-1} ) + \mathbb{W}^1 ( \mbf{p}, \mbf{a}_t ) + \tau c(\mbf{a}_t) + (1 - z^{(t-1)} \kern-0.3em - c(\mbf{p}))L + \max \{ A^{(t)} \kern-0.3em - z^{(t-1)} \kern-0.3em - c(\mbf{p}),  0 \}(U-L)\\ 
       & \hspace{32em} \leq (1+\varepsilon) [\ADV_{t} + \tau c(\mbf{a}_t) + (1 - A^{(t)} ) L].
       \end{split} \label{eq:const-constraint}
    \end{align}}
   \State Update running cost $\CLIP_t \gets \CLIP_{t-1} + f_t(\mbf{p}_t) + \mathbb{W}^1 ( \mbf{p}_t , \mbf{p}_{t-1} )$ and utilization $z^{(t)} \gets z^{(t-1)} + c(\mbf{p}_t)$.
   \State Solve \textbf{unconstrained} pseudo-cost minimization problem: {\scriptsize \begin{align}
       \Tilde{\mbf{k}}_t &= \argmin_{\mbf{k} \in K : \overline{c}(\mbf{k}) \leq 1-z^{(t-1)}} \overline{f}_t(\mbf{k}) + \lVert \mbf{k} - \mbf{k}_{t-1} \rVert_{\ell_1 (\mbf{w})} - \int_{\rho^{(t-1)}}^{\rho^{(t-1)} + \overline{c}(\mbf{k})} \psi^{(\varepsilon)} (u) \ du \label{eq:robustUnconstrained}
    \end{align}}
    \State Update \textit{robust pseudo-utilization} $\rho^{(t)} \gets \rho^{(t-1)} + \min \{  \overline{c}(\Tilde{\mbf{k}}_t), c(\mbf{p}_t) \}$.
   \EndWhile
\end{algorithmic}}
\end{algorithm*}

\subsection{Main results}

{\color{blue}
In \autoref{thm:constrobCLIP}, we give upper bounds on the robustness and consistency of \CarbonClipper.  
}

\begin{thm} \label{thm:constrobCLIP}
For any $\varepsilon \in (0, \eta - 1]$, \emph{\CarbonClipper} is $(1+\varepsilon)$-consistent and $O(\log n) \gamma^{(\varepsilon)}$-robust for \SOAD, where $\gamma^{(\varepsilon)}$ is the solution to \eqref{eq:gamma}.
\end{thm}

{\color{blue}
\noindent Furthermore, %
we give a lower bound on the best achievable robustness ratio for any %
$(1+\varepsilon)$-consistent algorithm, using a construction of a challenging metric space and service cost sequence -- since robustness and consistency are defined over all valid inputs (i.e., based on competitive ratio), this result characterizes the \textit{optimal robustness-consistency trade-off}, and implies that \CarbonClipper matches the optimal up to $\log$ factors that are due to the metric embedding.
}

\begin{thm} \label{thm:optimalconstrobCSWM}
Given untrusted advice \emph{\ADV} and $\varepsilon \in (0, \eta - 1]$, any $(1 + \varepsilon)$-consistent learning-augmented algorithm for \SOAD is at least $\gamma^{(\varepsilon)}$-robust, where $\gamma^{(\varepsilon)}$ is defined in \eqref{eq:gamma}.
\end{thm}

\noindent Learning-augmentation and robustness-consistency trade-offs have been considered in both $\MTS$ and $\OWT$ -- we briefly review how \autoref{thm:optimalconstrobCSWM} compares. 
{\color{blue}
For $\MTS$, \citet{Christianson:23MTS} show that for $\varepsilon \in (0, 1]$, any $(1+\varepsilon)$-consistent algorithm must be $2^{\Omega(\nicefrac{1}{\varepsilon})}$-robust.
While optimal trade-offs for the minimization variant of \OWT have not been studied, \citet{SunLee:21} show that %
any $\gamma$-robust algorithm must be $\theta / \left[ \nicefrac{\theta}{\gamma} + (\theta - 1) (1 - \nicefrac{1}{\gamma} \ln ( \nicefrac{\theta - 1}{\gamma-1} ) ) \right]$-consistent in the maximization case, where $\theta = \nicefrac{U}{L}$ is the price bound ratio.  
}
While these bounds are not directly comparable, it is notable that the extra structure of \SOAD allows it to avoid the unbounded exponential robustness of \MTS.%

\subsection{Proof overviews}
{\color{blue}
We now give proof sketches of \sref{Theorems}{thm:constrobCLIP} and \ref{thm:optimalconstrobCSWM}, relegating the full proofs to \autoref{apx:clip}.
}

\smallskip
\noindent \textbf{Proof Sketch of \autoref{thm:constrobCLIP}.}
We separately consider consistency and robustness in turn.  
\vspace{-0.4em}
\begin{lem} \label{lem:clipconst}
{\color{blue}
    \emph{\CarbonClipper} is \emph{$(1+\varepsilon)$}-consistent when the advice is correct, i.e., \emph{$\ADV(\mathcal{I}) = \OPT(\mathcal{I})$}.
}
\end{lem}
\vspace{-0.4em}
\noindent For consistency, recall that the constraint enforces that the expected cost of \CarbonClipper \textit{thus far} at time $j$ (i.e., before mandatory allocation) satisfies \eqref{eq:const-constraint}.
Since this constraint holds for all steps before the mandatory allocation, we must resolve the cost \textit{during the mandatory allocation}.  We characterize two worst-case scenarios based on whether \CarbonClipper has completed \textit{less} (\textbf{Case 1}, see \eqref{eq:case1consistency}) or \textit{more} (\textbf{Case 2}, see \eqref{eq:case2consistency}) of the deadline constraint compared to \ADV.
In either of these cases, \eqref{eq:case1consistency} and \eqref{eq:case2consistency} show that replacing the ``hedging terms'' that follow $\CLIP_{j-1}$ and $\ADV_{j-1}$ in the constraint with worst-case service and movement costs yields a consistency ratio that is $\le (1 + \varepsilon)$.
\vspace{-0.4em}
\begin{lem} \label{lem:cliprob}
    {\color{teal} \emph{\CarbonClipper} is $O(\log n) \gamma^{(\varepsilon)}$-robust, where $\gamma^{(\varepsilon)}$ is defined in \eqref{eq:gamma}.}
\end{lem}
\vspace{-0.4em}
{\color{blue}
\noindent For robustness, we define two cases that characterize ``bad'' advice, namely ``inactive'' advice that forces mandatory allocation (\textbf{Case 1}, see \eqref{eq:case1robustness}), and ``overactive'' advice that incurs sub-optimal cost (\textbf{Case 2}, see \eqref{eq:case2robustness}).}
For each of these, we derive bounds on the portion of \CarbonClipper's expected solution that is \textit{allowed} to come from the pseudo-cost minimization without violating consistency.

{\color{blue}
In Case 1, \CarbonClipper assumes that \ADV can satisfy the constraint at the best possible service cost $L$, so we derive an \textit{upper bound} describing the maximum utilization achievable via the pseudo-cost minimization before the mandatory allocation (see \sref{Prop.}{pro:zpcm}).}  In Case 2, \CarbonClipper must follow \ADV to avoid violating consistency, even if \ADV incurs sub-optimal cost -- we derive a \textit{lower bound} on the amount of utilization that \CarbonClipper must ``spend'' while continually satisfying the $(1+\varepsilon)$-consistency constraint (see \sref{Prop.}{pro:zadv}).  These characterizations enable pseudo-cost proof techniques (e.g., as in \sref{Thm.}{thm:etaCompPCM}) that show $O(\log n) \gamma^{(\varepsilon)}$-robustness in each case. \hfill \qed

\bigskip
\noindent \textbf{Proof Sketch of \autoref{thm:optimalconstrobCSWM}.}
In \sref{Definition}{dfn:yadversary-constrob}, we define a slight variant of the special metric and $y$-adversary construction from \sref{Thm.}{thm:lowerboundCSWM}, denoted by $\mathcal{A}'_y$.  Informally, $\mathcal{A}'_y$ presents ``good'' cost functions at distant points, before eventually presenting just the best service cost functions (i.e., $y$) at the starting point.
We consider two types of advice that each capture consistency and robustness, respectively.  In this setting, bad advice completes none of the deadline constraint before the mandatory allocation, while good advice makes the exact decisions that recover $\OPT(\mathcal{A}'_y)$.

{\color{blue}
Using the proof of \autoref{thm:lowerboundCSWM}, we characterize the cost of a learning-augmented algorithm \ALG according to two arbitrary \textit{constraint satisfaction functions} $s(y), t(y) : [L,U] \to [0,1]$ (see \eqref{eq:lb-costconstrob}). 
Conditioned on the advice that \ALG receives, any $\alpha$-consistent and $\gamma$-robust $\ALG$ must satisfy two conditions, where the robustness condition follows from the proof of \autoref{thm:lowerboundCSWM} (see \eqref{eq:rob-bound-gron-min}), and the consistency condition is given by $\gamma \int_{\mbf{\nicefrac{U}{\gamma}}}^L \ln \left( \frac{U - u - D + 2\tau}{ U - \nicefrac{U}{\gamma} - D - 2\tau } \right)  du + \left[ 2D + 2\tau \right] \left[ \gamma \ln \left( \frac{U - L - D + 2\tau}{ U - \nicefrac{U}{\gamma} - D - 2\tau} \right)  \right] \leq \alpha L - L$ (see \eqref{eq:const-bound-gron-min}).}  Substituting $\alpha \coloneqq (1+ \varepsilon)$ and binding the inequality above yields the result. \hfill \qed

\section{Generalization to Time-Varying Metrics}
\label{sec:timevarying}
{\color{blue}
Before moving to our case study, we present a generalization of the results in \sref{Sections}{sec:pcm} and \ref{sec:clip} to settings with \textit{time-varying metrics}.  
This is motivated by the applications of \SOAD (see \autoref{sec:examples}) since in practice, the distance between points in the metric (e.g., network delays, transit costs) may not be constant.
The extension to time-varying metrics is straightforward, and we present corollaries for both \PCM and \CarbonClipper after formalizing the extension of \SOAD that we consider.
}

\subsection{\SOAD with time-varying distances (\SOADt)}
In \SOAD with time-varying distances, we let $d_t(\cdot, \ \cdot) : t \in [T]$ denote a \textit{time-varying distance function} between points in $X$, and we assume that an online algorithm \ALG is always able to observe the current distance $d_t$ at time $t$.  Additionally, we redefine $D$ to be an upper bound on the normalized spatial distance between any two points in $X$ over the entire time horizon $T$, namely $D = \sup_{t \in [T]} \left( \max_{u, v \in X : u \not = v} \frac{d_t( u, \ v )}{\min \{ c^{(u)}, c^{(v)} \} } \right)$.
{\color{blue}
Although distances between the locations of $X$ are time varying, \SOADt assumes that the temporal switching cost between \ON and \OFF states at a single point $u \in X$ is constant (i.e., $\Vert \cdot \Vert_{\ell_1(\boldsymbol{\beta})}$ is \textit{not} time-varying) for simplicity of presentation.
}

\subsection{Main results}
In the following results, we show that our robust algorithm \PCM (see \autoref{alg:pcm}), and our learning-augmented algorithm \CarbonClipper (see \autoref{alg:carbonclipper}) are both sufficiently flexible to provide guarantees in \SOADt with minimal changes. {\color{blue} Note that the lower bounds in the time-invariant setting still apply to the time-varying setting (e.g., by setting $d_t$ constant for all $t \in [T]$).}

First, as a corollary to \autoref{thm:etaCompPCM}, we show that \PCM retains its $O( \log n ) \eta$ competitive bound in the setting of \SOADt.  We state the result here and give the full proof in \autoref{apx:timevaryPCM}.
\begin{cor}\label{cor:timevaryPCM}
\emph{\PCM} is $O(\log n ) \eta$-competitive for \SOADt, where $\eta$ is given by \eqref{eq:eta}.
\end{cor}
Furthermore, as a corollary to \autoref{thm:constrobCLIP}, we show that \CarbonClipper's consistency-robustness bound also holds for the time-varying setting of \SOADt when just one term is swapped within the consistency constraint.  We state the result here and give the full proof in \autoref{apx:timevaryCLIP}.
\begin{cor}\label{cor:timevaryCLIP}
With a minor change to the consistency constraint, \emph{\CarbonClipper} is $(1+\varepsilon)$-consistent and $O(\log n) \gamma^{(\varepsilon)}$-robust for \SOADt, where $\gamma^{(\varepsilon)}$ is the solution to \eqref{eq:gamma}.
\end{cor}

\section{Case Study: Carbon-aware Workload Management in Data Centers}
\label{sec:exp}
{\color{blue}
We end the paper with a case study application of \SOAD and our algorithms for the motivating application of carbon-aware workload management on a simulated global network of data centers.
}

\subsection{Experimental setup}\label{sec:expsetup}
{\color{blue}
We simulate a carbon-aware scheduler that schedules a delay-tolerant batch job on a network of data centers. %
We simulate a global network of data centers based on measurements between Amazon Web Services (AWS) regions.  %
We construct \CSWM instances as follows: we generate a job with length $J$ (in hours), an arrival time (rounded to the nearest hour), and a ``data size'' $G$, where $G$ gives the amount of data (in GB) to be transferred while migrating the job.  
The task is to finish the job before the deadline $T$ while minimizing total CO$_2$ emissions, which are a function of the scheduling decisions and the carbon intensity at each time step and region.
}

\noindent \textbf{AWS measurement data.}
{\color{blue}
We pick 14 AWS regions~\cite{aws} based on available carbon data (see \autoref{tab:characteristics} in the Appendix). %
Among these regions, we collect 72,900 pairwise measurements of latency and throughput, %
compute the mean and variance, %
and sample a latency matrix.  
To model migration overhead, 
we scale the data transferred (and corresponding latency) to match $G$.  %
These values are scaled by carbon data to define a distance metric on the regions in terms of \textit{CO$_2$ overhead}.%
}
{\color{blue}
To model network heterogeneity, %
we set a parameter $\kappa \in [0,1]$ to adjust the simulated energy of the network.  $\kappa$ is a ratio -- if $\kappa = 0.5$, a minute of data transfer from machine(s) in one region to machine(s) in another uses half as much energy as executing at the full allocation (i.e., $\smash{x_t^{\ON^{(\cdot)}}} = 1$) for one minute.
}

\noindent \textbf{Carbon data traces. \ }
{\color{blue}
We obtain %
hourly carbon intensity data for each region, expressed as grams of CO$_2$ equivalent per kilowatt-hour. } %
In the main body, we consider \textit{average carbon intensity}~\cite{electricity-map}, which gives the average emissions of all electricity generated on a grid at a certain time; 
{\color{blue}
this data spans 2020-2022 and includes all regions.  
In \autoref{apx:expmarginal}, we also consider \textit{marginal carbon intensity}~\cite{watttime}; this signal is available for 9 regions in 2022, and also includes proprietary forecasts.  %
}

We use the latency of moving data between regions to calculate a \textit{CO$_2$ overhead} for the metric $(X, d)$ (latency $\times$ energy $\times$ carbon intensity).  In most cases, we approximate the network's carbon intensity by the average across regions.  
{\color{blue}
When specified, we introduce variation by resampling the carbon intensity of up to $\Upsilon \in [0, n^2]$ links each time step. We henceforth call $\Upsilon$ a \textit{volatility factor}; resampling assigns a new random carbon intensity (within $[L, U]$) to a link between two regions.
}

\noindent \textbf{Cloud job traces. \ }
We use Google cluster traces~\cite{Reiss:12} that provide a real distribution of job lengths.  We normalize this distribution such that the maximum length is $12$ hours -- each job's length $J$ is drawn from the distribution and rounded up to the next integer, so $J$ falls in the range $\{ 1, ..., 12 \}$.

\noindent \textbf{Forecasts. \ }
{\color{blue}
We generate forecasts of the carbon intensity for each location and time.  These forecasts are used to solve for a predicted optimal solution that assumes they are correct, which becomes black-box \ADV for \CarbonClipper.
For the average carbon intensity signal, we generate synthetic forecasts by combining true data with random noise.\footnote{We use an open-source ML model that provides carbon intensity forecasts for U.S. regions~\cite{Maji:22:CC} to tune the magnitude of random noise such that $\ADV$'s empirical competitive ratio is slightly worse than an $\ADV$ that uses the ML forecasts.}
Letting $\smash{\text{Carbon}_t^{(u)}}$ denote the carbon intensity at data center $u$ and time $t$, our synthetic forecast is given by $\smash{\text{Pred}_t^{(u)}} = 0.6 \cdot \smash{\text{Carbon}_t^{(u)}} + 0.4 \cdot \text{Unif}(L, U)$.
}
{\color{blue}
To test \CarbonClipper's robustness, Experiment V directly manipulates \ADV. 
}
We set an \textit{adversarial factor} $\xi \in [0,1]$, where $\xi = 0$ implies $\ADV$ is correct.  We use a solver on true data to obtain two solutions, where one is given a \textit{flipped objective} (i.e., it maximizes carbon emissions).  
{\color{blue}
Letting $\{ \mbf{x}_t^\star \}_{t \in [T]}$ denote the decisions of \OPT and ${\{ \breve{\mbf{x}_t} \}}_{t \in [T]}$ denote the decisions of the maximization solution, we have $\ADV \coloneqq \{ (1-\xi) \mbf{x}_t^\star + \xi \breve{\mbf{x}_t}\}_{t \in [T]}$.  We note that although this is unrealistic in practice, manipulating \ADV directly allows us to to quantify the sensitivity of \CarbonClipper against all sources of error.
}

\noindent \textbf{Setup details. \ }
{\color{blue}
We simulate 1,500 jobs for each configuration.  Each job's arrival region and arrival time is uniformly random across all active regions and times.
Each job's deadline $T$ and data size $G$ are either fixed or drawn from a distribution, and this is specified. 
To set the parameters $L$ and $U$, %
we examine the preceding month of carbon intensities (in all regions) leading up to the arrival time and set $L$ and $U$ according to the minimum and maximum, respectively. 
We set the following defaults (i.e., unless otherwise specified):  The metric covers all 14 regions.  Each job's length is drawn from the Google traces as above. The temporal switching coefficient $\tau$ is set to $1$, the network energy factor $\kappa$ is set to $0.5$, and the volatility factor $\Upsilon$ is set to $0$ (i.e., the network is stable). 
}

\noindent \textbf{Benchmark algorithms. \ } %
We compute the offline optimal solution for each instance using CVXPY~\cite{CVXPY}.  %
{\color{blue}
We compare \CarbonClipper and \PCM against four baselines adapted from literature.  
The first is a \textbf{carbon-agnostic} approach that runs the job whenever it is submitted without migration, simulating the behavior of a non-carbon-aware scheduler.
}
{\color{blue}
We also consider two \textit{greedy baselines} that use simple decision rules.
The first of these is a \textbf{greedy} policy that examines the current carbon intensity across all regions at the arrival time, migrates to the ``greenest'' region (i.e., with lowest carbon intensity), and runs the full job.
This captures an observation~\cite{sukprasert2023quantifying} that one migration to a consistently low-carbon region yields most of the benefits of spatiotemporal shifting.  
We also consider a policy that we term \textbf{delayed greedy}, which examines the \textit{full forecast} across all regions, migrating to start the job at the ``best region and time'' (i.e., slot with lowest predicted carbon anywhere).  If there is not enough time to finish the job after the identified slot, it is scheduled to start as close to it as possible.
}
The final baseline is a \textbf{simple threshold}-based approach from temporal shifting literature~\cite{Lechowicz:23, Bostandoost:24:HotCarbon}; it sets a threshold $\sqrt{UL}$, based on prior work in online search~\cite{ElYaniv:01}.  At each time step, it %
runs the job in the best region whose carbon intensity is $\leq \sqrt{UL}$,  
{\color{blue}
without considering migration overheads.  If no regions are $\leq \sqrt{UL}$ at a particular time, the job is checkpointed in place, and a \textit{mandatory allocation} happens when approaching the deadline if the job is not finished.}
\vspace{-0.5em}

\begin{figure*}[t]
    \vspace{-0.5em}
    \minipage{\textwidth}
    \includegraphics[width=\linewidth]{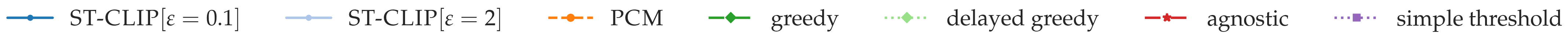}\vspace{0.1em}
    \endminipage\hfill\\
    \minipage{0.32\textwidth}
    \includegraphics[width=\linewidth]{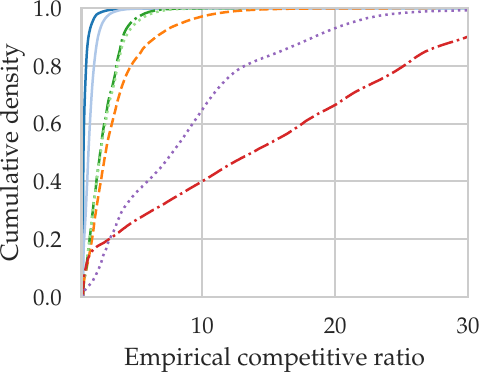}\vspace{-1em}
    \caption{ CDFs of competitive ratios for each algorithm, across all average carbon intensity experiments.}\label{fig:cdf}
    \endminipage\hfill
    \minipage{0.32\textwidth}
    \includegraphics[width=\linewidth]{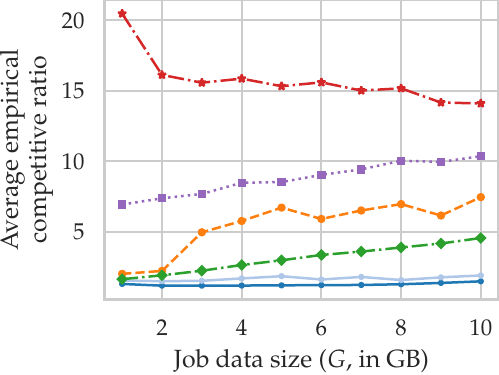}\vspace{-1em}
    \caption{ Average empirical competitive ratios for varying \textit{job data size $G$}, with $T \thicksim \text{Unif}_{\mathbb{Z}}(12, 48)$. }\label{fig:gb}
    \endminipage\hfill
    \minipage{0.32\textwidth}
    \includegraphics[width=\linewidth]{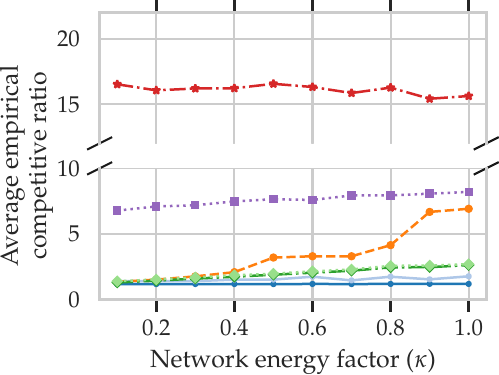}\vspace{-1em}
     \caption{ Average empirical competitive ratios for varying \textit{energy factor $\kappa$}, with $G = 4, T \thicksim \text{Unif}_{\mathbb{Z}}(12, 48).$ }\label{fig:energy}
    \endminipage\hfill  
\end{figure*}

\subsection{Experimental results} \label{sec:expresults}
{\color{blue}
We highlight several experiments here, referring to \autoref{apx:exp} for the extended set.
A summary is given in \autoref{fig:cdf}, where we plot a cumulative distribution function (CDF) of the empirical competitive ratio for all tested algorithms in Expts. I-IV and VI-VIII.  %
Given imperfect advice, $\CarbonClipper[\varepsilon = 2]$ significantly outperforms the baselines, improving on greedy, delayed greedy, simple threshold, and carbon-agnostic by averages of 32.1\%, 33.5\%, 79.4\%, and 88.7\%, respectively.}
In Expts. I-III, each job's deadline is a random integer between $12$ and $48$ (denoted by $T \thicksim \text{Unif}_{\mathbb{Z}} (12, 48)$).  In these experiments, both greedy policies outperform our robust algorithm, \PCM.  
{\color{blue}
This result aligns with prior findings~\cite{sukprasert2023quantifying}; since these experiments consider all 14 regions, there are consistent low-carbon grids in the mix that give an advantage to the greedy policies.
In Expt. IV, we examine this further, showing that the performance of greedy policies can degrade in realistic situations.}

\noindent\textbf{Experiment I: } \textit{Effect of job data size $G$. \ }
{\color{blue}
In \autoref{fig:gb}, we plot the average empirical competitive ratio for job data sizes $G \in \{ 1, \dots, 10 \}$. %
}
Recall that parameter $D$ depends on the diameter of the metric space (i.e., the worst migration overhead between regions); as $G$ increases, this maximum overhead grows.  As predicted by the theoretical bounds, \PCM's performance degrades as $G$ grows; we observe the same effect for the greedy policies and simple threshold.  {\color{blue} Since it can leverage advice, \CarbonClipper maintains consistent performance for many settings of $G$.}

\noindent\textbf{Experiment II: } \textit{Effect of network energy scale $\kappa$. \ }
\autoref{fig:energy} plots the average empirical competitive ratio for $\kappa \in [0.1, 1]$, fixing $G = 4$.  
{\color{blue}
As in Expt. I, $\kappa$ affects the parameter $D$ -- thus, the performance of \PCM degrades slightly as $\kappa$ grows.
When $\kappa$ is small, greedy policies perform nearly as well as \CarbonClipper, though they degrade as $\kappa$ increases; \CarbonClipper's usage of advice yields consistent performance.
}

\noindent\textbf{Experiment III: } \textit{Effect of volatility factor $\Upsilon$. \ }
{\color{blue}
In \autoref{fig:volatility}, we plot the average empirical competitive ratio for $\Upsilon \in [28, 196]$, fixing $G = 4$. }  Aligning with the theoretical results (\sref{Corollary}{cor:timevaryPCM} \& \ref{cor:timevaryCLIP}), we find that \PCM and \CarbonClipper's performance is robust to this volatility.  Both of the greedy policies do not consider the migration overhead and only migrate once, so their performance is consistent.  

\noindent\textbf{Experiment IV: } \textit{Effect of electric grids and data center availability. \ }
{\color{blue}
Greedy policies do well in Expts. I-III,  %
where some regions have consistently low-carbon grids.\footnote{\autoref{fig:carbon_traces} in the Appendix plots a sample of carbon intensity data for all 14 regions to motivate this visually.}
}
\begin{figure*}[t]
    \vspace{-0.5em}
    \minipage{\textwidth}
    \includegraphics[width=\linewidth]{plots/legend.png}\vspace{0.1em}
    \endminipage\hfill\\
     \minipage{0.32\textwidth}
     \includegraphics[width=\linewidth]{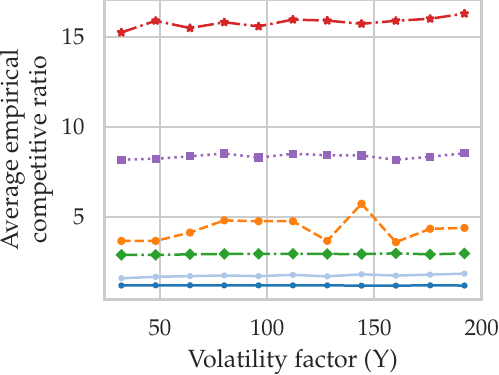}\vspace{-1em}
     \caption{ Average empirical competitive ratios for varying \textit{volatility factor $\Upsilon$}. \  $G = 4, T \thicksim \text{Unif}_{\mathbb{Z}}(12, 48).$ }\label{fig:volatility}
     \endminipage\hfill  
    \minipage{0.32\textwidth}
     \includegraphics[width=\linewidth]{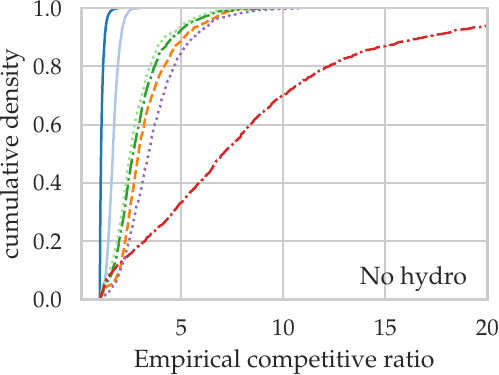}\vspace{-1em}
     \caption{ CDFs of empirical competitive ratios for each tested algorithm on the ``no hydro'' region subset.  }\label{fig:subsetNoHydro}
     \endminipage\hfill
     \minipage{0.32\textwidth}
     \includegraphics[width=\linewidth]{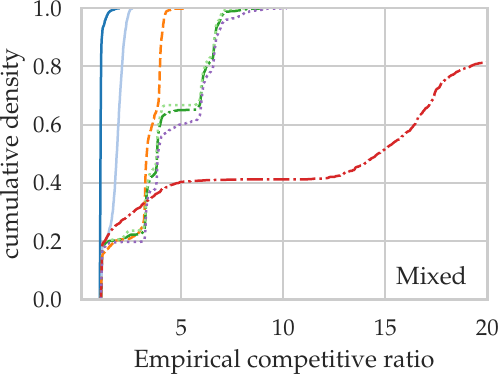}\vspace{-1em}
     \caption{ CDFs of empirical competitive ratios for each tested algorithm on the ``mixed'' region subset.  }\label{fig:subsetMixed}
     \endminipage\hfill
\end{figure*}
{\color{blue}
In practice, a greedy policy may face obstacles if it is unable to migrate to low-carbon regions.  For instance, such regions might reach capacity,
\textit{removing them} as migration options.  Jobs may also be \textit{restricted} from leaving a region due to regulations~\cite{GDPR:16}.  
Further, consistently low-carbon grids often leverage hydroelectric or nuclear sources that are difficult to build at scale %
compared to cheaper renewables~\cite{NREL:24}.  
This is important because it suggests future grids will moreso resemble those marked by renewable intermittency.}
\begin{wrapfigure}{r}{0.32\textwidth}
  \includegraphics[width=0.32\textwidth]{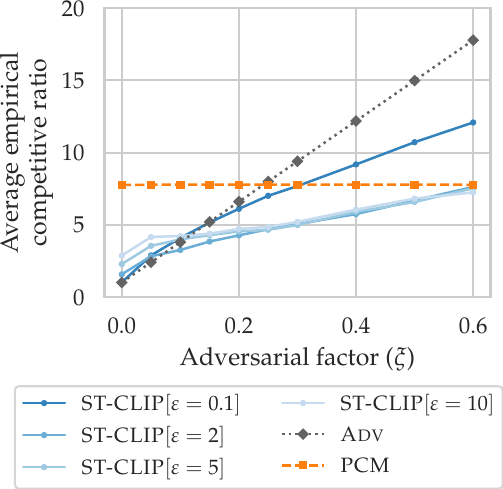}\vspace{-0.6em}
  \caption{ Average empirical competitive ratio for varying \textit{adversarial factor $\xi$}, with $G = 4, T = 12$.  $\xi \to 0$ implies that $\ADV \to \OPT$, and as $\xi$ grows, $\ADV$ degrades.}\label{fig:xi}
\end{wrapfigure}
\noindent In \autoref{fig:subsetNoHydro} and \ref{fig:subsetMixed}, we present results where the metric is %
a \textit{subset} of the 14 regions, giving results on more subsets in \autoref{apx:exp}.  By considering these subsets, we approximate issues discussed above (e.g., data center congestion, grid characteristics).  %
\autoref{fig:subsetNoHydro} considers a ``no hydroelectric'' subset that omits Sweden and Quebec.  Under this subset, \PCM closes the gap with the greedy policies, with an average competitive ratio that is within 4.24\% of both.  \autoref{fig:subsetMixed} considers a smaller subset of $5$ regions: South Korea, Virginia, Sydney, Quebec, and France.  On this mixed set of grids, \PCM outperforms greedy and delayed greedy by 30.91\% and 28.79\%, respectively.  
{\color{teal}
The results above highlight that situations do arise where greedy policies perform worse than both \CarbonClipper and \PCM.  However, such situations are not a majority -- out of the 14 random subsets that we tested, \PCM outperformed the greedy policy in \textit{four} subsets, including the ``mixed'' (\autoref{fig:subsetMixed}) and ``mixed 2'' (\autoref{fig:subsetApx3}) subsets. \PCM is relatively conservative in its decisions, being optimized for worst-case (adversarial) inputs -- this advantages greedy policies on the majority of instances that do not benefit from the ``worst-case hedging'' behavior that \PCM exhibits.
}

\noindent\textbf{Experiment V: } \textit{Effect of bad black-box advice $\xi$. \ }
{\color{blue}
\autoref{fig:xi} plots the effect of bad black-box advice on \CarbonClipper's performance. We test values of $\xi \in [0,0.6]$, %
generating \ADV according to the technique discussed in \autoref{sec:expsetup}.  
\CarbonClipper is initialized with $\varepsilon \in [0.1, 2, 5, 10]$, where $\varepsilon \to 0$ implies that it follows $\ADV$ more closely.  
We also plot the empirical competitive ratio of $\PCM$ and $\ADV$ as opposing baselines.  
}
We find that \CarbonClipper nearly matches $\ADV$ when it is correct, while degrading more gracefully as $\ADV$'s cost increases.  This result shows that \CarbonClipper is empirically robust to even adversarial black-box advice.

\section{Conclusion}
\label{sec:conclusion}
Motivated by sustainability applications,
we introduce and study spatiotemporal online allocation with deadlines (\SOAD), the first online problem that combines general metrics with deadline constraints, bridging the gap between existing literature on metrical task systems and online search.  Our main results present \PCM as a competitive algorithm for \SOAD and \CarbonClipper, a learning-augmented algorithm that achieves the optimal robustness-consistency trade-off.  We evaluate our proposed algorithms in a case study of carbon-aware workload management in data centers.
A number of questions remain for future work, including natural extensions motivated by applications.
For computing applications, \SOAD may be extended to model resource contention and/or delayed access to resources, particularly after moving the allocation to a new point.  
{\color{teal}
Similarly, an extension to model \textit{multiple} workloads with different deadlines would be natural (e.g., scheduling multiple batch jobs, dynamic job arrivals/departures). 
}
Our theoretical results contend with substantial generality (i.e., in the metric); it would be interesting to explore whether improved results can be obtained under a more structured setting.

\begin{acks}
We thank our shepherd Izzy Grosof and the anonymous SIGMETRICS reviewers for their valuable insight and feedback.

This material is based upon work supported by the U.S. Department of Energy, Office of Science, Office of Advanced Scientific Computing Research, Department of Energy Computational Science Graduate Fellowship under Award Number DE-SC0024386, an NSF Graduate Research Fellowship (DGE-1745301), National Science Foundation grants CAREER-2045641, CNS-2102963, CNS-2106299, CNS-2146814, CPS-2136197, CPS-2136199, NGSDI-2105494, NGSDI-2105648, CNS-2325956, 2020888, 2021693, 2045641, 2213636, 2211888, and support from VMWare and Adobe.
\end{acks}

\section*{Disclaimers}
This report was prepared as an account of work sponsored by an agency of the United States Government. Neither the U.S. Government nor any agency thereof, nor any of their employees, makes any warranty, express or implied, or assumes any legal liability or responsibility for the accuracy, completeness, or usefulness of any information, apparatus, product, or process disclosed, or represents that its use would not infringe privately owned rights. Reference herein to any specific commercial product, process, or service by trade name, trademark, manufacturer, or otherwise does not necessarily constitute or imply its endorsement, recommendation, or favoring by the U.S. Government or any agency thereof. The views and opinions of authors expressed herein do not necessarily state or reflect those of the U.S. Government or any agency thereof.

\received{October 2024}
\received[revised]{January 2025}
\received[accepted]{January 2025}

\bibliographystyle{ACM-Reference-Format}
\bibliography{main}

\appendix
\section*{Appendix}

\begin{table}[h]
{\color{blue}
	\caption{A summary of key notations.}
 \vspace{-3mm}
	\label{tab:notations}
 	\small
	\begin{center}
		\begin{tabular}[P]{|c|l|}
            \hline
            \textbf{Notation} & \textbf{Description} \\
            \hline
            $t \in [T]$ & Time step index \\
            \hline
            $\mathcal{X}$ & Feasible set for vector allocation decisions \\
            \hline
            $\mbf{x}_t \in \mathcal{X}$ & Allocation decision at time $t$ \\
            \hline
            $f_t(\cdot) : \mathcal{X} \to \mathbb{R}$ & (\textit{Online input}) Service cost function revealed to the player at time $t$ \\
            \hline 
            $c(\cdot) : \mathcal{X} \to [0,1]$ & Constraint function; describes the fraction satisfied by an allocation \\
			\hline
            $d(u, v) : u,v \in X \to \mathbb{R}$ & Spatial distance in the metric $(X,d)$ \\
			\hline
            $\Vert \mbf{x} - \mbf{x}' \Vert_{\ell_1(\beta)} : \mbf{x}, \mbf{x}' \in \mathcal{X} \to \mathbb{R}$ & Switching costs between \ON and \OFF allocations \\
			\hline
            $g( \cdot, \ \cdot ) \coloneqq d( \cdot, \ \cdot ) + \Vert \cdot \Vert_{\ell_1(\beta)}$ & Combined movement \& switching cost between points and allocations \\
			\hline
            \hline
            $u \in X$ & Point $u$ in an $n$-point metric space $(X, d)$\\
            \hline
            $\ON^{(u)}, \OFF^{(u)}$ & \ON state and \OFF state at point $u$, respectively \\
            \hline
            $x_t^{\ON^{(u)}}, x_t^{\OFF^{(u)}} \in [0,1]$ & Fractional allocations to \ON \ / \OFF states at point $u$ at time $t$ \\
            \hline
            $c^{(u)} \in (0,1] $ & Throughput coefficient; describes constraint satisfied by $x^{\ON^{(u)}} = 1$ \\
            \hline
            $f_t^{(u)} \in [c^{(u)}L, c^{(u)}U]$ & Service cost coefficient at $\ON^{(u)}$ \& time $t$; proportional to $L > 0$ and $U$.  \\
            \hline
            $\beta^{(u)} > 0$ & Switching coefficient; describes switching cost between $\ON^{(u)} \leftrightarrow \OFF^{(u)}$ \\
            \hline
            $\tau : \beta^{(u)} \le \tau c^{(u)}  \ \forall u \in X$ & Upper bound on normalized switching coefficient \\
            \hline
            $D =  \sup_{u, v \in X} \frac{d( u, \ v )}{\min \{ c^{(u)}, c^{(v)} \} }$ & Upper bound on normalized spatial distance between any two points \\
            \hline
            \hline
            $\mathcal{S}$ & Discrete set of all \ON and \OFF states \\
            \hline
            $\Delta_{\mathcal{S}}$ & Probability measure over $\mathcal{S}$ \\
            \hline
            $\mbf{p}_t \in \Delta_{\mathcal{S}}$ & Probability distribution (\& corresponding random allocation) at time $t$\\
            \hline
            $\mathbb{W}_1(\mbf{p}, \mbf{p}') : \mbf{p}, \mbf{p}' \in \Delta_{\mathcal{S}} \to \mathbb{R}$ & Optimal transport distance between distributions (in terms of $g( \cdot, \ \cdot)$) \\
            \hline
            $\delta_s \in \Delta_{\mathcal{S}}$ & Dirac measure supported at $\OFF^{(s)}$ \\
            \hline
            \hline
            $\mathcal{T} = (V,E)$ & Hierarchically separated tree (HST) constructed by \sref{Def.}{dfn:tree}\\
            \hline
            $K \subset \mathbb{R}^{\vert V \vert}$ & Vector space corresponding to $\mathcal{T}$ (see \sref{Def.}{dfn:vs})\\
            \hline
            $\mbf{k}_t \in K$ & Vector decision (\& corresponding prob. distribution) at time $t$\\
            \hline
            $\Phi \in \mathbb{R}^{\vert V \vert \times 2n}$ & Linear map such that $\Phi \mbf{p} \in K \text{ and } \Phi^{-1} \mbf{k} \in \Delta_{\mathcal{S}}$ \\
            \hline
            $\norm{\cdot} : K \to \mathbb{R}$ & Weighted $\ell_1$ norm that recovers optimal transport distances in $\mathcal{T}$ \\
            \hline
            \begin{tabular}{@{}c@{}}$\overline{f}_t(\mbf{k}) = f_t(\Phi^{-1} \mbf{k})$ \\ $\overline{c}(\mbf{k}) = c(\Phi^{-1} \mbf{k})$\end{tabular} & Notation shorthand for functions defined on vector space $K$ \\
            \hline
            \hline
            $z^{(t)} \in [0,1]$ & Utilization; fraction of constraint satisfied in expectation up to time $t$ \\
            \hline
            \hline
            $\ADV(\mathcal{I}) \coloneqq \{ \mbf{a}_t \in \Delta_{\mathcal{S}} \}_{t \in [T]}$ & Black-box advice provided to \CarbonClipper (see \sref{Def.}{dfn:adv})\\
            \hline
            $A^{(t)} \in [0,1]$ & \ADV utilization; fraction of constraint satisfied by $c(\mbf{a}_1) + \dots + c(\mbf{a}_t)$\\
            \hline
		\end{tabular}
	\end{center}
 }
\end{table}

{\color{blue}
\section{Deferred Examples}\label{apx:examples}

In this section, we detail two more examples of applications that motivate the \SOAD problem introduced in the main body, picking up from \autoref{sec:examples}.  

\smallskip
\noindent\textbf{Carbon-aware or cost-aware autonomous electric vehicle charging. \ } Consider an autonomous electric vehicle taxi (AEV) servicing a city~\cite{Waymo} with multiple charging stations.  Suppose that by the end of a day (i.e., deadline $T$), the AEV must replenish the charge that it will have used throughout the day.  Service costs $f_t^{(u)}$ can represent either the carbon emissions of charging at location $u$ during time slot $t$, or the charging cost plus opportunity cost of charging at location $u$ during time slot $t$.  
We note that even within a single city, the \textit{locational marginal emissions} (i.e., the carbon intensity of electricity at a specific location) may vary significantly~\cite{Brown:2020:SpatiotemporalLMP}, and charging prices can be similarly variable based on e.g., time-of-use and/or zonal energy pricing~\cite{LADWPpricingzones}
The metric space $(X,d)$ and the spatial movement cost $d(u_{t-1}, u_t)$ capture either the carbon overhead (in terms of ``wasted'' electricity) or the opportunity/time cost of moving to a different location for charging.  The temporal switching cost $\Vert \mbf{x}_t - \mbf{x}_{t-1}\Vert_{\ell_1(\boldsymbol{\beta})}$ captures the small overhead of stopping or restarting charging at a single location, due to extra energy or time spent connecting or disconnecting from the charger at location $u$.  Finally, the constraint function $c( \mbf{x}_t )$ captures how much charge is delivered during time $t$ according to decision $\mbf{x}_t$, where an $\mbf{x}_t$ that places a full allocation in the $\OFF^{(u)}$ state indicates that the AEV is serving customers (i.e., not charging).  Since the AEV may move around in the city while serving customers, $\mbf{x}_t$ should be updated exogenously to reflect the true state.  We note that \SOAD may not capture the case where the distance between charging stations is large (i.e., moving to a different location incurs substantial discharge) or the case where the AEV is at risk of fully discharging the battery before time $T$, which would require immediate charging.

\smallskip
\noindent\textbf{Allocating tasks to volunteers. \ }  Consider a non-profit that has a task to complete before some short-term deadline $T$, with several locations and scheduled time slots for volunteer efforts (e.g., stores, community centers) throughout a region.  Using platforms such as VolunteerMatch~\cite{Manshadi:23}, volunteers can signal their interests in tasks (e.g., via a ranking) and availability -- in assigning this task, the non-profit may want to maximize the engagement of their assigned volunteer(s).  Service costs $f_t^{(u)}$ can represent the aggregate rankings of the volunteers present at location $u$ during time slot $t$, where a lower number means that they are more interested in a given task.  The metric space $(X,d)$ and the spatial movement cost $d(u_{t-1}, u_{t})$ can capture the cost of e.g., moving supplies that must be present at the location to work on the task.  The temporal switching cost $\Vert \mbf{x}_t - \mbf{x}_{t-1}\Vert_{\ell_1(\boldsymbol{\beta})}$ may capture the cost of e.g., setting up or breaking down the setup required to work on the task at a given location $u$.  Finally, the constraint function $c( \mbf{x}_t )$ captures how much of the task can be completed during time $t$ according to assignment decision $\mbf{x}_t$.  We note that the fractional allocation to \ON \ / \OFF states specified by \SOAD may not be useful in this setting because e.g., groups of volunteers may not be fractionally divisible.
\vspace{1em}

}

\begin{table}[h]
\begin{tabular}{|l|l|lll|lll|}
\hline
                           &                                         & \multicolumn{3}{c|}{\begin{tabular}[c]{@{}c@{}}Average Carbon Intensity \vspace{-0.1em}\\ {\footnotesize \textit{(in gCO$_2$eq/kWh)} \cite{electricity-map}} \end{tabular}}                                                      & \multicolumn{3}{c|}{\begin{tabular}[c]{@{}c@{}}Marginal Carbon Intensity \vspace{-0.1em}\\ {\footnotesize \textit{(in gCO$_2$eq/kWh)} \cite{watttime}} \end{tabular}}                                                \\ \cline{3-8} 
\multirow{-2}{*}{Location} & \multirow{-2}{*}{AWS Region}            & \multicolumn{1}{l|}{ {\small Duration} }                                                                              & \multicolumn{1}{l|}{{\small Min.}}   & {\small Max.}   & \multicolumn{1}{l|}{{\small Duration}}                                                                             & \multicolumn{1}{l|}{ {\small Min.} } & {\small Max.} \\ \hline
\worldflag[length=0.75em, width=0.6em]{US} {\small Virginia, U.S.}             & {\footnotesize \texttt{us-east-1}}                               & \multicolumn{1}{l|}{}                                                                                      & \multicolumn{1}{l|}{293} & 567 & \multicolumn{1}{l|}{}                                                                                     & \multicolumn{1}{l|}{48}    & 1436    \\ \cline{1-2} \cline{4-5} \cline{7-8} 
\worldflag[length=0.75em, width=0.6em]{US} {\small California, U.S.}           & {\footnotesize \texttt{us-west-1}}      & \multicolumn{1}{l|}{}                                                                                      & \multicolumn{1}{l|}{83}  & 451 & \multicolumn{1}{l|}{}                                                                                     & \multicolumn{1}{l|}{67}    & 1100    \\ \cline{1-2} \cline{4-5} \cline{7-8} 
\worldflag[length=0.75em, width=0.6em]{US} {\small Oregon, U.S.}               & {\footnotesize \texttt{us-west-2}}      & \multicolumn{1}{l|}{}                                                                                      & \multicolumn{1}{l|}{42}  & 682 & \multicolumn{1}{l|}{}                                                                                     & \multicolumn{1}{l|}{427}    & 2000    \\ \cline{1-2} \cline{4-5} \cline{7-8} 
\worldflag[length=0.75em, width=0.6em]{CA} {\small Quebec, Canada}             & {\footnotesize \texttt{ca-central-1}}   & \multicolumn{1}{l|}{}                                                                                      & \multicolumn{1}{l|}{26}  & 109  & \multicolumn{1}{l|}{}                                                                                     & \multicolumn{1}{l|}{887}    & 1123    \\ \cline{1-2} \cline{4-5} \cline{7-8} 
\worldflag[length=0.75em, width=0.6em]{GB} {\small London, U.K.}               & {\footnotesize \texttt{eu-west-2}}      & \multicolumn{1}{l|}{}                                                                                      & \multicolumn{1}{l|}{56}  & 403 & \multicolumn{1}{l|}{}                                                                                     & \multicolumn{1}{l|}{706}    & 1082    \\ \cline{1-2} \cline{4-5} \cline{7-8} 
\worldflag[length=0.75em, width=0.6em]{FR} {\small France}                     & {\footnotesize \texttt{eu-west-3}}      & \multicolumn{1}{l|}{}                                                                                      & \multicolumn{1}{l|}{18}  & 199 & \multicolumn{1}{l|}{}                                                                                     & \multicolumn{1}{l|}{549}    & 1099    \\ \cline{1-2} \cline{4-5} \cline{7-8} 
\worldflag[length=0.75em, width=0.6em]{SE} {\small Sweden}                     & {\footnotesize \texttt{eu-north-1}}     & \multicolumn{1}{l|}{}                                                                                      & \multicolumn{1}{l|}{12}  & 59  & \multicolumn{1}{l|}{}                                                                                     & \multicolumn{1}{l|}{438}    & 2556    \\ \cline{1-2} \cline{4-5} \cline{7-8} 
\worldflag[length=0.75em, width=0.6em]{DE} {\small Germany}                    & {\footnotesize \texttt{eu-central-1}}   & \multicolumn{1}{l|}{}                                                                                      & \multicolumn{1}{l|}{130} & 765 & \multicolumn{1}{l|}{}                                                                                     & \multicolumn{1}{l|}{11}    & 1877    \\ \cline{1-2} \cline{4-5} \cline{7-8} 
\worldflag[length=0.75em, width=0.6em]{AU} {\small Sydney, Australia}          & {\footnotesize \texttt{ap-southeast-2}} & \multicolumn{1}{l|}{}                                                                                      & \multicolumn{1}{l|}{267} & 817  & \multicolumn{1}{l|}{\multirow{-9}{*}{\small \begin{tabular}[c]{@{}l@{}}01/01/2022 - \\ 12/31/2022\\ \\ Hourly\\ granularity\\ \\ 8,760\\ data points\end{tabular}}} & \multicolumn{1}{l|}{12}    & 1950    \\ \cline{1-2} \cline{4-8} 
\worldflag[length=0.75em, width=0.6em]{BR} {\small Brazil}                     & {\footnotesize \texttt{sa-east-1}}      & \multicolumn{1}{l|}{}                                                                                      & \multicolumn{1}{l|}{46}  & 292 & \multicolumn{3}{c|}{\cellcolor[HTML]{C0C0C0}}                                                                                                \\ \cline{1-2} \cline{4-5}
\worldflag[length=0.75em, width=0.6em]{ZA} {\small South Africa}               & {\footnotesize \texttt{af-south-1}}     & \multicolumn{1}{l|}{}                                                                                      & \multicolumn{1}{l|}{586} & 785 & \multicolumn{3}{c|}{\cellcolor[HTML]{C0C0C0}}                                                                                                \\ \cline{1-2} \cline{4-5}
\worldflag[length=0.75em, width=0.6em]{IL} {\small Israel}                     & {\footnotesize \texttt{il-central-1}}   & \multicolumn{1}{l|}{}                                                                                      & \multicolumn{1}{l|}{514} & 589 & \multicolumn{3}{c|}{\cellcolor[HTML]{C0C0C0}}                                                                                                \\ \cline{1-2} \cline{4-5}
\worldflag[length=0.75em, width=0.6em]{IN} {\small Hyderabad, India}           & {\footnotesize \texttt{ap-south-2}}     & \multicolumn{1}{l|}{}                                                                                      & \multicolumn{1}{l|}{552} & 758 & \multicolumn{3}{c|}{\cellcolor[HTML]{C0C0C0}}                                                                                                \\ \cline{1-2} \cline{4-5}
\worldflag[length=0.75em, width=0.6em]{KR} {\small South Korea}                & {\footnotesize \texttt{ap-northeast-2}} & \multicolumn{1}{l|}{\multirow{-14}{*}{\small \begin{tabular}[c]{@{}l@{}}01/01/2020 - \\ 12/31/2022\\ \\ Hourly \\ granularity\\ \\ 26,304\\ data points\end{tabular}}} & \multicolumn{1}{l|}{453} & 503 & \multicolumn{3}{c|}{\multirow{-5}{*}{\cellcolor[HTML]{C0C0C0} {\small \textit{Data not available}} }}                                                            \\ \hline
\end{tabular}
\caption{Summary of CO$_2$ data sets for each tested AWS region in our case study experiments, including the minimum and maximum carbon intensities, duration, granularity, and data availability.} \label{tab:characteristics}
\end{table}

\begin{figure*}[h]
    \minipage{\textwidth}
    \includegraphics[width=\linewidth]{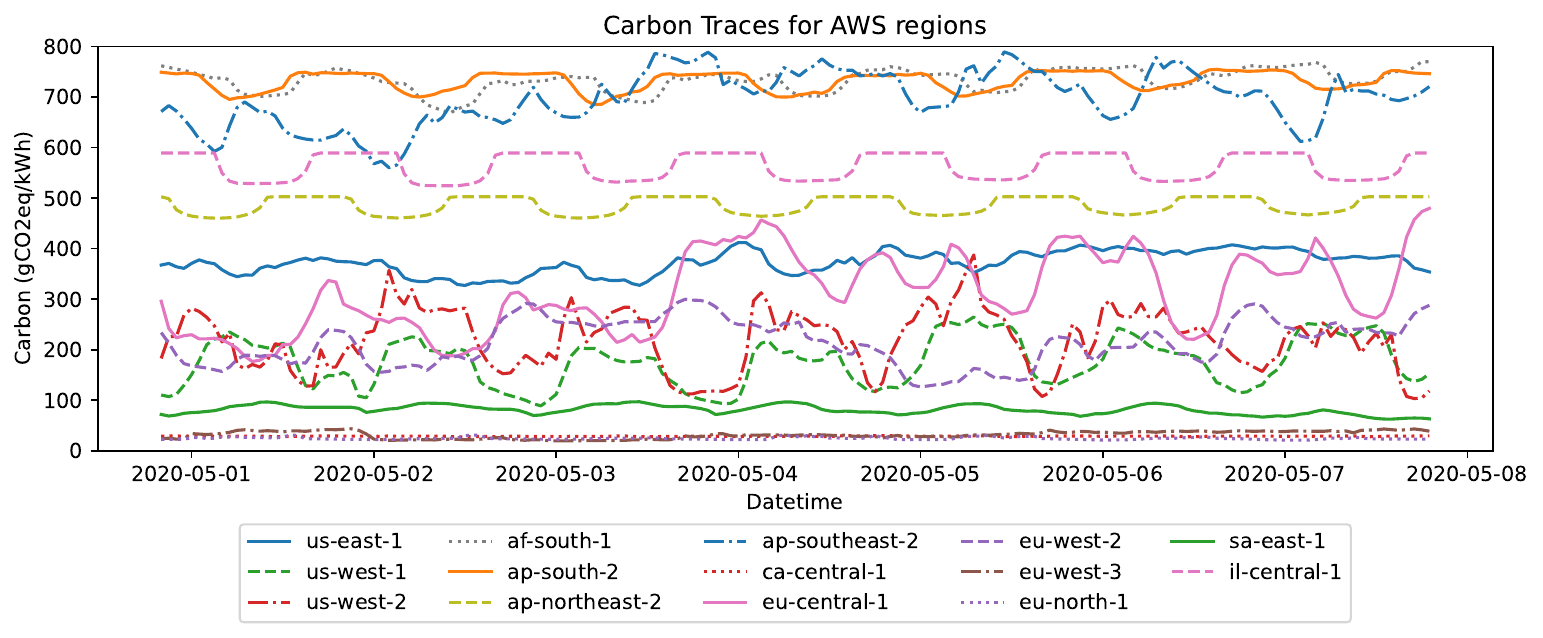}\vspace{-1em}
         \caption{ Average carbon intensity traces~\cite{electricity-map} for all 14 AWS regions, over a week-long period in $2020$.}\label{fig:carbon_traces}
    \endminipage\hfill
\end{figure*}

\section{Supplemental Experiments} \label{apx:exp}
In this section, we present additional results and figures to complement those in the main body.  In the first few results, we present additional experiments manipulating parameters using the \textit{average carbon intensity signal}.  Then, in \autoref{apx:expmarginal}, we present a supplemental slate of experiments using the \textit{marginal carbon intensity signal} obtained from WattTime~\cite{watttime}.

\begin{figure*}[t]
    \minipage{\textwidth}
    \includegraphics[width=\linewidth]{plots/legend.png}\vspace{0.1em}
    \endminipage\hfill\\
    \minipage{0.32\textwidth}
    \includegraphics[width=\linewidth]{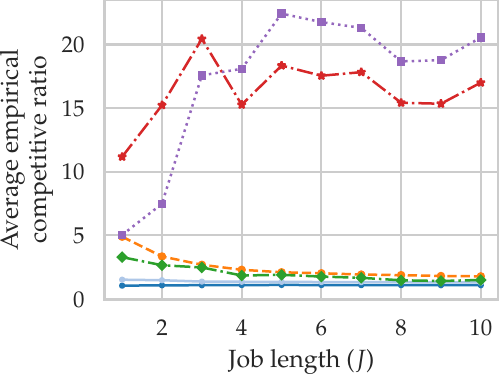}\vspace{-1em}
    \caption{ Average empirical competitive ratios for varying \textit{job length $J$}, with $G = 4, T \thicksim \text{Unif}_{\mathbb{Z}}(12, 48)$. }\label{fig:length}
    \endminipage\hfill  
    \minipage{0.32\textwidth}
    \includegraphics[width=\linewidth]{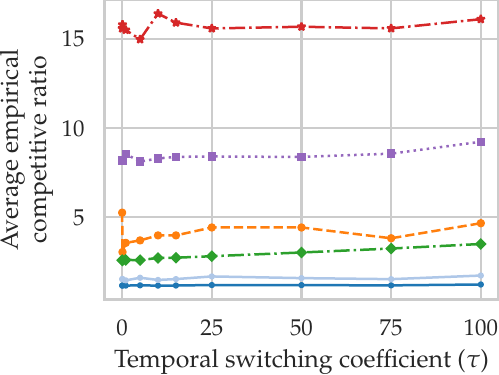}\vspace{-1em}
     \caption{ Average empirical competitive ratios for varying \textit{temporal switching coefficient $\tau$}, with $G = 4, T \thicksim \text{Unif}_{\mathbb{Z}}(12, 48).$ }\label{fig:tau}
     \endminipage\hfill
    \minipage{0.32\textwidth}
     \includegraphics[width=\linewidth]{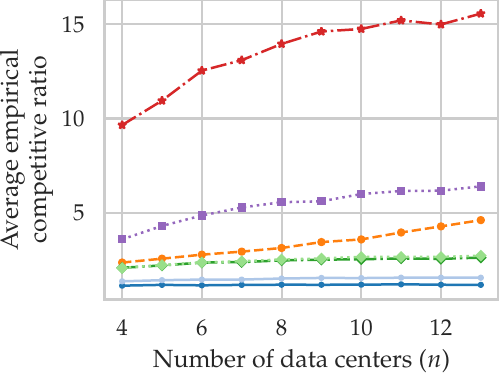}\vspace{-1em}
     \caption{ Average empirical competitive ratios for varying \textit{number of data centers $n$}, with $G = 4, T \thicksim \text{Unif}_{\mathbb{Z}}(12, 48).$ }\label{fig:random_subset}
     \endminipage\hfill  
\end{figure*}

\smallskip
\noindent\textbf{Experiment VI: } \textit{Effect of job length $J$. \ }
\autoref{fig:length} plots the average empirical competitive ratio for fixed job lengths $J \in \{ 1, \dots, 10 \}$, where $J$ is the length of the job in hours.  In this experiment, we fix each job's data size to $G = 4$.  
As $J$ increases, the empirical competitive ratio of \PCM improves, and \CarbonClipper remains consistent.  We note that the simple threshold technique is able to achieve good performance in the case when $J = 1$ -- intuitively, since simple threshold is agnostic to the switching overhead, its performance degrades when it uses more than one opportunity to migrate between regions (i.e., when the job takes more than one time slot). 

As is characteristic of realistic cloud traces, the Google cluster traces we use in Experiments I-V are mostly composed of shorter jobs between $1$ and $2$ hours long.  These results for fixed job lengths highlight that \CarbonClipper and \PCM do even better when given lengthy jobs -- {\color{blue}
such jobs are less frequent but take up a disproportionate amount of compute cycles (and thus contribute disproportionately to the carbon footprint of) a typical data center.}

\subsection*{Experiment VII: \textit{\textmd{\ Effect of temporal switching coefficient $\tau$.}}}
{\color{blue}
Recall that increasing or decreasing $\tau$ simulates jobs that have more or less time-consuming checkpoint and resume overheads, respectively.
In \autoref{fig:tau}, we plot the average empirical competitive ratio for varying $\tau \in \{0, \dots, 100\}$.  
}
In this experiment, we fix $J = 4$ and $G = 4$. %
Compared to varying $G$, $\tau$ has a smaller impact across the board, although as predicted by the theoretical bounds, the performance of \PCM degrades slightly as $\tau$ grows, and the greedy policies are similarly affected.

\subsection*{Experiment VIII: \textit{\textmd{\ Effect of number of data centers $n$.}}} 
{\color{blue}
Building off of the idea in Experiment IV, in \autoref{fig:random_subset} we plot the average empirical competitive ratio for varying $n \in [4, 13]$, where a random subset (of size $n$) is sampled from the base 14 regions for each batch job instance.  
For each job, we fix $G=4$, a deadline $T$ that is a random integer between $12$ and $48$, and use the average carbon intensity signal.  
}
We find that most algorithms' performance degrades as the size of the subset increases.  This is likely because the expected range of carbon intensities expands as more diverse electric grids are included in the subset.  As in previous experiments, \CarbonClipper's performance with black-box advice is consistent as $n$ increases.

\subsection*{Experiment IV (continued): \textit{\textmd{\ Effect of electric grids and data center availability.}}}
Continuing from Experiment IV in the main body, we present results where the metric space is constructed on four additional \textit{subsets} of the 14 regions.  By considering these region subsets, we approximate issues of data center availability, and electric grid characteristics that might face a deployment in practice.

\begin{figure*}[t]
    \minipage{\textwidth}
    \includegraphics[width=\linewidth]{plots/legend.png}\vspace{0.1em}
    \endminipage\hfill\\
     \minipage{0.32\textwidth}
     \includegraphics[width=\linewidth]{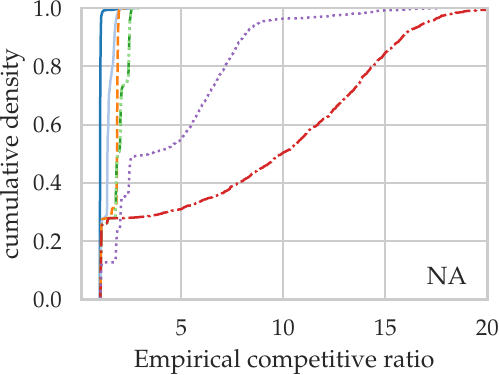}\vspace{-1em}
     \caption{ CDFs of empirical competitive ratios for each tested algorithm on the North America subset.  }\label{fig:subsetApx1}
     \endminipage\hfill
     \minipage{0.32\textwidth}
     \includegraphics[width=\linewidth]{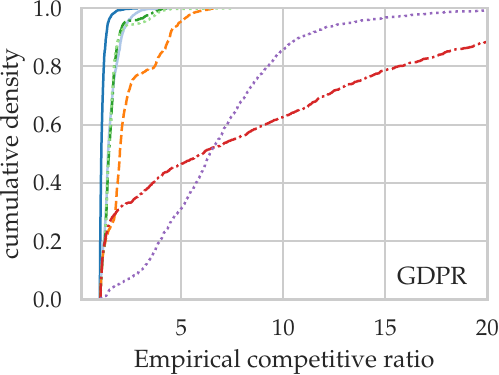}\vspace{-1em}
     \caption{ CDFs of empirical competitive ratios for each tested algorithm on the EU (GDPR) region subset.  }\label{fig:subsetApx2}
     \endminipage\hfill
     \minipage{0.32\textwidth}
     \includegraphics[width=\linewidth]{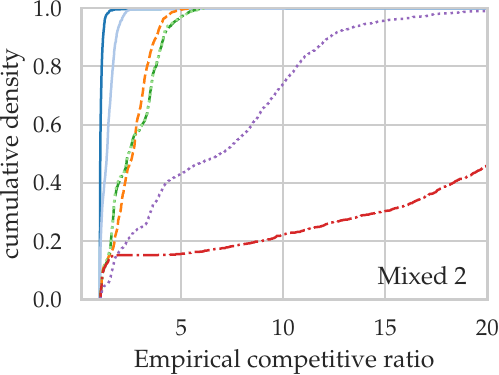}\vspace{-1em}
     \caption{ CDFs of empirical competitive ratios for each tested algorithm on the second ``mixed'' subset.  }\label{fig:subsetApx3}
     \endminipage\hfill
\end{figure*}

\autoref{fig:subsetApx1} considers a North American subset of $4$ regions: California, Oregon, Virginia, and Quebec.  Under this subset, \PCM slightly outperforms greedy and delayed greedy by 8.31\% and 8.24\%, respectively.
\autoref{fig:subsetApx2} considers an ``EU / GDPR'' subset that includes $4$ regions: France, Germany, London, and Sweden.  On this subset, with a large proportion of consistent low-carbon grids (i.e., both Sweden and France), the greedy and delayed greedy policies outperform \PCM by 21.54\% and 20.25\%, respectively.  These results highlight a ``best case'' situation where the greedy policies are able to outperform \PCM and nearly match the performance of \CarbonClipper.
\autoref{fig:subsetApx3} considers a second ``mixed'' subset of $7$ regions: California, South Korea, Germany, Hyderabad, Israel, Sweden, and South Africa.  On this subset, with a geographically distributed mix of high and low-carbon grids, \PCM outperforms greedy and delayed greedy by 3.74\% and 4.73\%, respectively.  

{\color{blue}
We briefly note that since the greedy policies exhibit fairly good performance across many of these experiments, for the intended application of carbon-aware workload management in data centers it may be worthwhile to evaluate the performance of \CarbonClipper when given black-box advice $\ADV$ that simply encodes the decisions of the greedy policy.  Since the black-box advice model can accommodate any arbitrary sequence of decisions, including heuristics, such a composition may achieve a favorable trade-off between average-case performance and worst-case guarantees if e.g., machine-learned forecasts are not available.
}

\subsection*{Experiment IX: \textit{\textmd{\ Runtime (wall clock) overhead measurements.}}}
In this experiment, we measure the wall clock runtime of each tested algorithm. Our experiment implementations are in Python, using NumPy~\cite{NumPy}, SciPy~\cite{SciPy}, and CVXPY~\cite{CVXPY} -- we use the \verb|time.perf_counter()| module in Python to calculate the total runtime (in milliseconds) for each algorithm on each instance, and report this value \textit{normalized by the deadline} to give the \textit{per-slot} (i.e., hourly) overhead of each algorithm.

\autoref{fig:runtime} reports these measurements for batch jobs with deadlines $T \in \{6, ..., 48\}$, fixed $G = 4$, and job lengths from the Google trace are truncated to $\nicefrac{T}{2}$ if necessary.  For this experiment, we run each algorithm and each instance in a single thread on a MacBook Pro with M1 Pro processor and 32 GB of RAM.

\begin{wrapfigure}{r}{0.32\textwidth}
    \vspace{0.2em}
  \includegraphics[width=0.32\textwidth]{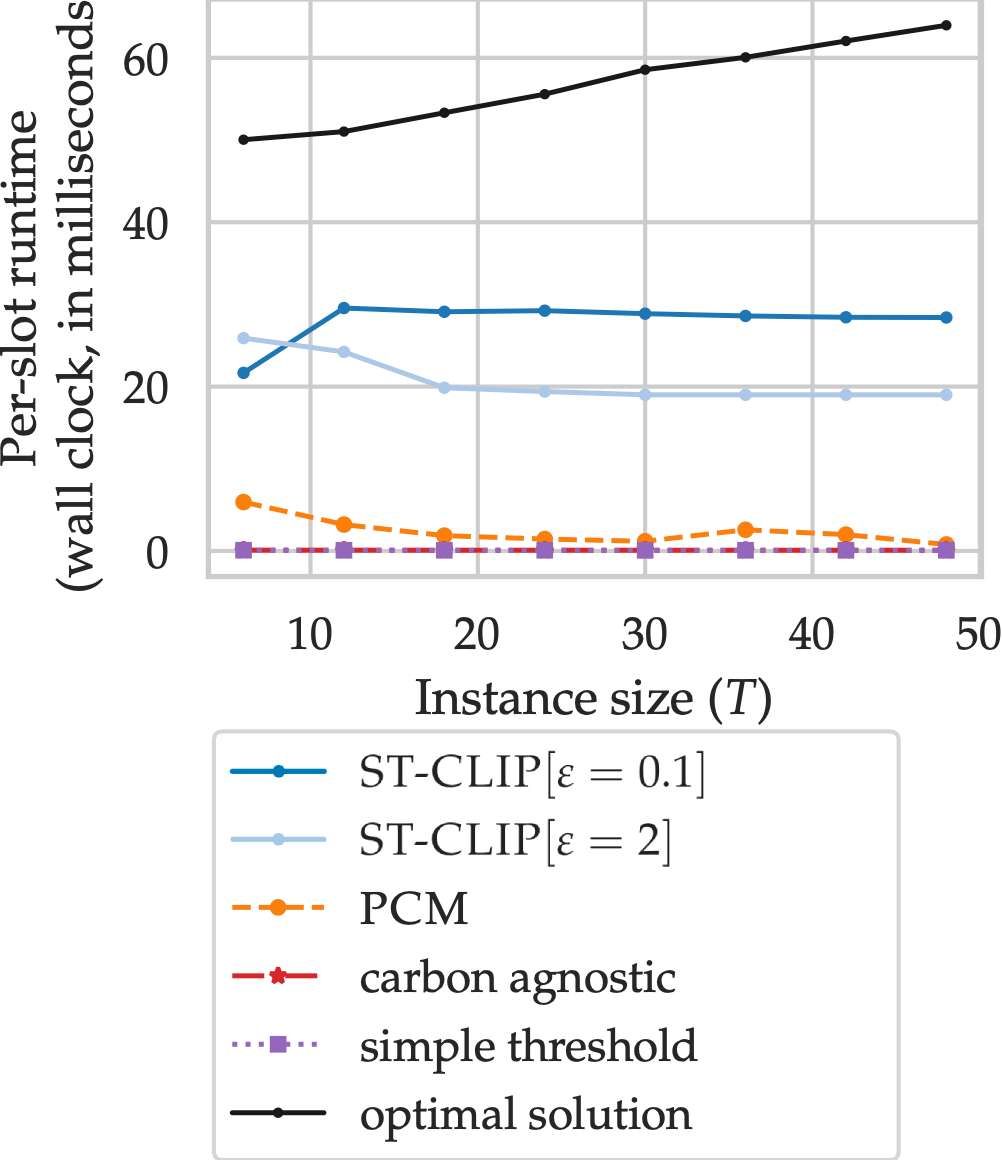}
  \caption{Average per-slot wall clock runtime for \textit{instance sizes (i.e., deadlines) $T \in \{6, ..., 48\}$}, with $G = 4$. }\label{fig:runtime}
\end{wrapfigure}
We find that the average per-slot (i.e., once per hour) runtimes of \PCM and \CarbonClipper are 2.3 milliseconds and 24.3 milliseconds, respectively -- this reflects the relative complexity of the optimization problems being solved in each, and the runtime is steady in the size of the instance, as is expected from an online algorithm.
The optimal solution takes, on average, $173.8$ milliseconds per slot to solve, although note that it finds the solution for all time slots at once.  As the size of the instance ($T$) grows, this per-slot time slightly increases.
{\color{blue}
Intuitively, the decision rule-based algorithms (carbon agnostic, simple threshold, and the greedy policies (not included in the plot)) have the lowest, and functionally negligible, runtime.
}
This last result suggests that one effective way to reduce the impact of \PCM and \CarbonClipper's runtime overhead would be to develop approximations that avoid computing the exact solution to the minimization problem.  However, for carbon intensity signals that are updated every 5 minutes to one hour, the overhead of \PCM and \CarbonClipper is likely reasonable in practice.

\@@par
\ifnum\@@parshape=\z@ \let\WF@pspars\@empty \fi %
\global\advance\c@WF@wrappedlines-\prevgraf \prevgraf\z@
\ifnum\c@WF@wrappedlines<\tw@ \WF@finale \fi

\subsection{Marginal Carbon Intensity} \label{apx:expmarginal}

\begin{figure*}[h]
    \minipage{\textwidth}
    \includegraphics[width=\linewidth]{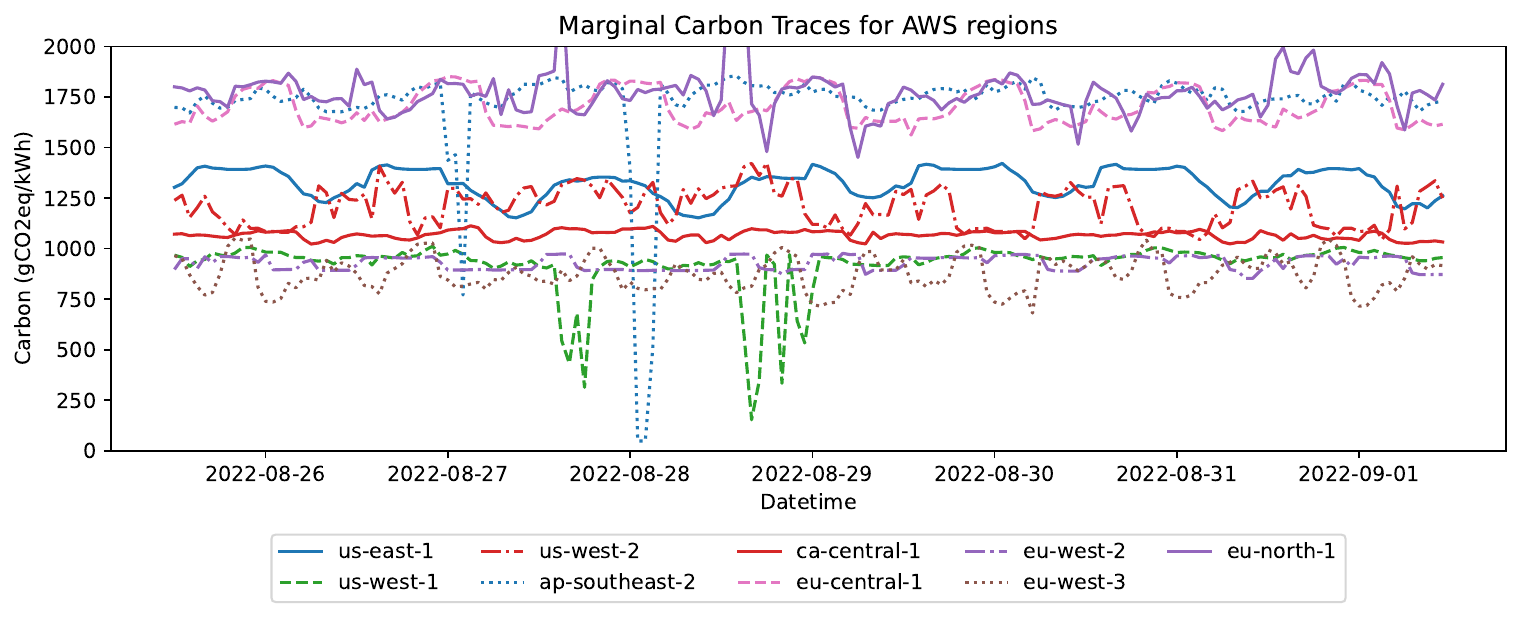}\vspace{-1em}
         \caption{ Marginal carbon intensity traces~\cite{electricity-map} for 9 AWS regions, over a week-long period in $2022$.}\label{fig:marginal_carbon_traces}
    \endminipage\hfill
\end{figure*}

\noindent In contrast to average carbon intensity, the marginal carbon intensity signal calculates the emissions of the generator(s) that are responding to changes in load on a grid at a certain time.  From WattTime~\cite{watttime}, we obtain  data for 9 of the 14 regions we consider, spanning all of 2022.  
{\color{blue}
This data also includes carbon intensity forecasts published by WattTime, and we use these forecasts directly instead of generating synthetic forecasts as in \autoref{sec:expsetup}. 
}
In \autoref{fig:marginal_carbon_traces}, we plot a one-week sample of carbon intensity data to motivate this visually.

A top-level summary of these experiments is given in \autoref{fig:cdf_marginal} -- this plot gives a cumulative distribution function (CDF) of the empirical competitive ratios for all tested algorithms, aggregating over all experiments that use the marginal carbon intensity signal.  In these experiments, we consider all of the $9$ regions for which the marginal data is available, and each job's deadline $T$ is a random integer between $12$ and $48$ (henceforth denoted by $T \thicksim \text{Unif}_{\mathbb{Z}} (12, 48)$).

{\color{blue}
In these experiments, we observe differences that are likely attributable to the characteristic behavior of the marginal carbon intensity signal, which generally represents the high emissions rate of a quick-to-respond generator (e.g., a gas turbine) unless the supply of renewables on the grid exceeds the current demand.
}
However, the relative ordering of performance has been largely preserved. 

\begin{figure*}[t]
    \minipage{\textwidth}
    \includegraphics[width=\linewidth]{plots/legend.png}\vspace{0.1em}
    \endminipage\hfill\\
    \minipage{0.32\textwidth}
     \includegraphics[width=\linewidth]{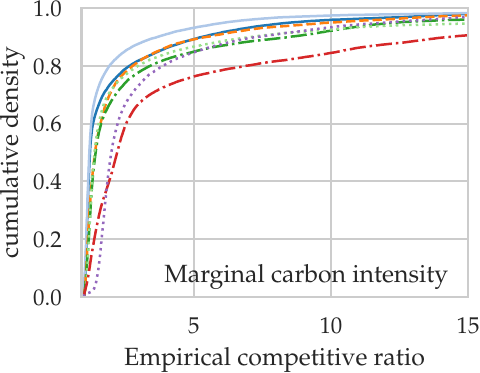}\vspace{-1em}
     \caption{ CDFs of competitive ratios for each algorithm, across all marginal carbon intensity experiments. }\label{fig:cdf_marginal}
     \endminipage\hfill
    \minipage{0.32\textwidth}
     \includegraphics[width=\linewidth]{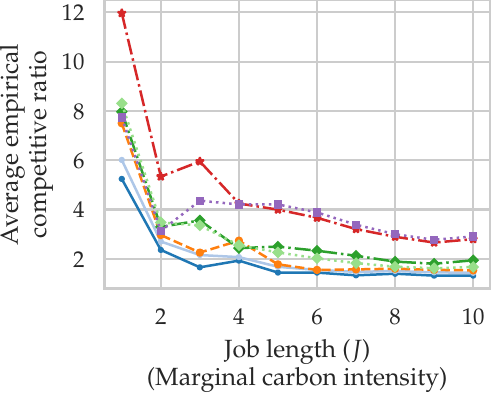}\vspace{-1em}
     \caption{ Average competitive ratios using marginal carbon data, with varying \textit{job length $J$} and $G = 4, T \thicksim \text{Unif}_{\mathbb{Z}}(12, 48)$. }\label{fig:marginal}
     \endminipage\hfill
     \minipage{0.32\textwidth}
     \includegraphics[width=\linewidth]{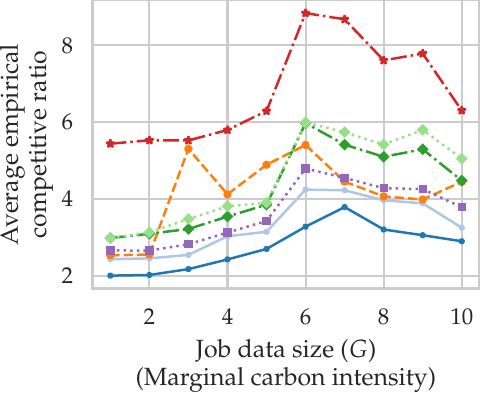}\vspace{-1em}
     \caption{ Average competitive ratios using marginal carbon data, with varying \textit{job data size $G$} and $T \thicksim \text{Unif}_{\mathbb{Z}}(12, 48)$. }\label{fig:marginalgb}
     \endminipage\hfill
\end{figure*}

Interestingly, we note that the real forecasts are worse in the marginal setting compared to the average carbon intensity signal, again likely because of the characteristics of marginal carbon.  Rather than predicting e.g., the diurnal patterns of an average signal, predicting marginal carbon requires a model to pick out specific time slots where curtailment is expected to occur.

{\color{blue}
Despite these challenges of forecasting the marginal signal, \CarbonClipper with $\varepsilon = 2$ outperforms the baselines in both average and worst-case performance, improving on the closest greedy policy by an average of 12.04\%, and outperforming delayed greedy, simple threshold, and carbon-agnostic by averages of 11.02\%, 37.12\%, and 85.67\%, respectively.   
}

\subsection*{Marginal Experiment I: \textit{\textmd{\ Effect of job length $J$.}}} 
{\color{blue}
In \autoref{fig:marginal}, we plot the average competitive ratio for different job lengths $J \in \{1, \dots, 10\}$.  
}
Each job has $G = 4$ and an integer deadline $T$ randomly sampled between $12$ and $48$. 
Compared to the greedy policies, our robust baseline \PCM performs favorably in these experiments, which is likely due to a combination of differences in the marginal setting and the less performant forecasts.
Similar to the average carbon setting, the simple threshold technique performs well when the job length is short, but performance suffers when it has more opportunities for migration.

\subsection*{Marginal Experiment II: \textit{\textmd{\ Effect of job data size $G$.}}}
{\color{blue}
In \autoref{fig:marginalgb}, we plot the average empirical competitive ratio for different job data sizes $G \in \{1, \dots, 10\}$ using the marginal carbon signal.
}
As predicted by the theoretical bounds, and observed in the main body for the average carbon experiments, the performance of \PCM degrades as $G$ grows -- we observe the same effect for the greedy policies and the simple threshold algorithm.  The performance of \CarbonClipper also grows, which suggests that larger migration overhead does have an impact when the advice suffers from a lack of precision due to the more challenging setting for forecasts posed by marginal carbon intensity.

\section{Proofs for \autoref{sec:problem} (Problem Formulation, Challenges and Technical Preliminaries)} \label{apx:equivRand}
{\color{blue}
In the following, we prove \autoref{thm:equivRand}, which shows that the expected cost of any randomized \SOAD decision $\mbf{p}_t \in \Delta_\mathcal{S}$ is equivalent to that of a decision which chooses a point in $X$ probabilistically according to the distribution of $\mbf{p}_t$ and then interprets the $\texttt{ON / OFF}$ probabilities at that point as deterministic allocations in $\mathcal{X}$.
}

\begin{proof}[Proof of \autoref{thm:equivRand}]
{\color{blue}
Suppose that $\mbf{p}_t, \mbf{p}_{t-1}$ are probability distributions over the randomized state space $\Delta_{\mathcal{S}}$, and let $\mathbb{E}\left[ \text{Cost}(\mbf{p}_t, \mbf{p}_{t-1}) \right]$ denote the expected cost of decision $\mbf{p}_t$.}  This cost is defined as follows:
\begin{align*}
\mathbb{E}\left[ \text{Cost}(\mbf{p}_t, \mbf{p}_{t-1}) \right] &= \mathbb{E} \left[ f_t(\mbf{p}_t) + g (\mbf{p}_t, \mbf{p}_{t-1} ) \right].
\end{align*}
Recall that because the cost function $f_t$ is linear and separable, the expectation can be written as: 
{\color{blue}
\[
\mathbb{E}\left[ f_t(\mbf{p}_t) \right] = \sum_{u \in X} f_t^{(u)} p_t^{\ON^{(u)}}.
\]
For any $\mbf{p}_t \in \Delta_{\mathcal{S}}$, let $\mbf{r}_t \coloneqq \{ 
r_t^{(u)} \gets p_t^{\ON^{(u)}} + p_t^{\OFF^{(u)}} : u \in X\} \in \Delta_X$, i.e., a vector that aggregates the total probabilities across the states space $\mathcal{S}$ at each point of $X$.
}

We note that by disaggregating the spatial and temporal switching costs, we have that the expectation $\mathbb{E} \left[ g (\mbf{p}_t, \mbf{p}_{t-1} ) \right]$ can be written in terms of the Wasserstein-1 distance with respect to the underlying metric $X$ and a linear temporal term that depends on the probability assigned to the \OFF state.  This is the case because the optimal transport plan $\mathbb{W}^1(\mbf{p}_t, \mbf{p}_{t-1})$ must always involve first moving probability mass to/from the \OFF and \ON states at each point of $X$, and then within the spatial metric -- 
{\color{blue}
this follows since \SOAD defines that movement within the spatial metric $X$ can only be made between \ON states -- i.e., a player moving from $\OFF^{(u)}$ to $\OFF^{(v)}$ must first traverse to $\ON^{(u)}$, then through the metric $X$, and finally through $\ON^{(v)}$.
\[
\mathbb{E} \left[ g (\mbf{p}_t, \mbf{p}_{t-1} ) \right] = \mathbb{W}^1 (\mbf{r}_t, \mbf{r}_{t-1}) + \sum_{u \in X} \beta^{(u)} \vert p_t^{\OFF^{(u)}} - p_{t-1}^{\OFF^{(u)}} \vert.
\]
}
Thus, the expected cost of $\mbf{p}_t$ can be written as:
{\color{blue}
\[
\mathbb{E}\left[ \text{Cost}(\mbf{p}_t, \mbf{p}_{t-1}) \right] = \sum_{u \in X} f_t^{(u)} p_t^{\ON^{(u)}} + \mathbb{W}^1 (\mbf{r}_t, \mbf{r}_{t-1}) + \sum_{u \in X} \beta^{(u)}  \vert p_t^{\OFF^{(u)}} - p_{t-1}^{\OFF^{(u)}} \vert.
\]
\noindent In the mixed probabilistic/deterministic setting, the true allocation to $\ON^{(u)}$ (denoted by $\tilde{p}_t^{\ON^{(u)}}$) for any point $u \in X$ is defined as $\tilde{p}_t^{\ON^{(u)}} = \nicefrac{p_t^{\ON^{(u)}}}{r_t^{(u)}}$ (conversely, we have  $\tilde{p}_t^{\OFF^{(u)}} = \nicefrac{p_t^{\OFF^{(u)}}}{r_t^{(u)}}$).
Letting $\mathbb{L}_t^{(u)} \in \{0,1\}$ denote an indicator variable that encodes the player's location (i.e., point) at time $t$ (0 if player is not at $u$, 1 if player is at $u$), the expected cost of $\tilde{\mbf{p}}_t$ can be written as follows:
\[
\mathbb{E}\left[ \text{Cost}(\tilde{\mbf{p}}_t, \tilde{\mbf{p}}_{t-1}) \right] = \mathbb{E}\left[ \sum_{u \in X} \mathbb{L}_t^{(u)} f_t^{(u)} \tilde{p}_t^{\ON^{(u)}} \right] + \mathbb{W}^1(\mbf{r}_t, \mbf{r}_{t-1}) + \mathbb{E}\left[ \sum_{u \in X} \left\vert \mathbb{L}_t^{(u)} \beta^{(u)} \tilde{p}_t^{\OFF^{(u)}} - \mathbb{L}_{t-1}^{(u)} \beta^{(u)} \tilde{p}_{t-1}^{\OFF^{(u)}} \right\vert \right].
\]
Noting that $\mathbb{E} \left[ \mathbb{L}_t^{(u)} \right] = r_t^{(u)}$ by linearity of expectation we have:
\[
\mathbb{E}\left[ \text{Cost}(\tilde{\mbf{p}}_t, \tilde{\mbf{p}}_{t-1}) \right] = \sum_{u \in X} r_t^{(u)} f_t^{(u)} \tilde{p}_t^{\ON^{(u)}} + \mathbb{W}^1(\mbf{r}_t, \mbf{r}_{t-1}) + \sum_{u \in X} \vert  r_t^{(u)} \beta^{(u)} \tilde{p}_t^{\OFF^{(u)}} - r_{t-1}^{(u)} \beta^{(u)} \tilde{p}_{t-1}^{\OFF^{(u)}} \vert.
\]
Furthermore, recalling the definitions of $\tilde{p}_t^{\ON^{(u)}}$ and $\tilde{p}_t^{\OFF^{(u)}}$, we have the following:
\[
\mathbb{E}\left[ \text{Cost}(\tilde{\mbf{p}}_t, \tilde{\mbf{p}}_{t-1}) \right] = \sum_{u \in X} f_t^{(u)} p_t^{\ON^{(u)}} + \mathbb{W}^1(\mbf{r}_t, \mbf{r}_{t-1}) + \sum_{u \in X} \beta^{(u)} \vert p_t^{\OFF^{(u)}} - p_{t-1}^{\OFF^{(u)}} \vert.
\]
Recalling that $\mathbb{W}^1(\mbf{p}_t, \mbf{p}_{t-1}) = \mathbb{W}^1(\mbf{r}_t, \mbf{r}_{t-1}) + \sum_{u \in X} \beta^{(u)} \vert p_t^{\OFF^{(u)}} - p_{t-1}^{\OFF^{(u)}} \vert$ by the structure of the spatial and temporal switching costs completes the proof, since $\mathbb{E}\left[ \text{Cost}(\mbf{p}_t, \mbf{p}_{t-1}) \right] = \mathbb{E}\left[ \text{Cost}(\tilde{\mbf{p}}_t, \tilde{\mbf{p}}_{t-1}) \right]$, and thus the expected cost is equivalent if a point (location) is first chosen probabilistically and the \ON \ / \OFF probabilities at that point are then interpreted as deterministic (fractional) allocations.}
\end{proof}

\section{Proofs for \autoref{sec:pcm} (A Competitive Online Algorithm)} \label{apx:pcm}
\subsection{Convexity of the pseudo-cost minimization problem in \PCM} \label{apx:pseudo-convex}

{\color{blue}
In this section, we prove \autoref{thm:PCMconvex}, which states that the pseudo-cost minimization problem central to the design of \PCM is a convex minimization problem, implying that it can be solved efficiently.

For convenience, let $h_t(\mbf{k}) : t \in [T]$ represent the pseudo-cost minimization problem's objective for a single arbitrary time step:
}
\begin{equation}
    h_t(\mbf{k}) = f_t(\mbf{k}) + \norm{\mbf{k} - \mbf{k}_{t-1}} - \int_{z^{(t-1)}}^{z^{(t-1)} + \overline{c}( \mbf{k} ) }\psi(u) du.
\end{equation}
 
\begin{proof}[Proof of \autoref{thm:PCMconvex}]

We prove the statement by contradiction.
By definition, the sum of two convex functions gives a convex function.  Since $\norm{\mbf{k} - \mbf{k}_{t-1}}$ is a norm and $\mbf{k}_{t-1}$ is fixed, by definition it is convex.  We have also assumed as part of the problem setting that each $f_t( \mbf{k} )$ is linear.  Thus, $f_t(\mbf{k}) + \norm{\mbf{k} - \mbf{k}_{t-1}}$ must be convex.
{\color{blue}
The remaining term is the negation of $\int_{z^{(t-1)}}^{z^{(t-1)} + \overline{c}( \mbf{k} ) }\psi(u) du$.  Let $w(\overline{c}(\mbf{k})) = \int_{z^{(t-1)}}^{z^{(t-1)} + \overline{c}( \mbf{k} ) }\psi(u) du$.
}
By the fundamental theorem of calculus, we have 
\[
\nabla w(\overline{c}(\mbf{k})) = \psi( z^{(t-1)} + \overline{c}( \mbf{k} ) ) \nabla \overline{c}(\mbf{k}).
\]
Let $b(\overline{c}(\mbf{k})) = \psi( z^{(t-1)} + \overline{c}( \mbf{k} ) )$.  Then we have 
\[
\nabla^2 w(\overline{c}(\mbf{k})) = \nabla^2 \overline{c}(\mbf{k}) w(\overline{c}(\mbf{k})) + \nabla \overline{c}(\mbf{k}) b'(\overline{c}(\mbf{k})) \nabla \overline{c}(\mbf{k})^\intercal.
\]
{\color{blue}
Since $\overline{c}(\mbf{k})$ is piecewise linear by the definition of \SOAD, we know that $\nabla^2 \overline{c}(\mbf{k}) b(\overline{c}(\mbf{k})) = 0$.
}
Since $\psi$ is monotonically decreasing on the interval $[0,1]$, we know that $ b'(\overline{c}(\mbf{k})) < 0$, and thus $\nabla \overline{c}(\mbf{k}) b'(\overline{c}(\mbf{k})) \nabla \overline{c}(\mbf{k})^\intercal$ is negative semidefinite.  This implies that $w(\overline{c}(\mbf{k}))$ is concave in $\mbf{k}$.

Since the negation of a concave function is convex, this causes a contradiction, because the sum of two convex functions gives a convex function.
Thus, $h_t(\cdot) = f_t(\mbf{k}) + \norm{ \mbf{k} - \mbf{k}_{t-1}} - \int_{z^{(t-1)}}^{z^{(t-1)} + \overline{c}( \mbf{k} ) }\psi(u) du$ is always convex under the assumptions of \SOAD.
\end{proof}

By showing that $h_t( \cdot ) $ is convex, it follows that the pseudo-cost minimization \eqref{eq:pseudocostPCM} in \PCM is a convex minimization problem (i.e., it can be solved efficiently using numerical methods).

\subsection{Proof of \autoref{thm:etaCompPCM}} \label{apx:etaCompPCM}

In the following, we prove Theorem~\ref{thm:etaCompPCM}. In what follows, we let $\mathcal{I} \in \Omega$ denote an arbitrary valid $\SOAD$ instance.  Let $z^{(j)} = \sum_{t\in[T]} \overline{c}( \mbf{k}_t )$ denote the final utilization before the mandatory allocation. Also note that $z^{(t)} = \sum_{m\in[t]} \overline{c}( \mbf{k}_m )$ is non-decreasing over $t$.

{\color{blue}
In what follows, we let $\eta$ be defined as the solution to $\ln \left( \frac{U - L - D - 2\tau}{U - \nicefrac{U}{\eta} - D} \right) = \frac{1}{\eta}$, which has a closed form given by:
\begin{equation}
    \eta \coloneqq \left[ W \left( \frac{(D+L-U+2\tau) \exp \left(\frac{D-U}{U} \right) }{ U } \right) + \frac{U - D}{U} \right]^{-1}.
\end{equation}
Note that setting $\eta$ as above satisfies the following equality within the pseudo-cost function $\psi$ (defined in \sref{Def.}{dfn:psi-pcm}):
$$\psi(1) = U - \tau + (\nicefrac{U}{\eta} - U + D + \tau) \exp(\nicefrac{z}{\eta}) = L + D .$$
}
We start by proving \sref{Lemma}{lem:pcm-opt-lb}, which states that $\OPT$ is lower bounded by 
\[
{\OPT}(\mathcal{I}) \ge \frac{\max \left\{ \psi(z^{(j)}) - D, L \right\}}{O(\log n)}.
\]
\paragraph{Proof of \sref{Lemma}{lem:pcm-opt-lb}.} 

Without loss of generality, denote the minimum gradient of any cost function (excluding \OFF states) by $\nabla_{\min}$.  Suppose that a cost function $f_m$ with $\nabla_{\min}$ gradient (at a dimension corresponding to any \ON state) arrives at time step $m$.

\noindent Recall that $\PCM$ solves the following pseudo-cost minimization problem at time $m$:
\[
\mbf{k}_m = \argmin_{\mbf{k} \in K : \overline{c}(\mbf{k}) \leq 1-z^{(t-1)}} f_m (\mbf{k}) + \norm{ \mbf{k} - \mbf{k}_{m-1} } - \int_{z^{(m-1)}}^{z^{(m-1)}+ \overline{c}(\mbf{k})} \psi(u) du
\] 

{\color{blue} %
\noindent By assumption, since $f_m( \cdot)$ is linear and satisfies $\nabla f_m < \psi(z^{(j)}) - D$, there must exist a dimension in $f_m$ (i.e., a service cost associated with an $\ON$ state) that satisfies the following.  Let $\ON^{[d]} \subset [d]$ denote the index set (i.e., the dimensions in $\mbf{k}$) that correspond to allocations in $\ON$ states.
\[
\exists i \in \ON^{[d]} : f_m(\mbf{k})_i \leq \left[ \nabla_{\min} \cdot \overline{c}(\mbf{k}) \right]_i.
\]
Also note that $\norm{ \mbf{k} - \mbf{k}_{m-1} }$ is upper bounded by $(D + \tau) \overline{c}(\mbf{k}) $, since in the worst-case, \PCM must pay the max movement and switching cost to move the allocation to the ``furthest'' point and make a decision $\mbf{k}$.}

{\color{blue}
Since $\psi$ is monotone decreasing on the interval $z \in [0,1]$, by definition we have that $\mbf{k}_m$ solving the true pseudo-cost minimization problem is lower-bounded by the $\Breve{\mbf{k}}_m$ solving the following minimization problem (specifically, the constraint satisfaction satisfies $\overline{c} ( \Breve{\mbf{k}}_m ) \leq \overline{c} ( \mbf{k}_m ) $):
\begin{equation}
\Breve{\mbf{k}}_m = \argmin_{\mbf{k} \in K : \overline{c}(\mbf{k}) \leq 1-z^{(t-1)}} \nabla_{\min} \cdot \overline{c}(\mbf{k}) + D \overline{c}(\mbf{k}) +\tau \overline{c}(\mbf{k}) - \int_{z^{(m-1)}}^{z^{(m-1)}+ \overline{c}(\mbf{k})} \psi(u) du. \label{eq:lbminprob}
\end{equation}
By expanding the right hand side, we have:}
\begin{align*}
    &  \nabla_{\min} \cdot \overline{c}(\mbf{k}) + D \overline{c}(\mbf{k}) + \tau \overline{c}(\mbf{k}) - \int_{z^{(m-1)}}^{z^{(m-1)}+ \overline{c}(\mbf{k})} \psi(u) du\\
    & \nabla_{\min} \cdot \overline{c}(\mbf{k}) + D \overline{c}(\mbf{k}) + \tau \overline{c}(\mbf{k}) - \int_{z^{(m-1)}}^{z^{(m-1)}+ \overline{c}(\mbf{k})} \left[ U - \tau + (\nicefrac{U}{\eta} - U + D + \tau) \exp(\nicefrac{z}{\eta}) \right] du\\
    & (\nabla_{\min} - U + D + \tau) \overline{c}(\mbf{k}) - \left[ (\tau - U + D) \eta + U \right] \left( \exp \left( \frac{z^{(m-1)}+ \overline{c}(\mbf{k})}{\eta} \right) - \exp \left( \frac{z^{(m-1)}}{\eta} \right) \right)
\end{align*}
{\color{blue}
Letting $\overline{c}(\mbf{k})$ be some scalar $y$ (which is valid since we assume there is at least one dimension $i \in [d]$ where the growth rate of $f_m(\cdot)$ is at most $\nabla_{\min}$), the pseudo-cost minimization problem finds the value $y$ that minimizes the following quantity:
}
\begin{align*}
    &  (\nabla_{\min} - U + D +\tau) y - \left[ (\tau - U + D) \eta + U \right] \left( \exp \left( \frac{z^{(m-1)}+ y}{\eta} \right) - \exp \left( \frac{z^{(m-1)}}{\eta} \right) \right)
\end{align*}
Taking the derivative of the above with respect to $y$ yields the following:
\begin{align}
    & =  \nabla_{\min} + D + \tau - U - \frac{\left[ (\tau - U + D) \eta + U \right] \exp \left( \frac{z^{(m-1)}+ y}{\eta} \right)}{\eta}
    & = \nabla_{\min} + D - \psi(z^{(m-1)} + y) \label{align:mingrad}
\end{align}
{\color{blue}
Note that since the minimization problem is convex by \autoref{thm:PCMconvex}, the unique solution to the above coincides with a point $y$ where the derivative is zero.  This implies that $\PCM$ will increase $\overline{c}(\mbf{k})$ until $\nabla_{\min} = \psi(z^{(m-1)} + y) - D$, which further implies that $\psi(z^{(m)}) = \nabla_{\min} + D$.  

Since this minimization $\Breve{\mbf{k}}_m$ is a lower bound on the true value of $\mbf{k}_m$, this implies that $\psi(z^{(t)}) - D$ is a lower bound on the minimum service cost coefficient (excluding \OFF states) seen \textit{so far} at time $t$.}  Further, the \textit{final utilization} $z^{(j)}$ gives that the minimum service cost coefficient over any cost function over the entire sequence is lower bounded by $\psi(z^{(j)}) - D$.  

Note that the best choice for $\OPT$ is to service the entire workload at the minimum service cost if it is feasible.  
{\color{blue}
Since the vector space $(K, \norm{\cdot})$ used by \PCM has at most $O(\log n)$ distortion with respect to the underlying metric used by $\OPT$ (see \sref{Definition}{dfn:tree}), this implies that $\OPT(\mathcal{I}) \geq \frac{ \max \left\{ \psi(z^{(j)}) - D, \ L \right\} }{O(\log n)}$.  \hfill \qed
}

\bigskip

Next, we prove \sref{Lemma}{lem:pcm-alg-ub}, which states that the expected cost of $\PCM(\mathcal{I})$ is upper bounded by $\ex[ {\PCM}(\mathcal{I}) ] \le \int_{0}^{z^{(j)}} \psi(u) du + (1- z^{(j)}) U + \tau z^{(j)}$.
\paragraph{Proof of \sref{Lemma}{lem:pcm-alg-ub}.}

Recall that $z^{(t)} = \sum_{m\in[t]} \overline{c}( \mbf{k}_m )$ is non-decreasing over $t$.

\noindent Observe that whenever $\overline{c}( \mbf{k}_t ) > 0$, we have that $f_t(\mbf{k}_t) + \norm{ \mbf{k}_t - \mbf{k}_{t-1} } < \int_{z^{(t-1)}}^{z^{(t-1)}+ \overline{c}( \mbf{k}_t )} \psi(u) du$.  Then, if $\overline{c}(\mbf{k}_t) = 0$, which corresponds to the case when $\mbf{k}_t$ allocates all of the marginal probability mass to $\OFF$ states, we have the following:
\begin{equation}
f_t(\mbf{k}_t) + \norm{ \mbf{k}_t - \mbf{k}_{t-1} } - \int_{z^{(t-1)}}^{z^{(t-1)}+ \overline{c}( \mbf{k}_t )}  \psi(u) du = 0 + \norm{ \mbf{k}_t - \mbf{k}_{t-1} } - 0 = \norm{ \mbf{k}_t - \mbf{k}_{t-1} } \label{align:excesscost}
\end{equation}
{\color{blue}
This gives that for any time step where $\overline{c}(\mbf{k}_t ) = 0$, we have the following inequality, which follows by observing that from \sref{Assumption}{ass:switching-cost}, any marginal probability mass assigned to $\ON$ states in the previous time step can be moved to $\OFF$ states at a cost of at most $\tau \overline{c}(\mbf{k}_{t-1})$.
\begin{equation}
    f_t(\mbf{k}_t) + \norm{ \mbf{k}_t - \mbf{k}_{t-1} } \le \norm{ \mbf{k}_t - \mbf{k}_{t-1} } \le \tau \overline{c}(\mbf{k}_{t-1}), \forall t\in[T] : \overline{c}(\mbf{k}_t ) = 0.
\end{equation}
Since any time step where $\overline{c}( \mbf{k}_t ) > 0$ implies that $f_t(\mbf{k}_t) + \norm{ \mbf{k}_t - \mbf{k}_{t-1} } < \int_{z^{(t-1)}}^{z^{(t-1)}+ \overline{c}(\mbf{k}_t )} \psi(u) du$, we have the following inequality across all time steps (\textit{i.e., an upper bound on the excess cost not accounted for by the pseudo-cost}):
}
\begin{equation}
    f_t(\mbf{k}_t) + \norm{ \mbf{k}_t - \mbf{k}_{t-1} } - \int_{z^{(t-1)}}^{z^{(t-1)}+ \overline{c}(\mbf{k}_t )} \psi(u) du \leq \tau \overline{c}(\mbf{k}_{t-1}), \forall t\in[T].
\end{equation}
Thus, we have
\begin{align}
   \tau z^{(j)} = \sum_{t\in[j]} \tau \overline{c}(\mbf{k}_{t-1}) &\ge \sum_{t\in[j]}  \left[ f_t(\mbf{k}_t) + \norm{ \mbf{k}_t - \mbf{k}_{t-1} } - \int_{z^{(t-1)}}^{z^{(t-1)}+ \overline{c}(\mbf{k}_t )} \psi(u) du \right]\\
   &= \sum_{t\in[j]}  \left[ f_t(\mbf{k}_t) + \norm{ \mbf{k}_t - \mbf{k}_{t-1} } \right] - \int_{0}^{z^{(j)}} \psi(u) du \\
    &= \PCM(\mathcal{I}) - (1-z^{(j)})U - \int_{0}^{z^{(j)}}\psi(u) du.
\end{align}
Combining Lemma~\ref{lem:pcm-opt-lb} and Lemma~\ref{lem:pcm-alg-ub} gives
\begin{align}
    \texttt{CR} \le \frac{\ex[ \PCM(\mathcal{I}) ]}{\OPT(\mathcal{I})} \le \frac{\int_{0}^{z^{(j)}} \psi(u) du + (1- z^{(j)}) U + \tau z^{(j)}}{\max \{ \psi(z^{(j)}) - D, L \} } \leq \eta,
\end{align}
where the last inequality holds since for any $z\in[0,1]$:
\begin{align}
  \int_{0}^{z} \psi(u) du + \tau z + (1- z) U &=  \int_{0}^{z} \left[U - \tau + (\nicefrac{U}{\eta} - U + D + \tau) \exp(\nicefrac{z}{\eta}) \right] du + (1 - z)U + \tau z \\
  &=  \left[ \left((\tau - U + D){\eta}+U\right)\mathrm{e}^\frac{u}{{\eta}}+ U u - \tau u  \right]_{0}^{z} + (1 - z)U + \tau z \\ %
  &= \left((\tau - U + D){\eta}+U\right)\mathrm{e}^\frac{z}{{\eta}} - (\tau - U + D){\eta}\\
  &= \eta \left[ \left((\tau - U + D) + \frac{U}{\eta} \right)\mathrm{e}^\frac{z}{{\eta}} - {\tau} + U - D \right]\\
  &= \eta [\psi(z) - D]. \label{align:bestservicecost}
\end{align}
This completes the proof of \autoref{thm:etaCompPCM}. \hfill \qed

\subsection{Proof of \autoref{thm:lowerboundCSWM}} \label{apx:lowerboundCSWM}

In this section, we prove Theorem~\ref{thm:lowerboundCSWM}, which states that $\eta$ (as defined in \eqref{eq:eta}) is the optimal competitive ratio for \SOAD.
To show this lower bound, we first define a family of special two-stage adversaries, a corresponding metric space $X$, and then show that the competitive ratio for any algorithm is lower bounded under the instances provided by these adversaries.

Prior work has shown that difficult instances for online search problems with a minimization objective occur when inputs arrive at the algorithm in an decreasing order of cost~\cite{Lorenz:08, Lechowicz:23, ElYaniv:01, SunZeynali:20}.  
{\color{blue}
For \SOAD, we extend this idea and additionally consider how adaptive adversaries can strategically present good service cost functions at distant points in the metric first, followed by good service costs at the starting point (e.g., ``at home''), to create a family of sequences that simultaneously penalize the online player for moving ``too much'' and for not moving enough.
}

We now formalize two such families of adversaries, namely $\{\mathcal{G}_y \}_{y \in [L,U]}$ and $\{ \mathcal{A}_y \}_{y \in [L, U]}$, where $\mathcal{A}_y$ and $\mathcal{G}_y$ are both called $y$-adversaries.

\begin{figure*}[!h]
    \minipage{0.5\textwidth}
    \centering
    \includegraphics[width=\linewidth]{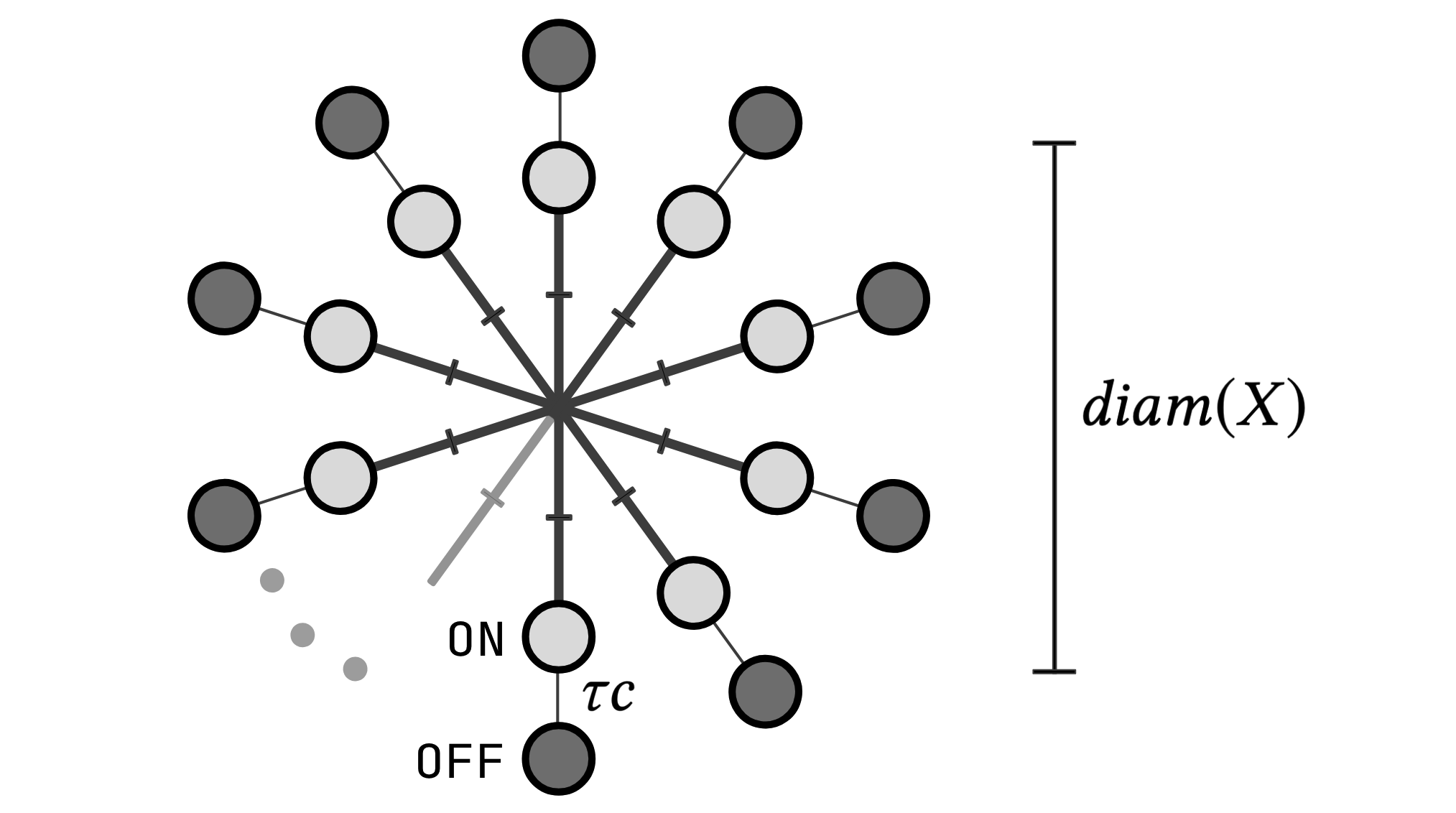}\vspace{-0.4em}
    \endminipage\hfill
    \minipage{0.5\textwidth}
    \caption{A motivating illustration of the lower bound star metric considered in  \sref{Definition}{dfn:yadversary-caslb}. Light gray circles represent the points of the metric space $(X, d)$, and darker circles represent the \OFF states of \SOAD. Note that the distance between any two points in the metric is $\diam(X) = D c$. }
    \endminipage
\end{figure*} \label{fig:metric-lb}
\begin{dfn}[$y$-adversaries for \SOAD] 
Let $m \in \mathbb{N}$ be sufficiently large, and $\sigma := \nicefrac{(U-L)}{m}$. 
The metric space $X$ is a \textit{weighted star metric} with $n$ points%
, each with an \texttt{ON} and \texttt{OFF} state.  
{\color{blue}
For the constraint function $c(\cdot)$, we set one throughput constant for all \ON states such that $c^{(u)} \ll 1 : u \in X$.  This value is henceforth simply denoted by $c$.  See \autoref{fig:metric-lb} for an illustration.
}

The movement cost can be represented by a weighted $\ell_1$ norm $\norm{\cdot}$ that combines the spatial distances given by the metric with the temporal switching cost.
{\color{blue}
Recall \sref{Assumption}{ass:switching-cost} -- at any single point $u$, the switching cost between the \ON and \OFF states is given by $\tau c \cdot  \vert x_t^{\ON^{(u)}} -  x_{t-1}^{\ON^{(b)}} \vert$ (i.e., for two arbitrary allocations $\mbf{x}_t$ and $\mbf{x}_{t-1}$) for $\tau > 0$.
Furthermore, for any two disjoint points in the metric $u, v : u \not = v$, the distance between $\ON^{(u)}$ and $\ON^{(v)}$ is exactly $Dc = \diam(X)$.
}

Let $y$ denote a value on the interval $[L, U]$, which represents the ``best service cost function'' presented by the $y$-adversary in their sequence.
We define two distinct stages of the input.

{\color{blue}
In \textbf{Stage 1}, the adversary presents two types of cost functions such that ``good cost functions'' are always at points that are distant to the online player.  These cost functions are denoted (for convenience) as $\text{Up}(x) = U x$, and $\text{Down}^i (x) = (U - i \sigma) x$, where one such cost function is delivered at each point's $\ON$ state, at each time step.

Without loss of generality, let the starting point be $s \in X$ and the start state be $\OFF^{(s)}$ (for both \emph{\ALG} and \emph{\OPT}).  In the first time step, the adversary presents cost function $\text{Up}(x) = U x$ at the starting point's $\ON$ state (i.e., $\ON^{(s)}$), and $\text{Down}^1 (x)$ at all of the other $(n-1)$ $\ON$ states.

If \emph{\ALG} \emph{ever} moves a non-zero fractional allocation to an $\ON$ state other than the starting point, that point becomes \emph{inactive}, meaning that the adversary will present $\text{Up}(x)$ at that location in the next time step and for the rest of the sequence.

In the second time step, the adversary presents cost function $\text{Up}(x) = U x$ at the starting point and any inactive points, and $\text{Down}^1 (x)$ at all of the other $\ON$ states. } The adversary continues to sequentially present  $\text{Down}^1 (x)$ in this manner until it has presented it at least $\mu$ times (where $\mu \coloneqq \nicefrac{1}{c}$).  It then moves on to present $\text{Down}^2 (x)$, $\text{Down}^3 (x)$, and so forth.  The adversary follows the above pattern, presenting ``good cost functions'' to a shrinking subset of $\ON$ states until they present $\text{Down}^{m_y} (x) = y x$ up to $\mu$ times at any remaining active states.  In the time step after the last $\text{Down}^{m_y} (x) = y x$ is presented, the adversary presents $\text{Up}(x)$ everywhere, and \textbf{Stage 1} ends.

{\color{blue}
In \textbf{Stage 2}, the adversary presents $\text{Up}(x)$ and $\text{Down}^i (x)$ cost functions at the starting point $s$, in an alternating fashion.  All other $\ON$ states are considered inactive in this stage, so they only receive $\text{Up}(x)$.  In the first time step, the adversary presents $\text{Down}^1 (x)$ at the starting point, followed by $\text{Up}(x)$ in the following time step.  In the third time step, the adversary presents $\text{Down}^2 (x)$ at the starting point, followed by $\text{Up}(x)$ in the subsequent time step.  The adversary continues alternating in this manner until they present $\text{Down}^{m_y} (x) = y x$ at the starting point.  In the $\mu-1$ time steps after $\text{Down}^{m_y} (x) = y x$ is presented, the adversary presents $\text{Down}^{m_y} (x) = y x$ at the starting point (allowing \emph{\OPT} to reduce their switching cost). } Finally, the adversary presents $\text{Up}(x)$ everywhere for the final $\mu$ time steps, and \textbf{Stage 2} ends.

The first family of $y$-adversaries, $\{\mathcal{G}_y \}_{y \in [L,U]}$, only uses \textbf{Stage 1} -- the sequence ends when \textbf{Stage 1} ends.
The second family, $\{ \mathcal{A}_y \}_{y \in [L, U]}$, sequentially uses both stages -- cost functions are presented using the full \textbf{Stage 1} sequence first, which is then followed by the full \textbf{Stage 2} sequence.  This concatenated sequence ends when  \textbf{Stage 2} ends.

{\color{blue}
Note that the final cost function for any point in any $y$-adversary instance is always $\text{Up}(x)$. }

\end{dfn} \label{dfn:yadversary-caslb}

\begin{proof}[Proof of \autoref{thm:lowerboundCSWM}]
Let $s(y)$ and $t(y)$ denote \textit{constraint satisfaction functions} mapping $[L,U] \rightarrow [0,1]$, that fully describe $\ALG$'s expected deadline constraint satisfaction (i.e., $\mathbb{E}\left[ \sum_{t \in [j]} \overline{c}(\mbf{k}) \right]$ before the mandatory allocation) during \textbf{Stage 1} and \textbf{Stage 2}, respectively.  

Note that for large $m$, \textbf{Stage 1} and \textbf{Stage 2} for $y = z - \sigma$ are equivalent to first processing \textbf{Stage 1} and \textbf{Stage 2} for $y = z$, and then processing an additional batch of cost functions such that the best cost function observed is $\text{Down}^{m_{z-\sigma}} (x) = (z - \sigma) x$.

As the expected deadline constraint satisfactions at each time step are unidirectional (irrevocable), we must have that $s(y)$ and $t(y)$ both satisfy $s(y - \sigma) \geq s(y)$ and $t(y - \sigma) \geq t(y)$, i.e. $s(y)$ and $t(y)$ are non-decreasing on $[L, U]$.
Note that the optimal solutions for adversaries $\mathcal{G}_y$ and $\mathcal{A}_y$ are not the same.  

{\color{blue}
In particular, we have that $\OPT(\mathcal{G}_y) = \min\{ y + D + \tau, U \}$.  For relatively large $y$, the optimal solution may choose to satisfy the constraint at the starting point, but for sufficiently small $y$, the optimal solution on $\mathcal{G}_y$ may choose to move to a distant point.
For $\mathcal{A}_y$, since the ``good cost functions'' arrive at the starting point, we have $\OPT(\mathcal{A}_y) = y$ for any $y \in [L, U]$.
}

Due to the adaptive nature of the $y$-adversary, any $\ALG$ incurs expected movement cost proportional to $s(y)$ during \textbf{Stage 1}.  
{\color{blue}
Furthermore, since $c \ll 1$ for all $\ON$ states, note that as soon as $s(y) > c$, in expectation, $\ALG$ has \textit{moved away from} the starting point.  Thus, during \textbf{Stage 2}, $\ALG$ must also incur expected movement cost proportional to $t(y)$ as it ``moves back'' to the starting point.
}
Let $\mbf{l} = \nicefrac{U}{\eta^\star} - 2\tau$ denote the worst marginal cost that an $\eta^\star$-competitive $\ALG$ should be willing to accept in either stage.
The total expected cost incurred by an $\eta^\star$-competitive online algorithm $\ALG$ on adversaries $\mathcal{A}_y$ and $\mathcal{G}_y$ can be expressed as follows:
\begingroup
\allowdisplaybreaks
\begin{align}
\ex [ \ALG(\mathcal{G}_y) ] &= s(\mbf{l}) \mbf{l} - \int^y_{\mbf{l}} u d s(u) + D s(y) + (1 - s(y))U + 2\tau s(y) \\
\ex [ \ALG(\mathcal{A}_y) ] &= s(\mbf{l}) \mbf{l} - \int^y_{\mbf{l}} u d s(u) + D s(y) + t(\mbf{l}) \mbf{l} - \int^y_{\mbf{l}} u d t(u) + Dt(y) +\\ 
& \hspace{15em} (1 - s(y) - t(y))U + 2\tau \left[ s(y) + t(y) \right] \label{eq:lb-special-cost}
\end{align}
\endgroup
{\color{blue}
In the above expressions, $uds(u)$ is the expected cost of buying $ds(u)$ constraint satisfaction at cost $u$.  The same convention extends to $udt(u)$.  $Ds(y)$ and $Dt(y)$ represent the movement cost paid by $\ALG$ during \textbf{Stage 1} and \textbf{Stage 2}, respectively, given that ``good service cost functions'' arrive at distant points (i.e. at distance $D$).}  $(1 - s(y))U$ and $(1 - s(y) - t(y))U$ give the expected cost of the mandatory allocation on adversary $\mathcal{G}_y$ and $\mathcal{A}_y$, respectively.  Similarly, $2\tau$ gives the expected temporal switching cost due to allocation decisions made (excepting the mandatory allocation).

For any $\eta^\star$-competitive algorithm, the constraint satisfaction functions $s(\cdot)$ and $t(\cdot)$ must simultaneously satisfy $\ex[ \ALG(\mathcal{G}_y) ] \leq \eta^\star \OPT(\mathcal{G}_y)$ and $\ex[ \ALG(\mathcal{A}_y) ] \leq \eta^\star \OPT(\mathcal{A}_y)$ for all $y \in [L, U]$.  This gives a necessary condition that the functions must satisfy as follows:
{\color{blue} 
\begin{align*}
s(\mbf{l}) \mbf{l} - \int^y_{\mbf{l}} u d s(u) + D s(y) + (1 - s(y))U + 2\tau s(y) \leq \eta^\star \left[ y + D + 2\tau \right]& \\
s(\mbf{l}) \mbf{l} - \int^y_{\mbf{l}} u d s(u) + t(\mbf{l}) \mbf{l} - \int^y_{\mbf{l}} u d t(u) + (1 - s(y) - t(y))U + \left[ D + 2\tau \right] (s(y) + t(y)) \leq \eta^\star \left[ y \right]&
\end{align*}
}
By integral by parts, the above expressions imply that the constraint satisfaction functions $s(y)$ and $t(y)$ must satisfy the following conditions:
\begin{align}
s(y) &\geq \frac{U - \eta^\star y - \eta^\star D - 2 \eta^\star \tau}{U - y - D - 2\tau} - \frac{1}{U - y - D - 2\tau} \int_{\mbf{l}}^y s(u) du\\
t(y) &\geq \frac{(y + D - U + 2\tau) s(y) }{U - y - D - 2\tau} + \frac{U - \eta^\star y}{U - y - D - 2\tau} - \frac{1}{U - y - D - 2\tau} \int_{\mbf{l}}^y s(u) + t(u) du\\ 
&\geq \frac{\eta^\star D + \eta^\star 2\tau}{U - y - D - 2\tau} - \frac{1}{U - y - D - 2\tau} \int_{\mbf{l}}^y t(u) du \label{align:lb-conditionineq}
\end{align}
In what follows, we substitute $\mbf{l} = \frac{U}{\eta} - 2\tau$.
By Gr\"{o}nwall's Inequality \cite[Theorem 1, p. 356]{Mitrinovic:91}, we have the following:
\begingroup
\allowdisplaybreaks
\begin{align*}
s(y) & \geq \frac{U - \eta^\star y - \eta^\star D - 2 \eta^\star \tau}{U - y - D - 2\tau} - \int_{\mbf{l}}^y \frac{U - \eta^\star u - \eta^\star D - 2 \eta^\star \tau}{(U - u - D - 2\tau)^2} du \\
& \geq \frac{U - \eta^\star y - \eta^\star D - 2 \eta^\star \tau}{U - y - D - 2\tau} - \left[ \frac{2 \tau \eta^\star + U \eta^\star - 2\tau -U}{u - U + D + 2\tau} - \eta^\star \ln \left( U - u - D - 2\tau \right) \right]_{\mbf{l}}^y \\
& \geq \eta^\star \ln \left( U - y - D - 2\tau \right) - \eta^\star \ln \left( U - \mbf{l} - D - 2\tau \right) - \frac{\eta^\star D + \eta^\star 2 \tau}{\frac{U}{\eta^\star} - U + D}, \quad \forall y \in [L, U].\\
t(y) & \geq \frac{\eta^\star D + \eta^\star 2 \tau}{U - y - D - 2\tau} -  \int_{\mbf{l}}^y \frac{\eta^\star D + \eta^\star 2 \tau}{(U - u - D- 2\tau)^2} du \\
& \geq \frac{\eta^\star D + \eta^\star 2 \tau}{U - y - D - 2\tau} - \left[ \frac{- \eta^\star D - \eta^\star 2 \tau}{u - U + D + 2\tau} \right]_{\mbf{l}}^y \\
& \geq \frac{\eta^\star D + \eta^\star 2 \tau}{\mbf{l} - U + D + 2\tau}.
\end{align*}
\endgroup
By the definition of the adversaries $\mathcal{G}_y$ and $\mathcal{A}_y$, we have that $t(L) \leq 1 - s(L)$, and thus $s(L) \leq 1 - t(L)$.  We combine this inequality with the above bounds to give the following condition for any\\ $\eta^\star$-competitive online algorithm:
\[
\eta^\star \ln \left( U - L - D - 2\tau \right) - \eta^\star \ln \left( U - \mbf{l} - D - 2\tau \right) - \frac{\eta^\star D + \eta^\star 2 \tau}{\frac{U}{\eta^\star} - U + D} \leq s(L) \leq 1 - t(L) \leq 1 - \frac{\eta^\star D + \eta^\star 2 \tau}{\mbf{l} - U + D + 2\tau}.
\]
The optimal $\eta^\star$ is obtained when the above inequality is binding, and is given by the solution to the following transcendental equation (after substituting $\frac{U}{\eta} - 2\tau$ for $\mbf{l}$):
\begin{align}
\ln \left[ \frac{U - L - D - 2\tau}{U - \nicefrac{U}{\eta^\star} - D} \right] = \frac{1}{\eta^\star}. \label{eq:etastar}
\end{align}
The solution to the above is given by the following (note that $W(\cdot)$ denotes the Lambert W function):
\begin{align}
    \eta^\star \to \left[ W \left( \frac{(D+L-U+2\tau) \exp \left(\frac{D-U}{U} \right) }{ U } \right) + \frac{U - D}{U} \right]^{-1}
\end{align}
\end{proof}

\section{Proofs for \autoref{sec:clip} (A Learning-augmented Algorithm)} \label{apx:clip}
\subsection{Proof of \autoref{thm:constrobCLIP}} \label{apx:constrobCLIP}

In this section, we prove Theorem~\ref{thm:constrobCLIP}, which states that \CarbonClipper is $(1+\varepsilon)$-consistent for any $\varepsilon \in (0, \eta - 1]$, and $O(\log n) \gamma^{(\varepsilon)}$-robust, where $\gamma^{(\varepsilon)}$ is the solution to \eqref{eq:gamma}.

\begin{proof}
We show the result by separately considering consistency (the competitive ratio when advice is correct) and robustness (the competitive ratio when advice is not correct) in turn.

Recall that the black-box advice \ADV is denoted by a decision $\mbf{a}_t$ at each time $t$.
Throughout the proof, we use shorthand notation $\CLIP_t$ to denote the expected cost of $\CarbonClipper$ up to time $t$, and $\ADV_t$ to denote the cost of $\ADV$ up to time $t$.
We start by proving \sref{Lemma}{lem:clipconst} to show that \CarbonClipper is $(1+\varepsilon)$-consistent.

\paragraph{Proof of \sref{Lemma}{lem:clipconst}.}

First, we note that the constrained optimization enforces that the expected cost of \CarbonClipper so far plus a term that forecasts the mandatory allocation is always within $(1 + \varepsilon)$ of the advice.  There is always a feasible $\mbf{p}_t$ that satisfies the constraint, because setting $\mathbf{k}_t = \Phi \mbf{a}_t$ is always within the feasible set.  Formally, if time step $j \in [T]$ denotes the time step marking the start of the mandatory allocation, the constraint in \eqref{eq:const-constraint} holds for every time step $t \in [j]$.

Thus, to show $(1+\varepsilon)$ consistency, we must resolve the cost of any actions during the \textit{mandatory allocation} and show that the final expected cost of \CarbonClipper is upper bounded by $(1+\varepsilon) \ADV$.

Let $\mathcal{I} \in \Omega$ be an arbitrary valid \CSWM sequence.
If the mandatory allocation begins at time step $j < T$, both $\CarbonClipper$ and $\ADV$ must greedily satisfy the constraint during the last $m$ time steps $[j, T]$.  This is assumed to be feasible, and the cost due to switching in and out of \texttt{ON / OFF} states is assumed to be negligible as long as $m$ is sufficiently large.

Let $(1 - z^{(j-1)})$ denote the remaining deadline constraint that must be satisfied by $\CarbonClipper$ in expectation at these final $m$ time steps, and let $(1 - A^{(j-1)})$ denote the remaining deadline constraint to be satisfied by $\ADV$.
We consider two cases, corresponding to the cases where $\CarbonClipper$ has \textit{underprovisioned} with respect to \ADV (i.e., it has completed less of the deadline constraint in expectation) and \textit{overprovisioned} (i.e., completed more of the deadline constraint), respectively.

\begin{center} \textbf{Case 1: $\CarbonClipper(\mathcal{I})$ has ``underprovisioned'' ($(1 - z^{(j-1)}) > (1 - A^{(j-1)})$). \ } \end{center} 
In this case, \CarbonClipper must satisfy \textit{more} of the deadline constraint (in expectation) during the mandatory allocation compared to \ADV.
From the previous time step, we know that the following constraint holds: 
\begin{align*}
&\CLIP_{j-1} + \mathbb{W}^{1} (  \mbf{p}_{j-1}, \mbf{a}_{j-1} ) + \tau c(\mbf{a}_{j-1}) + (1 - A^{(j-1)} ) L + (A^{(j-1)} - z^{(j-1)})U \le\\
& \hspace{20em} (1+ \varepsilon) \left[ \ADV_{j-1} + \tau c(\mbf{a}_{j-1}) + (1 - A^{(j-1)} ) L \right].
\end{align*}

Let $\{ \mbf{p}_t \}_{t \in [j, T]}$ and $\{ \mbf{a}_t \}_{t \in [j, T]}$ denote the decisions made by $\CarbonClipper$ and $\ADV$ during the mandatory allocation, respectively.  Conditioned on the fact that $\CarbonClipper$ has completed $z^{(j-1)}$ fraction of the deadline constraint in expectation, we have that $\mathbb{E} \left[ \sum_{t=j}^T c(\mbf{p}_t) \right] = (1 - z^{(j-1)})$ and $\sum_{t=j}^T c(\mbf{a}_t)  = (1 - A^{(j-1)})$.

Considering $\{ f_t(\cdot) \}_{t \in [j, T]}$, by assumption we have a lower bound based on $L$, namely $\sum_{t=j}^T f_t(\mbf{a}_t) \geq L \sum_{t=j}^T c(\mbf{a}_t)$.  
For the service costs that \CarbonClipper must incur \textit{over and above} what \ADV incurs, we have a upper bound based on $U$, so $\mathbb{E} \left[ \sum_{t=j}^T f_t(\mbf{p}_t) \right] \leq \sum_{t=j}^T f_t(\mbf{a}_t) + U (\sum_{t=j}^T c(\mbf{p}_t)- \sum_{t=j}^T c(\mbf{a}_t))$

{\color{blue}
Note that the worst case for $\CarbonClipper$ occurs when $\sum_{t=j}^T f_t(\mbf{a}_t)$ exactly matches this lower bound, i.e., $\sum_{t=j}^T f_t(\mbf{a}_t) = L \sum_{t=j}^T c(\mbf{a}_t)$, as $\ADV$ is able to satisfy the rest of the deadline constraint at the best possible marginal price.  Furthermore, note that if $\CarbonClipper$ and $\ADV$ are in different points of the metric at time $j$, the term $\mathbb{W}^{1} ( \mbf{p}_{j-1}, \mbf{a}_{j-1} )$ in the left-hand-side of the constraint allows $\CarbonClipper$ to ``move back'' and follow \ADV just before the mandatory allocation begins, thus leveraging the same cost functions as \ADV.
}
By the constraint in the previous time step, we have the following:
\begin{align}
\CLIP_{j-1} + \mathbb{W}^{1}  (  \mbf{p}_{j-1}, \mbf{a}_{j-1} ) + & \tau c(\mbf{a}_{j-1}) + (1 - A^{(j-1)} ) L + (A^{(j-1)} - z^{(j-1)})U& \\
& \leq (1+\varepsilon) [ \ADV_{j-1} + \tau c(\mbf{a}_{j-1}) + (1 - A^{(j-1)} ) L  ], \\
\CLIP_{j-1} + \tau c(\mbf{a}_{j-1}) + & L  \sum_{t=j}^T c(\mbf{a}_t) + U \left( \sum_{t=j}^T c(\mbf{p}_t)- \sum_{t=j}^T c(\mbf{a}_t) \right)  \\
& \leq (1+\varepsilon) \left[ \ADV_{j-1} + \tau c(\mbf{a}_{j-1}) + L \sum_{t=j}^T c(\mbf{a}_t) \right] \leq (1+\varepsilon) \ADV(\mathcal{I}).  \\
\mathbb{E} \left[ \CarbonClipper(\mathcal{I}) \right] &\leq (1+\varepsilon) \ADV(\mathcal{I}). \label{eq:case1consistency}
\end{align}

\begin{center} \textbf{Case 2: $\CarbonClipper(\mathcal{I})$ has ``overprovisioned'' ($(1 - z^{(j-1)}) \le (1 - A^{(j-1)})$). \ } \end{center}
In this case, \CarbonClipper must satisfy \textit{less} of the deadline constraint (in expectation) during the mandatory allocation compared to \ADV.

\noindent From the previous time step, we know that the following constraint holds: 
\begin{align*}
\CLIP_{j-1} + \mathbb{W}^{1} ( \mbf{p}_{j-1}, \mbf{a}_{j-1} ) + \tau c(\mbf{a}_{j-1}) + (1 - z^{(j-1)} ) L  \le (1+ \varepsilon) \left[ \ADV_{j-1} + \tau c(\mbf{a}_{j-1}) + (1 - A^{(j-1)} ) L \right].
\end{align*}

Let $\{ \mbf{p}_t \}_{t \in [j, T]}$ and $\{ \mbf{a}_t \}_{t \in [j, T]}$ denote the decisions made by \CarbonClipper and $\ADV$ during the mandatory allocation, respectively. As previously, we have that $\mathbb{E} \left[ \sum_{t=j}^T c(\mbf{p}_t) \right] = (1 - z^{(j-1)})$ and $\sum_{t=j}^T c(\mbf{a}_t)  = (1 - A^{(j-1)})$.

Considering $\{ f_t(\cdot) \}_{t \in [j, T]}$, we have a lower bound on $\sum_{t=j}^T f_t(\cdot)$ based on $L$, namely $\sum_{t=j}^T f_t(\mbf{p}_t) \geq L \sum_{t=j}^T c(\mbf{p}_t)$ and $\sum_{t=j}^T f_t(\mbf{a}_t) \geq L \sum_{t=j}^T c(\mbf{a}_t)$.
Since \CarbonClipper has ``overprovisioned'', we know $\mathbb{E} \left[ \sum_{t=j}^T c(\mbf{p}_t) \right] \leq \sum_{t=j}^T c(\mbf{a}_t)$, and thus it follows that $\ex\left[ \sum_{t=j}^T f_t(\mbf{p}_t) \right] \leq \sum_{t=j}^T f_t(\mbf{a}_t)$.

By the constraint in the previous time step, we have:
\begin{align*}
\frac{\CLIP_{j-1} + \mathbb{W}^{1} ( \mbf{p}_{j-1},  \mbf{a}_{j-1} ) + \tau c(\mbf{a}_{j-1}) + (1 - z^{(j-1)}) L}{\ADV_{j-1} + \tau c(\mbf{a}_{j-1}) + (1 - A^{(j-1)} ) L } = \hspace{10em} &\\
\frac{\CLIP_{j-1} + \mathbb{W}^{1} ( \mbf{p}_{j-1},  \mbf{a}_{j-1} ) + \tau c(\mbf{a}_{j-1}) + L \sum_{t=j}^T c(\mbf{p}_t)}{\ADV_{j-1} +\tau c(\mbf{a}_{j-1}) + L \sum_{t=j}^T c(\mbf{a}_t)} &\leq (1 + \varepsilon).
\end{align*}
\noindent Let $y = \ex \left[ \sum_{t=j}^T f_t(\mbf{p}_t) \right] - L \sum_{t=j}^T c(\mbf{p}_t)$, and let $y' = \sum_{t=j}^T f_t(\mbf{a}_t) - L \sum_{t=j}^T c(\mbf{a}_t)$.  By definition, $y \geq 0$ and $y' \geq 0$.
Note that by resolving the mandatory allocation and by definition, we have that the final expected cost $\mathbb{E} \left[ \CarbonClipper(\mathcal{I}) \right] \le \CLIP_{j-1} + \mathbb{W}^1 ( \mbf{p}_{j-1},  \mbf{a}_{j-1} ) + \tau c(\mbf{a}_{j-1}) + L \sum_{t=j}^T c(\mbf{p}_t) + y$, and $\ADV(\mathcal{I}) \ge \ADV_{j-1} + \tau c(\mbf{a}_{j-1}) + L \sum_{t=j}^T c(\mbf{a}_t) + y'$.

Furthermore, since \CarbonClipper has ``overprovisioned'' and by the linearity of the cost functions $f_t( \cdot )$, we have that $y \leq y'$.
Combined with the constraint from the previous time step, we have the following bound:
\begin{align}
\frac{\mathbb{E} \left[ \CarbonClipper(\mathcal{I}) \right]}{\ADV(\mathcal{I})} \leq \frac{\CLIP_{j-1} + \mathbb{W}^1 ( \mbf{p}_{j-1},  \mbf{a}_{j-1} ) + \tau c(\mbf{a}_{j-1}) + L \sum_{t=j}^T c(\mbf{p}_t) + y}{\ADV_{j-1} + \tau c(\mbf{a}_{j-1}) + L \sum_{t=j}^T c(\mbf{a}_t) + y'} & \\
\leq \frac{\CLIP_{j-1} + \mathbb{W}^1 ( \mbf{p}_{j-1},  \mbf{a}_{j-1} ) + \tau c(\mbf{a}_{j-1}) + L \sum_{t=j}^T c(\mbf{p}_t)}{\ADV_{j-1} + \tau c(\mbf{a}_{j-1}) + L \sum_{t=j}^T c(\mbf{a}_t)}  &\leq (1 + \varepsilon). \label{eq:case2consistency}
\end{align}

Thus, by combining the bounds in each of the above two cases, the result follows, and we conclude that $\CarbonClipper$ is $(1+\varepsilon)$-consistent with accurate advice. \hfill \qed

\bigskip

\noindent Having proved consistency, we next prove \sref{Lemma}{lem:clipconst} to show that \CarbonClipper is $O(\log n) \gamma^{(\varepsilon)}$-robust.

\paragraph{Proof of \sref{Lemma}{lem:cliprob}.}
Let $\varepsilon \in (0, \eta^\star -1]$ be the target consistency (recalling that \CarbonClipper is $(1+\varepsilon)$-consistent), and let $\mathcal{I} \in \Omega$ denote an arbitrary valid \CSWM sequence.
To prove the robustness of \CarbonClipper, we consider two ``bad cases'' for the advice $\ADV(\mathcal{I})$, and show that in the worst-case, \CarbonClipper's competitive ratio is bounded by $O(\log n) \gamma^{(\varepsilon)}$.

\smallskip

\begin{center} \textbf{Case 1: $\ADV(\mathcal{I})$ is ``inactive''. \ } \end{center}
Consider the case where \ADV accepts nothing during the main sequence and instead satisfies the entire deadline constraint at the end of the sequence immediately before the mandatory allocation, incurring the worst possible movement \& switching cost in the process.  In the worst-case, this gives that $\ADV(\mathcal{I}) = U + D + \tau$.

{\color{blue}
Based on the consistency constraint (and using the fact that \CarbonClipper will always be ``overprovisioning'' with respect to \ADV throughout the main sequence), we can derive an upper bound on the constraint satisfaction that \CarbonClipper is ``allowed to accept'' from the robust pseudo-cost minimization.  Recall the following constraint:
}
\begin{align*}
\CLIP_{t-1} + f_t(\mbf{p}_t) + \mathbb{W}^1 ( \mbf{p}_t, \mbf{p}_{t-1} ) + \mathbb{W}^1 ( \mbf{p}_t, \mbf{a}_t ) + \tau c(\mbf{a}_t)& + (1 - z^{(t-1)} - c( \mbf{p}_t) )L\\
&\le (1+ \varepsilon) \left[ \ADV_t + \tau c(\mbf{a}_t) + (1 - A^{(t)} ) L \right].
\end{align*}

\begin{pro}\label{pro:zpcm}
Under ``inactive'' advice, $z_{\PCM}$ is an upper bound on the amount that \emph{\CarbonClipper} can accept from the pseudo-cost minimization in expectation without violating $(1+\varepsilon)$-consistency, and is defined as:
\[
z_{\PCM} = \gamma^{(\varepsilon)} \ln \left[ \frac{U - L - D - 2\tau}{ U - \nicefrac{U}{\gamma^{(\varepsilon)}} - D - 2\tau } \right]
\]
\end{pro}
\begin{proof}
Consider an arbitrary time step $t$.  If \CarbonClipper is \textit{not} allowed to make a decision that makes progress towards the constraint (i.e., it cannot accept anything more from the robust pseudo-cost minimization), we have that $c(\mbf{p}_t)$ is restricted to be $0$. Recall that $\mbf{a}_t = \delta_s$ (where $\delta_s$ is the Dirac measure at the starting \OFF state) for any time steps before the mandatory allocation, because the advice is assumed to be inactive.
By definition, since any cost functions accepted so far in expectation (i.e., in $\CLIP_{t-1}$) can be attributed to the robust pseudo-cost minimization, we have the following in the worst-case, using the same techniques used in the proof of \autoref{thm:etaCompPCM}:
\begin{align*}
\CLIP_{t-1} = \int_0^{z^{(t-1)}} \psi^{(\varepsilon)} (u) du + \tau z^{(t-1)}.
\end{align*}

Combining the above with the left-hand side of the consistency constraint, we have the following. Observe that $\mbf{p}_t$ is in an \OFF state, $\mbf{a}_t = \delta_s$, and any prior movement costs to make progress towards the constraint can be absorbed into the pseudo-cost $\psi$ since $\norm{\mbf{k}_{t} - \mbf{k}_{t-1}} \geq \mathbb{W}^1 ( \mbf{p}_t, \mbf{p}_{t-1} )$.  
{\color{blue}
Furthermore, in the worst-case, $\mathbb{W}^1 ( \mbf{p}_t, \mbf{a}_t ) = D z^{(t-1)}$ (i.e., the pseudo-cost chooses to move to a point in the metric that is a distance $D$ away from \ADV). 
}
\begin{align*}
\CLIP_{t-1} + \mathbb{W}^1 ( \mbf{p}_t, \mbf{a}_t ) + (1 - z^{(t-1)})L &= \int_0^{z^{(t-1)}} \psi^{(\varepsilon)} (u) du + \tau z^{(t-1)} + D z^{(t-1)} + (1 - z^{(t-1)})L.
\end{align*}
As stated, let $z^{(t-1)} = z_{\PCM}$.
Then by properties of the pseudo-cost,  
\begin{align*}
\CLIP_{t-1} + \mathbb{W}^1 ( \mbf{p}_t, \mbf{a}_t ) + (1 - z_{\PCM})L &= \int_0^{z_{\PCM}} \psi^{(\varepsilon)} (u) du + \tau z_{\PCM} + (1 - z_{\PCM})U + \\
& \hspace{10em} (1 - z_{\PCM})L + Dz_{\PCM} - (1 - z_{\PCM})U, \\
&= \gamma^{(\varepsilon)} \left[ \psi^{(\varepsilon)} (z_{\PCM}) - D \right] + (1 - z_{\PCM})L + Dz_{\PCM} - (1 - z_{\PCM})U, \\
&= \gamma^{(\varepsilon)} L + \left(L - U \right) \left(1 - \gamma^{(\varepsilon)} \ln \left[ \frac{U - L - D - 2\tau}{ U - \nicefrac{U}{\gamma^{(\varepsilon)}} - D - 2\tau } \right] \right) + Dz_{\PCM}, \\
&= \gamma^{(\varepsilon)} L + L - U + \left(U-L+D\right) \gamma^{(\varepsilon)} \ln \left[ \frac{U - L - D - 2\tau}{ U - \nicefrac{U}{\gamma^{(\varepsilon)}} - D - 2\tau } \right].
\end{align*}

\noindent Substituting for the definition of $\gamma^{(\varepsilon)}$, we obtain:
\begin{align*}
\CLIP_{t-1} + \mathbb{W}^1 ( \mbf{p}_t, \mbf{a}_t ) + (1 - z_{\PCM})L &= \gamma^{(\varepsilon)} L + L - U + \left(U-L+D\right) \gamma^{(\varepsilon)} \ln \left[ \frac{U - L - D - 2\tau}{ U - \nicefrac{U}{\gamma^{(\varepsilon)}} - D - 2\tau } \right], \\
&= \left[ \varepsilon L + U - \gamma^{(\varepsilon)} (U-L+D) \ln \left[ \frac{U - L - D - 2\tau}{ U - \nicefrac{U}{\gamma^{(\varepsilon)}} - D - 2\tau } \right] \right] +\\
& \hspace{6em} L - U + \left(U-L+D\right) \gamma^{(\varepsilon)} \ln \left[ \frac{U - L - D - 2\tau}{ U - \nicefrac{U}{\gamma^{(\varepsilon)}} - D - 2\tau } \right], \\
&= \varepsilon L + L = (1+\varepsilon) L.
\end{align*}

\noindent This completes the proposition, since $(1+\varepsilon) L $ is exactly the right-hand side of the consistency constraint (note that $(1+\varepsilon) \left[ \ADV_t + \tau c(\mbf{a}_t) + (1-A_t) L \right] = (1+\varepsilon) L$).
\end{proof}

If \CarbonClipper is constrained to use at most $z_{\PCM}$ of its utilization to be robust, the remaining $(1-z_{\PCM})$ utilization must be used for the mandatory allocation and/or to follow $\ADV$.  
{\color{blue}
Note that if \CarbonClipper has moved away from \ADV's point in the metric, and \ADV turns out to be ``inactive'' bad advice that incurs sub-optimal service cost late in the sequence, the consistency constraint will become non-binding and \CarbonClipper will not have to move back to follow \ADV in the metric.
}
Thus, we have the following worst-case competitive ratio for \CarbonClipper, specifically for Case 1, where we assume $\OPT(\mathcal{I}) \to \nicefrac{\psi^{(\varepsilon)} (z_{\PCM})}{O(\log n)} = \nicefrac{L}{O(\log n)}$, as in, e.g., \sref{Lemma}{lem:pcm-opt-lb}:
\begin{align}
\frac{\mathbb{E} \left[ \CarbonClipper(\mathcal{I}) \right]}{\OPT(\mathcal{I})} &\le \frac{ \int_0^{z_{\PCM}} \psi^{(\varepsilon)}(u) du + \tau z_{\PCM} + (1 - z_{\PCM})U}{\nicefrac{L}{O(\log n)}} \label{eq:case1robustness} \\
&\le \frac{ \int_0^{z_{\PCM}} \psi^{(\varepsilon)}(u) du + \tau z_{\PCM} + (1 - z_{\PCM}) D + (1 - z_{\PCM})U}{\nicefrac{L}{O(\log n)}}
\end{align}
By the definition of $\psi^{(\varepsilon)}(\cdot)$, we have the following:
\begin{align*}
\frac{\mathbb{E} \left[ \CarbonClipper(\mathcal{I}) \right]}{\OPT(\mathcal{I})} &\le \frac{ \int_0^{z_{\PCM}} \psi^{(\varepsilon)}(u) du + \tau z_{\PCM} + (1 - z_{\PCM}) D + (1 - z_{\PCM})U}{\nicefrac{L}{O(\log n)}}\\
&\le \frac{ \gamma^{(\varepsilon)} \left[ \psi^{(\varepsilon)}(z_{\PCM}) - 2D \right]}{\nicefrac{L}{O(\log n)}} \le \frac{ \gamma^{(\varepsilon)} \left[ L + 2D - 2D \right]}{\nicefrac{L}{O(\log n)}} \ \leq \ O(\log n) \gamma^{(\varepsilon)}.
\end{align*}

\begin{center} \textbf{Case 2: $\ADV(\mathcal{I})$ is ``overactive''. \ } \end{center}
We now consider the case where \ADV incurs bad service cost due to ``accepting'' cost functions which it ``should not'' (i.e. $\ADV(\mathcal{I}) \gg \OPT(\mathcal{I})$).
Let $\ADV(\mathcal{I}) = \mathbb{V} \gg \OPT(\mathcal{I})$ (i.e. the final total cost of \ADV is $\mathbb{V}$ for some $\mathbb{V} \in [L, U]$, and $\mathbb{V}$ is much greater than the optimal solution).

This is without loss of generality, since we can assume that $\mathbb{V}$ is the ``best marginal service and movement cost'' incurred by \ADV at a particular time step and the consistency ratio changes strictly in favor of \ADV.
Based on the consistency constraint, we can derive a lower bound on the amount that \CarbonClipper \textit{must} accept from \ADV in expectation to stay $(1+\varepsilon)$-consistent.
To do this, we consider the following sub-cases:

\noindent $\bullet$ \textbf{Sub-case 2.1}: Let $\mathbb{V} \geq \frac{U+D}{1+\varepsilon}$.

In this sub-case, \CarbonClipper can fully ignore the advice, because the following consistency constraint is never binding (note that $\ADV_t \geq \frac{U+D}{1+\varepsilon} A^{(t)}$):
\begin{align*}
 \CLIP_{t-1} + f_t(\mbf{p}_t) + \mathbb{W}^1 ( \mbf{p}_t, \mbf{p}_{t-1}) + \mathbb{W}^1 ( \mbf{p}_t, \mbf{a}_t )& + \tau c(\mbf{a}_t) + (1 - A^{(t)} ) L + (A^{(t)} - z^{(t-1)} - c(\mbf{p}_t))U\\
 & \hspace{5em} \le (1+ \varepsilon) \left[ \ADV_t + \tau c(\mbf{a}_t) + (1 - A^{(t)} ) L \right], \\
 (D + \tau) c(\mbf{a}_t) + (1 - A^{(t)} ) L + (A^{(t)})U  &\le (1+ \varepsilon) \left[ \mathbb{V} c(\mbf{a}_t) + \tau c(\mbf{a}_t) + (1 - A^{(t)} ) L\right],  \\
 (D + \tau) A^{(t)} + (1 - A^{(t)} ) L + U A^{(t)} &\le (1+ \varepsilon) \left[ \frac{U+D}{1+\varepsilon} A^{(t)} + \tau A^{(t)} + (1 - A^{(t)} ) L \right]
\end{align*}

\noindent $\bullet$ \textbf{Sub-case 2.2}: Let $\mathbb{V} \in ( L, \ \frac{U+D}{1+\varepsilon} )$.

In this case, in order to remain $(1 + \varepsilon)$-consistent, \CarbonClipper must follow \ADV and incur some ``bad cost'', denoted by $\mathbb{V}$.
We derive a lower bound that describes the minimum amount that \CarbonClipper must follow \ADV in order to always satisfy the consistency constraint.

\begin{pro}\label{pro:zadv}
Under ``overactive'' advice, $z_{\ADV}$ is a lower bound on the amount that \emph{\CarbonClipper} must accept from the advice in order to always satisfy the consistency constraint, and is defined as:
\[
z_\ADV \geq 1 - \frac{ \mathbb{V} \varepsilon }{U + D - \mathbb{V}}.
\]
\end{pro}
\begin{proof}
For the purposes of showing this lower bound, we assume there are no marginal service costs in the instance that would otherwise be accepted by the robust pseudo-cost minimization.

\noindent Based on the consistency constraint, we have the following:
\begin{align*}
\CLIP_{t-1} + f_t(\mbf{p}_t) + \mathbb{W}^1 ( \mbf{p}_t, \mbf{p}_{t-1}) + \mathbb{W}^1 ( \mbf{p}_t, \mbf{a}_t )& + \tau c(\mbf{a}_t) + (1 - A^{(t)} ) L + (A^{(t)} - z^{(t-1)} - c(\mbf{p}_t))U\\
& \hspace{5em} \le (1+ \varepsilon) \left[ \ADV_t + \tau c(\mbf{a}_t) + (1 - A^{(t)} ) L \right].
\end{align*}
We let $f_t(\mbf{p}_t) + \mathbb{W}^1 ( \mbf{p}_t, \mbf{p}_{t-1}) + \mathbb{W}^1 ( \mbf{p}_t, \mbf{a}_t ) + \tau c(\mbf{a}_t) \leq v c(\mbf{p}_t)$ for any $\mbf{p}_t : c(\mbf{p}_t) < c(\mbf{a}_t)$, which holds by linearity of the cost functions $f_t(\cdot)$ and a prevailing condition that $c( \mbf{p}_t ) \leq c( \mbf{a}_t )$ for the ``bad service costs'' accepted by \ADV.
{\color{blue}
Note that \CarbonClipper must ``follow'' \ADV to distant points in the metric to avoid violating consistency, and recall that $\mbf{p}_t = \mbf{a}_t$ is always in the feasible set.  }
Under this condition that $\CarbonClipper$ follows \ADV, $\mathbb{W}^1 ( \mbf{p}_t, \mbf{p}_{t-1}) + \tau c(\mbf{a}_t)$ is upper bounded by the movement cost of $\ADV$ and absorbed into $\mathbb{V}$.  
{\color{blue}
The term $\mathbb{W}^1 ( \mbf{p}_t, \mbf{a}_t ) $ is upper bounded by $D(A^{(t)} - c(\mbf{p}_t))$ by \sref{Assumption}{ass:switching-cost} of the metric.  }

\begin{align*}
\CLIP_{t-1} + \mathbb{V} c(\mbf{p}_t) + L - L A^{(t)}& + U A^{(t)} + D A^{(t)} - U z^{(t-1)} - U c(\mbf{p}_t) - D c(\mbf{p}_t) \le \\
& \hspace{9em} (1+ \varepsilon) \left[ \mathbb{V}A^{(t-1)} + \mathbb{V}c(\mbf{a}_t) + L - L A^{(t)}  \right],\\
\mathbb{V} c(\mbf{p}_t) - D c(\mbf{p}_t) - U c(\mbf{p}_t) &\le (1+ \varepsilon) \left[ \mathbb{V}A^{(t-1)} + \mathbb{V}c(\mbf{a}_t) + L - L A^{(t)}  \right] - \\
& \hspace{9em} \CLIP_{t-1} - L + L A^{(t)} - U A^{(t)} - DA^{(t)} + U z^{(t-1)}, \\
\mathbb{V} c(\mbf{p}_t) - D c(\mbf{p}_t) - U c(\mbf{p}_t) &\le \mathbb{V} A^{(t)} - D A^{(t)} - U A^{(t)} - \CLIP_{t-1} + U z^{(t-1)} +\\
& \hspace{9em} \varepsilon \left[ \mathbb{V} A^{(t-1)} + \mathbb{V} c(\mbf{a}_t) + L - L A^{(t)}  \right], \\
c(\mbf{p}_t) &\ge \frac{\mathbb{V} A^{(t)} - D A^{(t)} - U A^{(t)} - \CLIP_{t-1} + U z^{(t-1)} + \varepsilon \left[ \mathbb{V} A^{(t)} + L - L A^{(t)}  \right]}{\mathbb{V}-D-U}.
\end{align*}

\noindent In the event that $A^{(t-1)} = 0$ (i.e., nothing has been accepted so far by either \ADV or \CarbonClipper), we have:
\begin{align*}
c(\mbf{p}_t) &\ge \frac{\mathbb{V} c(\mbf{a}_t) - Dc(\mbf{a}_t) - U c(\mbf{a}_t) + \varepsilon \left[ \mathbb{V} c(\mbf{a}_t) + L - L c(\mbf{a}_t)  \right]}{\mathbb{V}-D-U}, \\
c(\mbf{p}_t) &\ge c(\mbf{a}_t) - \frac{\varepsilon \left[ \mathbb{V} c(\mbf{a}_t) + L - L c(\mbf{a}_t)  \right]}{U+D-\mathbb{V}}.
\end{align*}

Through a recursive definition, we can show that for any $A^{(t)}$, given that \CarbonClipper has satisfied $z^{(t-1)}$ of the deadline constraint by following \ADV so far, it must set $\mbf{p}_t$ such that:
\begin{align*}
z^{(t)} &\ge z^{(t-1)} + c(\mbf{a}_t) - \frac{\varepsilon \left[  \mathbb{V} c(\mbf{a}_t) + L - L c(\mbf{a}_t)  \right]}{U+D-\mathbb{V}}.
\end{align*}

{\color{blue}
\noindent Continuing the assumption that $\mathbb{V}$ is constant, if \CarbonClipper has accepted $z^{(t-1)}$ thus far, we have the following if we assume that all of the constraint satisfaction up to this point happened in a single previous time step $m$:
}
\begin{align*}
c(\mbf{p}_t) &\ge A^{(t)} - \frac{U c(\mbf{p}_m) + D c(\mbf{p}_m) - \CLIP_{t-1} + \varepsilon \left[ \mathbb{V} A^{(t)} + L - L A^{(t)}  \right]}{U+D-\mathbb{V}}, \\
c(\mbf{p}_t) &\ge c(\mbf{a}_t) + c(\mbf{a}_m) - c(\mbf{p}_m) - \frac{\varepsilon \left[ \mathbb{V} (c(\mbf{a}_t) + c(\mbf{a}_m)) + L - L (c(\mbf{a}_t) + c(\mbf{a}_m))  \right]}{U+D-\mathbb{V}}, \\
c(\mbf{p}_t) + c(\mbf{p}_m) &\ge c(\mbf{a}_t) + c(\mbf{a}_m) - \frac{\varepsilon \left[ \mathbb{V} (c(\mbf{a}_t) + c(\mbf{a}_m)) + L - L (c(\mbf{a}_t) + c(\mbf{a}_m))  \right]}{U+D-\mathbb{V}}, \\
z^{(t)} &\ge A^{(t)} - \frac{\varepsilon \left[ \mathbb{V} A^{(t)} + L - L A^{(t)}  \right]}{U+D-\mathbb{V}}.
\end{align*}

\noindent This gives intuition into the desired $z_\ADV$ bound.  The above motivates that the \textit{aggregate} expected constraint satisfaction by \CarbonClipper at any given time step $t$ must satisfy a lower bound.  Consider that the worst case for Sub-case 2.2 occurs when all of the $v$ prices accepted by \ADV arrive first, before any prices that would be considered by the pseudo-cost minimization.  Then let $A^{(t)} = 1$ for some arbitrary time step $t$, and we have the stated lower bound on $z_\ADV$.
\end{proof}

If \CarbonClipper is forced to use $z_\ADV$ of its utilization to be $(1+\varepsilon)$ consistent against \ADV, that leaves at most $(1-z_\ADV)$ utilization for robustness.  We define $z' = \min ( 1 - z_\ADV, z_{\PCM})$ and consider the following two cases.  

\noindent $\bullet$ \textbf{Sub-case 2.2.1}: if $z' = z_{\PCM}$, the worst-case competitive ratio is bounded by the following.  Note that if $z' = z_{\PCM}$, the amount of utilization that \CarbonClipper can use to ``be robust'' is exactly the same as in \textbf{Case 1}, and we again have that $\OPT(\mathcal{I}) \to \nicefrac{\psi^{(\varepsilon)} (z_{\PCM})}{O(\log n)} = \nicefrac{L}{O(\log n)}$:
\begin{align}
\frac{\mathbb{E} \left[ \CarbonClipper(\mathcal{I}) \right]}{\OPT(\mathcal{I})} &\leq \frac{\int_0^{z_{\PCM}} \psi^{(\varepsilon)}(u) du + \tau z_{\PCM} + (1 -z_{\ADV} - z_{\PCM})U + z_{\ADV} \mathbb{V}}{\nicefrac{L}{O(\log n)}}, \\
&\leq \frac{\int_0^{z_{\PCM}} \psi^{(\varepsilon)}(u) du + \tau z_{\PCM} + (1 - z_{\PCM})D + (1 - z_{\PCM})U}{\nicefrac{L}{O(\log n)}}, \label{eq:case2robustness} \\
&\leq O(\log n) \gamma^{(\varepsilon)}.
\end{align}

\noindent $\bullet$ \textbf{Sub-case 2.2.2}: if $z' = 1 - z_{\ADV}$, the worst-case competitive ratio is bounded by the following.  Note that \CarbonClipper \textit{cannot} use $z_{\PCM}$ of its utilization for robustness, so the following bound assumes that the ``robust service costs'' accepted by \CarbonClipper are bounded by the \textit{worst $(1 - z_{\ADV})$ fraction} of the pseudo-cost function $\psi^{(\varepsilon)}$ (note that $\psi^{(\varepsilon)}$ is non-increasing on $z \in [0,1]$):
\begin{align*}
\frac{\mathbb{E} \left[ \CarbonClipper(\mathcal{I}) \right]}{\OPT(\mathcal{I})} \leq \frac{\int_0^{1 - z_{\ADV}} \psi^{(\varepsilon)}(u) du + \tau (1 - z_{\ADV}) + z_{\ADV} \mathbb{V} }{\nicefrac{L}{O(\log n)}}.
\end{align*}

\noindent Note that if $z' = 1 - z_\ADV$, we know that $1 - z_\ADV < z_{\PCM}$, which further gives the following by definition of $z_\ADV$:
\begin{align*}
    1-z_{\PCM} &< 1 - \frac{\mathbb{V} \varepsilon}{U+D-\mathbb{V}}, \\
    \mathbb{V} \varepsilon &< (U+D-\mathbb{V}) z_{\PCM}, \\
    \mathbb{V} &< \frac{U+D}{(1+ \frac{\varepsilon}{z_{\PCM}})}.
\end{align*}
\noindent By plugging $\mathbb{V}$ back into the definition of $z_\ADV$, we have that $z_\ADV \mathbb{V} \leq \left( \frac{(1-z_{\PCM})(U+D)}{1 + \frac{\varepsilon}{z_{\PCM}}} \right)$, giving the following:
\begin{align}
\frac{\mathbb{E} \left[ \CarbonClipper(\mathcal{I}) \right]}{\OPT(\mathcal{I})} &\leq \frac{\int_0^{1 - z_{\ADV}} \psi^{(\varepsilon)}(u) du + \tau (1 - z_{\ADV}) + \left( \frac{(1-z_{\PCM})(U+D)}{1 + \frac{\varepsilon}{z_{\PCM}}} \right) }{\nicefrac{L}{O(\log n)}}, \\
&\leq \frac{\int_0^{z_{\PCM}} \psi^{(\varepsilon)}(u) du + \tau z_{\PCM} + (1 - z_{\PCM})D + (1 - z_{\PCM})U}{\nicefrac{L}{O(\log n)}},\\
&\leq  O(\log n) \gamma^{(\varepsilon)}.
\end{align}

Thus, by combining the bounds in each of the above two cases, the result follows, and we conclude that $\CarbonClipper$ is $O(\log n) \gamma^{(\varepsilon)}$-robust. \hfill \qed

Having proven \sref{Lemma}{lem:clipconst} (consistency) and \sref{Lemma}{lem:cliprob} (robustness), the statement of \autoref{thm:constrobCLIP} follows: \CarbonClipper is $(1+\varepsilon)$-consistent and $O(\log n) \gamma^{(\varepsilon)}$-robust given any advice for \CSWM. 
\end{proof}

\subsection{Proof of \autoref{thm:optimalconstrobCSWM}} \label{apx:optimalconstrobCSWM}

In this section, we prove Theorem~\ref{thm:optimalconstrobCSWM}, which states that $\gamma^{(\varepsilon)}$ (as defined in \eqref{eq:gamma}) is the optimal robustness for any $(1 + \varepsilon)$-consistent learning-augmented \CSWM algorithm.
\begin{proof}
To show this result, we build off the same $y$-adversaries for \CSWM defined in \sref{Definition}{dfn:yadversary-caslb}, where $y \in [L,U]$.  For the purposes of showing consistency, we define a slightly tweaked adversary $\mathcal{A}'_y$:
\begin{dfn}[$\mathcal{A}'_y$ adversary for learning-augmented \CSWM] \label{dfn:yadversary-constrob}
Recall the $\mathcal{A}_y$ adversary defined in \sref{Definition}{dfn:yadversary-caslb}.
During \textbf{Stage 1} of the adversary's sequence, $\mathcal{A}_y$ and $\mathcal{A}'_y$ are identical.
{\color{blue}
In \textbf{Stage 2}, $\mathcal{A}'_y$ presents $\text{Up}(x)$ at the starting point's \ON state $\ON^{(s)}$ once, followed by $\text{Down}^{m_y} (x) = y x$ at $\ON^{(s)}$.
All other $\ON$ states are considered inactive in this stage, so they only receive $\text{Up}(x)$.  
In the $\mu-1$ time steps after $\text{Down}^{m_y} (x) = y x$ is presented, the adversary presents $\text{Down}^{m_y} (x) = y x$ at the starting point in the metric (allowing \emph{\OPT} and \emph{\ADV} to reduce their switching cost).  Finally, the adversary presents $\text{Up}(x)$ everywhere for the final $\mu$ time steps, and \textbf{Stage 2} ends.
}
\end{dfn}

\noindent As in the proof of \autoref{thm:lowerboundCSWM}, for adversary $\mathcal{A}'_y$, the optimal offline objective is $\OPT(\mathcal{A}'_y) \to y$.
Against these adversaries, we consider two types of advice -- the first is \textit{bad} advice, which stays with their full allocation at the starting $\OFF$ state (i.e., $\mbf{a}_t = \delta_s$) for all time steps $t < j$ before the mandatory allocation, incurring a final cost of $U + 2\tau$.  

{\color{blue}
On the other hand, \textit{good} advice sets $\mbf{a}_t = \delta_s$ for all time steps up to the first time step when $y$ is revealed at the starting point in the metric, after which it sets $a_t^{\ON^{(s)}} = \nicefrac{1}{\mu}$ to achieve final cost $\ADV(\mathcal{A}'_y) = \OPT(\mathcal{A}'_y) = y + \nicefrac{2\tau}{\mu}$.
}

We let $\left( s(y) + t(y) \right)$ denote a \textit{robust constraint satisfaction function} $[L,U] \rightarrow [0,1]$, that fully quantifies the actions of a learning-augmented algorithm $\ALG$ playing against adaptive adversary $\mathcal{A}'_y$, where $\left( s(y) + t(y) \right)$ gives the progress towards the deadline constraint under the instance $\mathcal{A}'_y$ before (either) the mandatory allocation or the black-box advice sets $a_t^{\ON^{(s)}} > 0$.  Since the conversion is unidirectional (irrevocable), we must have that $s(y - \sigma) + t(y - \sigma) \geq \left( s(y) + t(y) \right)$, i.e. $\left( s(y) + t(y) \right)$ is non-increasing in $[L, U]$.  

As in the proof of \autoref{thm:lowerboundCSWM}, the adaptive nature of $\mathcal{A}'_y$ forces any algorithm to incur a movement and switching cost proportional to $\left( s(y) + t(y) \right)$ during the robust phase, specifically denoted by $\left(D + 2\tau \right) \left( s(y) + t(y) \right)$. 
Recall that by the proof of \autoref{thm:lowerboundCSWM}, for any $\gamma$-competitive online algorithm $\ALG$, we have the following condition on $\left( s(y) + t(y) \right)$:
\begin{align}
\left( s(y) + t(y) \right) & \geq \gamma \ln \left( U - y - D - 2\tau \right) - \gamma \ln \left( U - \mbf{l} - D - 2\tau \right), \quad \forall y \in [L, U]. \label{eq:rob-bound-gron-min}
\end{align}

\noindent Furthermore, we have that the expected cost of $\ALG$ on adversary $\mathcal{A}'_y$ is given by:
\begin{align}
    \mathbb{E} \left[ \ALG(\mathcal{A}'_y) \right] &= s(\mbf{l}) \mbf{l} - \int^y_{\mbf{l}} u d s(u) + D s(y) + t(\mbf{l}) \mbf{l} - \int^y_{\mbf{l}} u d t(u) + Dt(y) +\\
    &\hspace{14em} (1 - s(y) - t(y))U + 2\tau \left[ s(y) + t(y) \right] \label{eq:lb-costconstrob}
\end{align}

\noindent To simultaneously be $\alpha$-consistent when the advice \textit{is} correct, $\ALG$ must satisfy $\mathbb{E} \left[ \ALG(\mathcal{A}'_L) \right] \leq \alpha \OPT(\mathcal{A}'_L) = \alpha L$.  If the advice is correct, $\ALG$ must pay an additional factor of $D$ to move back and follow \ADV in the worst case -- but can satisfy the rest of the deadline constraint at the best cost functions $L$.  It must also still pay for switching incurred by the robust algorithm (recall that $\OPT$ pays no switching cost).  Using integral by parts, we have:
\begin{align}
\int^L_{\mbf{l}} s(u) + t(u) du + \left[ 2D + 2\tau \right] \left( s(L) + t(L) \right) +  (1 - s(y) - t(y))L + L \left( s(y) + t(y) \right) &\leq \alpha L,\\
\int^L_{\mbf{l}} s(u) + t(u) du + \left[ 2D + 2\tau \right] \left( s(L) + t(L) \right) &\leq \alpha L - L. \label{eq:const-bound-gron-min}
\end{align}
By combining equations \eqref{eq:rob-bound-gron-min} and \eqref{eq:const-bound-gron-min}, and substituting $\mbf{l} = \nicefrac{U}{\gamma}$, the robust constraint satisfaction function $\left( s(y) + t(y) \right)$ of any $\gamma$-robust and $\alpha$-consistent online algorithm must satisfy:
\begin{align}
\gamma \int_{\mbf{l}}^L \ln \left( \frac{U - u - D - 2\tau}{ U - \nicefrac{U}{\gamma} - D - 2\tau } \right)  du + \left[ 2D + 2\tau \right] \left[ \gamma \ln \left( \frac{U - L - D - 2\tau}{ U - \nicefrac{U}{\gamma} - D - 2\tau } \right)  \right] &\leq \alpha L - L.
\end{align}
When all inequalities are binding, this equivalently gives that 
\begin{align}
\alpha \geq \gamma +1-\frac{U}{L} + \frac{\gamma(U-L+D)}{L}  \ln \left( \frac{U - L - D - 2\tau}{ U - \nicefrac{U}{\gamma} - D - 2\tau } \right). \label{eq:eta-gamma}
\end{align}
We define $\alpha$ such that $\alpha \coloneqq (1 + \varepsilon)$.  By substituting for $\alpha$ into \eqref{eq:eta-gamma}, we recover the definition of $\gamma^{(\varepsilon)}$ as given by \eqref{eq:gamma}, which subsequently completes the proof.
Thus, we conclude that any $(1+\varepsilon)$-consistent algorithm for \CASLB is at least $\gamma^{(\varepsilon)}$-robust.
\end{proof}

\section{Proofs for \autoref{sec:timevarying} (Generalization to Time-varying Metrics)} \label{apx:timevarying}
\subsection{Proof of \sref{Corollary}{cor:timevaryPCM}} \label{apx:timevaryPCM}

{\color{blue}
In this section, we prove \sref{Corollary}{cor:timevaryPCM}, which states that \PCM is $O(\log n) \eta$-competitive for \CSWMt, the variant of \CSWM where distances in the metric $(X, d)$ are allowed to be time-varying (see \autoref{sec:timevarying}).
}

We recall a few mild assumptions on the \CSWMt problem that directly imply the result.
Let $d_t(\cdot, \ \cdot)$ denote the distance between points in $X$ at time $t \in [T]$.
{\color{blue}
We redefine $D$ as $D = \sup_{t \in [T]} \left( \max_{u, v \in X : u \not = v} \frac{d_t( u, \ v )}{\min \{ c^{(u)}, c^{(v)} \} } \right)$, i.e., it is an upper bound on distance between any two points in the metric at any time over the horizon $T$.  
Recall that the temporal switching cost between \ON and \OFF states at a single point is defined as \textit{non-time-varying}, so $\Vert \cdot \Vert_{\ell_1 (\mathbf{\beta})}$ gives the temporal switching cost for all $t \in [T]$.
}

We also assume that the tree embedding-based vector space $(K, \norm{\cdot})$ defined by \sref{Definition}{dfn:tree} is appropriately reconstructed at each step, and that \PCM has knowledge of the current distances (i.e., $\norm{\cdot}$ accurately reflects $d_t(\cdot, \ \cdot)$ at time $t$).

Under these assumptions, every step in the proof of \autoref{thm:etaCompPCM} holds.  In particular, note that in \sref{Lemma}{lem:pcm-opt-lb}, the only fact about the distance function that is used is the fact that the distance between two \ON states is upper bounded by $D$, in \eqref{eq:lbminprob}, and that the vector space $(K, \norm{\cdot})$ has expected $O(\log n)$ distortion with respect to the underlying metric, which follows by definition.   
{\color{blue}
In \sref{Lemma}{lem:pcm-alg-ub}, most of \PCM's movement cost is absorbed into the integral over the pseudo-cost function $\psi$, and the only other fact about the distance that is used is that the distance between $\ON$ and $\OFF$ states at a single point $u \in X$ is fixed and bounded by $\tau c^{(u)}$, which follows by definition.
}
Thus, we conclude that \PCM is $O(\log n) \eta$-competitive for \CSWMt. \hfill \qed

\bigskip

\subsection{Proof of \sref{Corollary}{cor:timevaryCLIP}} \label{apx:timevaryCLIP}
In this section, we prove \sref{Corollary}{cor:timevaryCLIP}, which states that a minor change to the consistency constraint enables \CarbonClipper to be $(1+\varepsilon)$-consistent and $O(\log n) \gamma^{(\varepsilon)}$-robust for \CSWMt, given any $\varepsilon \in (0, \eta - 1]$.

We start by recalling assumptions on the \CSWMt problem that inform the result.  
{\color{blue}
Recall that we redefine $D$ as $D = \sup_{t \in [T]} \left( \max_{u, v \in X : u \not = v} \frac{d_t( u, \ v )}{\min \{ c^{(u)}, c^{(v)} \} } \right)$, i.e., it is an upper bound on distance between any two points at any time over the horizon $T$, and the temporal switching cost between \ON and \OFF states at a single point is defined as \textit{non-time-varying}, so $\Vert \cdot \Vert_{\ell_1 (\mathbf{\beta})}$ gives the temporal switching cost for all $t \in [T]$.
}

We also assume that the tree embedding-based vector space $(K, \norm{\cdot})$ defined by \sref{Definition}{dfn:tree} is appropriately reconstructed at each step, and that \CarbonClipper has knowledge of the current distances (i.e., $\norm{\cdot}$ and $\mathbb{W}^1(\cdot, \ \cdot)$ accurately reflect $d_t(\cdot, \ \cdot)$ at time $t$).

For the consistency constraint, we define a modified Wasserstein-1 distance function $\overline{\mathbb{W}}^1$ that will be used in the consistency constraint to hedge against the time-varying properties of the metric.  
{\color{blue}
This distance computes the optimal transport between two distributions on $\Delta_{\mathcal{S}}$ while assuming that the underlying distances are given by $\overline{d}( \cdot, \ \cdot)$, which is itself defined such that $\overline{d}( u, \ v ) = D \min \{ c^{(u)}, c^{(v)} \} : u, v \in X : u \not = v$.
\begin{equation}
    \overline{\mathbb{W}}^1 (\mbf{p}, \mbf{p}') \coloneqq \min_{\pi \in \Pi(\mbf{p}, \mbf{p}')} \mathbb{E} \left[ \overline{d}( \mbf{x}, \mbf{x}' ) \right],
\end{equation}
where $( \mbf{x}, \mbf{x}' ) \thicksim \pi_t$ and $\Pi(\mbf{p}, \mbf{p}')$ is the set of distributions over $X^2$ with marginals $\mbf{p}$ and $\mbf{p}'$.

Intuitively, the purpose of $\overline{\mathbb{W}}^1$ is to leverage the $D$ upper bound between points in the metric to give a ``worst-case optimal transport'' distance between distributions, assuming that the time-varying distances increase in future time steps. 
}

To this end, within the definition of the consistency constraint \eqref{eq:const-constraint}, \CarbonClipper for \CSWMt replaces the term $\mathbb{W}^1( \mbf{p}, \mbf{a}_t )$ with $\overline{\mathbb{W}}^1( \mbf{p}, \mbf{a}_t )$ -- this term hedges against the case where \CarbonClipper must move to follow \ADV in a future time step, and in the time-varying distances case, we charge \CarbonClipper an extra amount to further hedge against the case where the underlying distances between \CarbonClipper and \ADV grow in future time steps.

\noindent Paralleling the proof of \autoref{thm:constrobCLIP}, we consider consistency and robustness independently.

\noindent\textbf{Consistency. \ }
For consistency, according to the proof of \sref{Lemma}{lem:clipconst}, we show that resolving the mandatory allocation remains feasible in the time-varying case.  First, note that there is always a feasible $\mbf{p}_t$ that satisfies the consistency constraint, since even in the time-varying case, setting $\mbf{k}_t = \Phi \mbf{a}_t$ is always within the feasible set -- this follows by observing that at a given time $t$, if \CarbonClipper has moved away from \ADV, it has already ``prepaid'' a worst-case movement cost of $\overline{\mathbb{W}}^1( \mbf{p}_m, \mbf{a}_m )$ (for some previous time step $m$) in order to move back and follow \ADV.  
{\color{blue}
The remainder of the proof of \sref{Lemma}{lem:clipconst} only uses the fact that at the beginning of the mandatory allocation (at time $j \in [T]$), $\overline{\mathbb{W}}^1( \mbf{p}_{j-1}, \mbf{a}_{j-1} )$ is an upper bound on the movement cost paid by \CarbonClipper (if necessary) while migrating to \ADV's points in the metric to take advantage of the same service cost functions.  
}
By definition of $\overline{\mathbb{W}}^1$, this follows for any time-varying distance $d_t( \cdot, \ \cdot)$, and the remaining steps in the proof hold. \hfill \qed

\noindent\textbf{Robustness. \ }
For robustness, following the proof of \sref{Lemma}{lem:cliprob}, we show that a certain amount of \CarbonClipper's utilization can be ``set aside'' for robustness.  First, note that in Case 1 (i.e., ``inactive'' advice), the proof of \sref{Proposition}{pro:zpcm} only requires that $\mathbb{W}^1(\mbf{p}_t, \mbf{a}_t)$ is bounded by $D$, which follows immediately by the definition of $\overline{\mathbb{W}}^1$ -- the remaining steps for Case 1 follow.  In Case 2 (i.e., ``overactive'' advice), note that the proof of Sub-case 2.1 similarly only requires that $\mathbb{W}^1(\mbf{p}_t, \mbf{a}_t)$ is bounded by $D$, which follows because $\mathbb{W}^1(\mbf{p}_t, \mbf{a}_t) \leq \overline{\mathbb{W}}^1(\mbf{p}_t, \mbf{a}_t) \leq D \max \{ c(\mbf{p}_t), c(\mbf{a}_t) \}$.  Likewise, the remaining steps for Case 2 follow.  
\hfill \qed

\smallskip

Thus, we conclude that \CarbonClipper is $(1+\varepsilon)$-consistent and $O(\log n) \gamma^{(\varepsilon)}$-robust for \CSWMt, given any $\varepsilon \in (0, \eta - 1]$.

\end{document}